\documentclass[format=acmsmall, review=false, screen=false, timestamp=false]{acmart}
\usepackage[textsize=scriptsize,backgroundcolor=yellow!40]{todonotes}
\usepackage{etex}
\usepackage{multirow}
\usepackage{amsmath}
\usepackage{caption}
\usepackage{url}
\usepackage{titleref}
\usepackage{subfig}
\usepackage{lscape}
\usepackage{array}
\usepackage{tabularx}
\usepackage{tabulary}
\usepackage{enumerate,letltxmacro}
\usepackage{enumitem}
\usepackage{longtable}
\usepackage{wrapfig}
\usepackage{setspace}
\usepackage{float}
\usepackage{multicol}
\usepackage{booktabs}
\usepackage{microtype}
\usepackage[figuresright]{rotating}
\usepackage[ansinew]{inputenc} %Windows
\usepackage{cancel}
\usepackage{paralist}
\usepackage{savesym}
\usepackage{textcomp}
\savesymbol{vec}
\usepackage{MnSymbol}
\usepackage[algoruled,lined,linesnumbered,commentsnumbered]{algorithm2e}
\SetAlFnt{\small}
\SetAlCapFnt{\small}
\SetAlCapNameFnt{\small}
\SetAlCapHSkip{0pt}
\IncMargin{-\parindent}
\usepackage{lineno,blindtext}
\usepackage{scalerel}
\usepackage{everypage}
\usepackage{spverbatim}
\let\oldFootnote\footnote
\newcommand\nextToken\relax

\renewcommand\footnote[1]{%
    \oldFootnote{#1}\futurelet\nextToken\isFootnote}

\newcommand\isFootnote{%
    \ifx\footnote\nextToken\textsuperscript{,}\fi}

\graphicspath{{fig/}}

\input{artem_tosem.sty}

\usepackage{tikz,pgf,xcolor}
\usetikzlibrary{arrows,automata,trees,plotmarks,shadows,shapes}
\usepackage{pgfplots}
\usetikzlibrary{pgfplots.groupplots}
\definecolor{mybluecolor}{RGB}{50,106,218}
\definecolor{myredcolor}{RGB}{176,53,53}
\definecolor{mygreencolor}{RGB}{93,172,0}
\definecolor{myyellowcolor}{RGB}{255,163,34}
\definecolor{mypurplecolor}{RGB}{86,35,132}
\definecolor{mytealcolor}{RGB}{30,161,165}

\newcommand{\funcCall}[2]{{{\mathit{#1}}\left({#2}\right)}}

\newcommand{\plotdata}[8]{\begin{tikzpicture}
\begin{semilogxaxis}[
height=4.5cm,
width=.48\linewidth,
axis x line=bottom,
axis y line=left,
xlabel=log size,
ylabel={#1},
xlabel near ticks,
ylabel near ticks,
y tick label style={/pgf/number format/.cd, fixed, fixed zerofill, precision=2, 
/tikz/.cd, font=\footnotesize},
x tick label style={font=\footnotesize},
label style={font=\footnotesize},
legend style={font=\footnotesize, legend pos=#8},
grid=both
]
\pgfplotstableread[col sep=comma]{fig/plots/data.csv}{\table}
\addplot[very thick,mark=*,color=#2] table [x index=0,y index=#6,col sep=comma] 
{\table};
\addlegendentryexpanded{#4}
\addplot[very thick,mark=diamond*,color=#3] table [x index=0,y index=#7,col 
sep=comma] {\table};
\addlegendentryexpanded{#5}
\end{semilogxaxis}
\end{tikzpicture}
}

\newcommand{\plottime}[9]{\begin{tikzpicture}
\begin{loglogaxis}[ % semilogxaxis
height=4.5cm,
width=.48\linewidth,
axis x line=bottom,
axis y line=left,
xlabel=log size,
ylabel=$\mathit{time\,(ms)}$,
xlabel near ticks,
ylabel near ticks,
y tick label style={font=\footnotesize},
x tick label style={font=\footnotesize},
label style={font=\footnotesize},
legend style={font=\footnotesize,at={(0.5,-0.2)},anchor=north,legend columns=-1},
grid=both
]
\pgfplotstableread[col sep=comma]{fig/plots/time.csv}{\table}
\addplot[very thick,mark=*,color=#1] table [x index=0,y index=#7,col sep=comma] {\table};
\addlegendentryexpanded{#4}
\addplot[very thick,mark=diamond*,color=#2] table [x index=0,y index=#8,col sep=comma] {\table};
\addlegendentryexpanded{#5}
\addplot[very thick,mark=square*,color=#3] table [x index=0,y index=#9,col sep=comma] {\table};
\addlegendentryexpanded{#6}
\end{loglogaxis} % semilogxaxis
\end{tikzpicture}
}

\newcommand{\MONTH}{%
  \ifcase\the\month
  \or January% 1
  \or February% 2
  \or March% 3
  \or April% 4
  \or May% 5
  \or June% 6
  \or July% 7
  \or August% 8
  \or September% 9
  \or October% 10
  \or November% 11
  \or December% 12
  \fi}

\renewcommand{\textrightarrow}{$\rightarrow$}

\received{March 2019}
\received[Revised]{January 2020}
\received[Accepted]{\MONTH~\the\year}

\acmJournal{TOSEM}
\acmYear{2020} %\acmVolume{1} \acmNumber{1} \acmArticle{3} \acmMonth{3} 
\acmDOI{10.1145/3387909}
\copyrightyear{\the\year}
\begin{document}

\title[]{Monotone Precision and Recall Measures for Comparing Executions and Specifications of Dynamic Systems}

\author{Artem Polyvyanyy}
\email{artem.polyvyanyy@unimelb.edu.au}
\orcid{0000-0002-7672-1643}
\affiliation{%
  \institution{The University of Melbourne}
  \streetaddress{Level 8, Doug McDonell Building}
  \city{Parkville}
  \state{VIC}
  \postcode{3010}
  \country{Australia}
}

\author{Andreas Solti}
\email{solti@ai.wu.ac.at}
\orcid{0000-0002-0537-6598}
\affiliation{\institution{Vienna University of Economics and Business}\country{Austria}}

\author{Matthias Weidlich}
\email{matthias.weidlich@hu-berlin.de}
\orcid{0000-0003-3325-7227}
\affiliation{\institution{Humboldt University of Berlin}\country{Germany}}

\author{Claudio Di Ciccio}
\email{diciccio@di.uniroma1.it}
\orcid{0000-0001-5570-0475}
\affiliation{\institution{Sapienza University of Rome}\country{Italy}}

\author{Jan Mendling}
\email{mendling@ai.wu.ac.at}
\orcid{0000-0002-7260-524X}
\affiliation{\institution{Vienna University of Economics and Business}\country{Austria}}

\renewcommand{\shortauthors}{A. Polyvyanyy, A. Solti, M.Weidlich, C. Di Ciccio, and J. Mendling}

\begin{abstract}
The behavioural comparison of systems is an important concern of software engineering research. For example, the areas of \emph{specification discovery} and \emph{specification mining} are concerned with measuring the consistency between a collection of execution traces and a program specification. This problem is also tackled in \emph{process mining} with the help of measures that describe the quality of a process specification automatically discovered from execution logs. Though various measures have been proposed, it was recently demonstrated that they neither fulfil essential properties, such as \emph{monotonicity}, nor can they handle infinite behaviour. 
In this paper, we address this research problem by introducing a new framework for the definition of behavioural quotients. We proof that corresponding quotients guarantee desired properties that existing measures have failed to support. We demonstrate the application of the quotients for capturing precision and recall measures between a collection of recorded executions and a system specification. We use a prototypical implementation of these measures to contrast their monotonic assessment with measures that have been defined in prior research. 
\end{abstract}




\maketitle

%%%%%%%%%%%%%%%%%%%%%%%%%
%%%%%%%%%%%%%%%%%%%%%%%%%%%%%%%%%%%%%%%%%%%%%%%%%%%%%%%%%%%%%%%%%%%%%%%%%%%%%%%
\section{Introduction}
\label{sec:introduction}
%%%%%%%%%%%%%%%%%%%%%%%%%%%%%%%%%%%%%%%%%%%%%%%%%%%%%%%%%%%%%%%%%%%%%%%%%%%%%%%

The analysis of dynamic systems is a focus of software engineering research~\cite{DBLP:journals/scp/Harel87,DBLP:journals/jss/Vogel-HeuserFST15}, and other related areas, for example business process management~\cite{Weske2012,Dumas.etal/2018:FundamentalsofBPM}, information systems~\cite{basu2002research,breuker2016comprehensible},
social science~\cite{abbott1990measuring,cornwell2015social}, and management science~\cite{pentland2003conceptualizing}.
Software engineering research is primarily concerned with the analysis of behaviours captured in software systems, program specifications, and execution traces.
This analysis often takes the form of behaviour comparison, with use cases
ranging from
specification discovery~\cite{Cook1998,Reiss2001,Mariani2005,Lo2007} and specification mining~\cite{Ammons2002,Lo2006a,Pradel.etal/ICSM2010:FrameworkForEvaluationOfSpecificationMinersFSMs,Santhiar2014}, through
conformance checking between requirements and specifications~\cite{Ali2013},
software evolution~\cite{DAmbros2008},
software test coverage~\cite{Berner2007,Tuya2016},
and black-box software 
testing~\cite{Walkinshaw.etal/FM2009,Weyuker/ATPLS1983},
 to
measurements of accuracy of the reverse-engineered specifications~\cite{Lo2006,Walkinshaw2013}.
For example, specification discovery and specification mining study ways to infer software specifications from program executions.
The quality of such inference techniques is often defined in terms of
measurements of discrepancies between the execution traces used as input and
the resulting program specifications.
\emph{Process mining}~\cite{Aalst16} integrates these perspectives by comparing the behaviour of a system as specified with the behaviour recorded during execution and has applications in
computationally-intensive theory development~\cite{berente2018data}.

A key challenge in the analysis of dynamic systems is the definition of
meaningful measures that express the degree to which \emph{different
system behaviours are in line with each other}.
Technically, such comparisons are formulated
in a \emph{relative} manner, defining a
\emph{quotient} of some aspect of one behaviour
over the same aspect of another behaviour.
For instance, the quotients of the behaviours of a system at different points in
time reveal how the system has changed. In process mining, in turn, the
quotient of the behaviour of a system as
recorded in a log over the behaviour as specified can be
used to analyse the trustworthiness of the latter. Yet, defining such quotients
is challenging: A recent
commentary on measures in process mining identifies a set of intuitive
properties and shows that none of the available measures
fulfils them~\cite{TaxLSFA17}.

We approach the above problem based on the notion of a \emph{formal language}.
This is a suitable starting point because the sequential (state-based)
behaviour of a dynamic system, e.g., a software system or information system,
can be modelled as a state machine or an
automaton~\cite{Cheng1993,Boerger2005}.
An action represents an atomic unit of work, which, depending on the
type of system, may be a program instruction, a Web service call, or a manual activity executed
by a human agent. The behaviour of a system, therefore, can
be represented by a \emph{language} that defines a set of words over its
actions. Then, each word is one possible execution (also known as a run,
trace, sequence, or process) of the system.
Alternatively, the comparison of the behaviours of dynamic systems was tackled in the literature using model structure~\cite{Walkinshaw2013} or abstract representations of the behaviours~\cite{Weidlich2011}.

Behavioural comparison based on quotients of languages faces two major challenges.
First and foremost, quotients have to satisfy essential properties in order to facilitate a reasonable interpretation. One
such property is \emph{monotonicity}: When increasing the amount of behaviour in the
numerator of a quotient while leaving the amount of behaviour in the denominator unchanged, the quotient
shall increase as well. Existing quotients as proposed, e.g., in the field of
process mining~\cite{Aalst16} to compare recorded and specified behaviour,
do not satisfy this well-motivated property~\cite{TaxLSFA17,Aalst18a}.
The second challenge relates to the definition of quotients in the presence of systems that describe infinite behaviours, i.e., the behaviours that consist of infinitely
many words.
In that case, quotients defined over standard aspects of languages, such as their cardinality, are not meaningful for
behavioural comparison.
In process mining, this issue has been avoided by using behavioural abstractions that capture a language by means of pairwise relations over its actions~\cite{Weidlich2011}.
Yet, such an abstraction does not capture the complete language semantics of a system~\cite{PolyvyanyyADG16} and, thus, introduces a bias into the behavioural
comparison. 
In software engineering, this issue is avoided by substituting the behaviour of a program specification with a finite collection of its simulated execution traces~\cite{Lo2006,Walkinshaw2013}.
Still, these approaches suffer from the problem of sampling the suitable finite portion of a possibly infinite behaviour~\cite{Walkinshaw2008}.

In this paper, we address the problem of \emph{how to define meaningful quotients for behavioural comparison of finite and infinite languages.}
To solve this problem, we define measures that quantify the relation between
the specified and recorded behaviours. Concretely, this article contributes:
%\smallskip
\begin{compactenum}[(i)]
\item
A framework for the definition of behavioural quotients that guarantee desired properties.
\item
The definition of two quotients as instantiations of the framework that are grounded in
the cardinality of a language (for finite languages) and the entropy of an automaton (for finite and infinite languages).
\item
Application of the proposed quotients to define monotone precision and recall measures between the behaviour as recorded in an execution log of a system and the behaviour captured in a specification of the system.
\item
A publicly available implementation of the proposed precision and recall quotients.
\item
An evaluation using execution logs of real IT systems that contrasts the monotonicity of our precision and recall quotients with the state-of-the-art measures in process mining.
\end{compactenum}
\smallskip

\noindent
The remainder of this article is structured as follows:
\autoref{sec:back} describes the background of the research problem we address.
\autoref{sec:preliminaries} introduces formal preliminaries in terms of languages and automata.
The framework for the definition of quotients is introduced in \autoref{sec:framework}.
This section also includes two instantiations of the framework and a discussion of formal properties of the quotients.
In \autoref{sec:precision_recall}, we present quotients of precision and recall for comparisons of a collection of recorded system executions with a system specification.
\autoref{sec:implementation} discusses our open source implementation of quotients for comparing specifications and executions of systems.
The precision and recall quotients are compared to other measures in a series of experiments using real-world data in \autoref{sec:precision_recall_evaluation}.
\autoref{sec:related_work} discusses our contributions in the light of related work.
\autoref{sec:discussions} discusses threats to the validity of the reported conclusions, lessons we learned in the course of this work, and issues related to the adoption of the presented methods in software engineering practice. 
Finally, \autoref{sec:conclusion} concludes the paper.

%%%%%%%%%%%%%%%%%%%%%%%%%%%%%%%%%%%%%%%%%%%%%%%%%%%%%%%%%%%%%%%%%%%%%%%%%%%%%%%
\section{Background on Behavioural Comparison}
\label{sec:back}
%%%%%%%%%%%%%%%%%%%%%%%%%%%%%%%%%%%%%%%%%%%%%%%%%%%%%%%%%%%%%%%%%%%%%%%%%%%%%%%

The behaviour of dynamic systems can be captured by the help of languages over their actions. This comes with the benefit that their behavioural differences and commonalities can be analyzed by comparing the respective languages. Behavioural comparisons can be summarized using measures that quantify an aspect of a language, such as its cardinality, i.e., the number of words defined by the language. A ratio of such aspects facilitates a relative comparison of two languages by putting one behaviour is into perspective of some base behaviour. We refer to such a ratio as a
\emph{language quotient}, \ie
\smash{$\mathit{(language)\ quotient} := \frac{\mathit{measure}(\mathit{language}_1)}{\mathit{measure}(\mathit{language}_2)}.$}

%%%%%%%%%%%%%%%%%%%%%%%%%%%%%%%%%%%%%%%%%%%%%%%%%%%%%%%%%%%%%%%%%%%%%%%%%%%%%%%

\renewcommand{\arraystretch}{.6}

\begin{figure*}[h]
\vspace{-3mm}
\centering
\subfloat[Actions]{
   \includegraphics[scale=1] {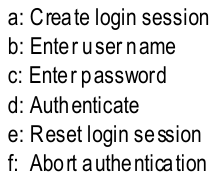}
   \label{fig:intro_example_actions}
 }
%\hspace{1mm}
\subfloat[System $\mathcal{S}_1$]{
   \includegraphics[scale=1] {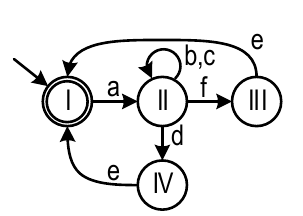}
   \label{fig:intro_example_aut_1}
 }
%\hspace{1mm}
\subfloat[System $\mathcal{S}_2$]{
   \includegraphics[scale=1, trim=0 0 0 0] {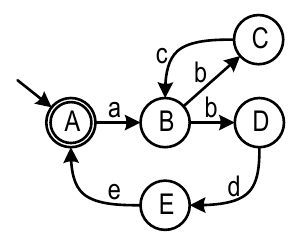}
   \label{fig:intro_example_aut_2}
 }
%\hspace{1mm}
\subfloat[System $\mathcal{S}_3$]{
  \includegraphics[scale=1, trim=0 3mm 0 0] {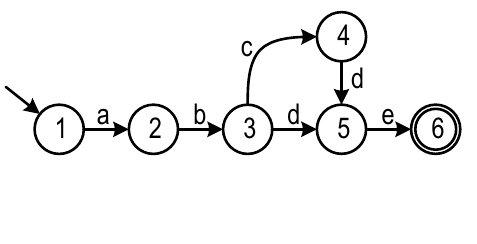}
  \label{fig:intro_example_aut_3}
}
\linebreak
\subfloat[Logs]{
  \label{fig:intro_logs}
  \resizebox{\textwidth}{!}{%
  $\begin{array}{l}
    \mathcal{L}_1:= [\sequence{a,b,d,e}, \sequence{a,b,c,b,c,d,e}]
    \qquad
    \mathcal{L}_2:= \mathcal{L}_1  \cupplus
      [\sequence{a,b,c,c,d,e}, \sequence{a,f,e},\sequence{a,f,e}]
    \qquad
    \mathcal{L}_3:= [\sequence{a,b,c,b,c,d,e}, \sequence{a,b,b,f}, \sequence{a,f,e}]
    \end{array}$
  }
}
\vspace{-1mm}
\caption{Exemplary systems and logs capturing a login process.}
\label{fig:intro_example}
\vspace{-2mm}
\end{figure*}

\renewcommand{\arraystretch}{1}

\begin{example}
\label{ex:running}
For illustration purposes, consider the scenario of a user logging into some 
application. 
\autoref{fig:intro_example_actions} lists the corresponding actions, such as \emph{creating a login session} or conducting the actual \emph{authentication}. Specific realisations of this scenario are given as finite automata in \figurenames~\ref{fig:intro_example_aut_1}--\ref{fig:intro_example_aut_3}. 
Albeit similar, the systems $\mathcal{S}_1$, $\mathcal{S}_2$, and $\mathcal{S}_3$ define different languages over the actions, denoted by $L(\mathcal{S}_1)$,
$L(\mathcal{S}_2)$, and $L(\mathcal{S}_3)$, respectively. 
Note that the languages of $\mathcal{S}_2$ and $\mathcal{S}_1$ are in a subset relation, \ie it holds that $L(\mathcal{S}_2)\subset L(\mathcal{S}_1)$.
Furthermore, \autoref{fig:intro_logs} depicts three logs, $\mathcal{L}_1$, $\mathcal{L}_2$, and $\mathcal{L}_3$, each representing recorded executions of
actual login processes. Each log $\mathcal{L}$ is a multiset of sequences over
actions and, thus, also induces a language $L(\mathcal{L})$. The
latter contains all words that occur at least once in the log.
\hfill\ensuremath{\lrcorner}
\end{example}

The three automata from our example may represent (i) different systems, (ii) 
different versions of the same system, or (iii) system specifications and their
implementations.
In any case, it is useful to quantify to which extent the automata describe the
same behaviour---this answers the question in how far (i) different systems
provide the same functionality; (ii) the functionality of a system has changed
over several versions; and (iii) a specification has been implemented
correctly and completely.

We note that similar questions emerge in the field
of process mining~\cite{Aalst16}, which targets the analysis of
information systems based on recorded executions of a
process. Given a specification and a log,
process mining strives for
quantifying the share of recorded behaviour that
is in line with the specification (\emph{fitness} or \emph{recall} of the log)
or the share of specified behaviour that is actually recorded (\emph{precision}
of the specification).

To address the above use cases, we essentially ask \emph{how much
one system extends the behaviour of another system.} 
For systems $\mathcal{S}_x$ and $\mathcal{S}_y$, such that
$L(\mathcal{S}_y)\subseteq L(\mathcal{S}_x)$, we may answer this question
with a quotient defined using language cardinality as a measurement function:

\smash{$\mathit{(language)\ extension}(\mathcal{S}_x,\mathcal{S}_y) := \frac{|L(\mathcal{S}_x)|}{|L(\mathcal{S}_y)|}.$}

A slightly different way to assess the relation between these systems, however,
is the question of \emph{how much of the behaviour of one system is covered by
another system.} 
To this end, set-algebraic operations over languages may be incorporated in the
definition of a quotient, as in the following definition:
\smash{$\mathit{(language)\ coverage}(\mathcal{S}_x,\mathcal{S}_y) := \frac{|L(\mathcal{S}_x) \cap L(\mathcal{S}_y)|}{|L(\mathcal{S}_x)|}.$}

The above quotients of language extension and coverage provide a
straight-forward means for behavioural comparison of systems, specifications of
systems, and logs. Yet, they are useful only if the applied measurement
function provides a meaningful mapping of a language into a numerical domain.
For the cardinality function used above, we argue that this is the case solely
for finite languages. For languages that define an infinite number of words,
the numerator or denominator of a quotient may become infinity. Leaving aside
the obvious definitional issues, any definition of a value for such a quotient
would not only be arbitrary, but would also result in a single value for all
infinite languages, regardless of their characteristics.

\begin{example}
Taking up Example~\ref{ex:running}, we may compute the \emph{language 
extension} 
using
cardinality as a measure for the logs $\mathcal{L}_1$ and $\mathcal{L}_2$,
capturing that $L(\mathcal{L}_2)$ contains twice as many words as
$L(\mathcal{L}_1)$. However, \emph{language extension} based on cardinality is
not meaningful for any pair of languages of systems $\mathcal{S}_1$,
$\mathcal{S}_2$, and $\mathcal{S}_3$, since $L(\mathcal{S}_1)$ and
$L(\mathcal{S}_2)$ are infinite. In the same vein, computing the \emph{language
  coverage} of a specification and a log, to assess the fitness of the log or 
the
precision of the specification, is not meaningful for the systems
$\mathcal{S}_1$ and $\mathcal{S}_2$, and any of the logs.
\hfill\ensuremath{\lrcorner}
\end{example}

Beyond the challenge posed by infinite languages, we note that quotients have 
to satisfy specific properties. We illustrate these properties using the 
examples introduced above.

\begin{example}
The languages of automata $\mathcal{S}_2$ and $\mathcal{S}_1$ are in the 
subset relation, which should
be reflected in the respective quotients of language extension.
For example, given any log $\mathcal{L}$ such that $L(\mathcal{L}) \subseteq 
L(\mathcal{S}_2)$, it should hold that a quotient of $L(\mathcal{L})$ to 
$L(\mathcal{S}_1)$ should yield a smaller value than a quotient of 
$L(\mathcal{L})$ to $L(\mathcal{S}_2)$.
Since language $L(\mathcal{S}_1)$ contains $L(\mathcal{S}_2)$ and is strictly 
larger, the additional behaviour shall lower the value of the respective 
ratio.
\hfill\ensuremath{\lrcorner}
\end{example}

Desired properties of quotients such as those discussed above translate into
requirements on the measurement functions that capture a particular aspect of
languages. As we will discuss 
in
the remainder, monotonicity of the
measurement function and the existence of a supremum that bounds the
measurement space are of particular relevance in this context. The former
means that adding behaviour to a system strictly
increases (or strictly decreases) the measure, whereas the latter implies that
a specific value is defined as empty behaviour.

Many measures for behavioural comparison proposed in the literature	neglect
such properties, raising debates on how to interpret the obtained results. In
the domain of process mining, e.g., it
was recently shown that none of the existing measures to assess the precision
of a specification against a log satisfies a set of
well-motivated properties~\cite{TaxLSFA17,Aalst18a}.

Against this background, the fundamental
challenge of using quotients for
behavioural comparison is to come up with a framework for their meaningful
definition. That is, the framework should provide guarantees on the quotients to
satisfy a collection of desirable properties.

%%%%%%%%%%%%%%%%%%%%%%%%%%%%%%%%%%%%%%%%%%%%%%%%%%%%%%%%%%%%%%%%%%%%%%%%%%%%%%%
\section{Preliminaries}
\label{sec:preliminaries}
%%%%%%%%%%%%%%%%%%%%%%%%%%%%%%%%%%%%%%%%%%%%%%%%%%%%%%%%%%%%%%%%%%%%%%%%%%%%%%%

This section presents formal notions used to support the discussions in the subsequent sections.

%%%%%%%%%%%%%%%%%%%%%%%%%%%%%%%%%%%%%%%%%%%%%%%%%%%%%%%%%%%%%%%%%%%%%%%%%%%%%%%
\subsection{Multisets, Sequences, and Languages} % , and Event Logs
\label{sec:math}
%%%%%%%%%%%%%%%%%%%%%%%%%%%%%%%%%%%%%%%%%%%%%%%%%%%%%%%%%%%%%%%%%%%%%%%%%%%%%%%

A \emph{multiset}, or a \emph{bag}, is a generalization of a set, \ie 
a collection that can contain multiple instances of the same element.
By $\mathcal{B}(A)$, we denote the set of all finite multisets over some set $A$. 
For some multiset $B \in \mathcal{B}(A)$, $B(a)$ denotes the multiplicity of 
element $a$ in $B$. 
For example, $B_1:=[]$, $B_2:=[b,a,a]$, and $B_3:=[a^2,b]$ are multisets over the set $\{a,b\}$.
Multiset $B_1$ is \emph{empty}, \ie it contains no elements, whereas $B_2(a)=2=B_3(a)$, $B_2(b)=1=B_3(b)$, and, hence, it holds that $B_2=B_3$.
The standard set operations have been extended to deal with multisets as follows.
If element $a$ is a member of multiset $B$, this is denoted by $a \in B$; 
otherwise, 
one writes $a \not\in B$.
The union of two multisets $C$ and $D$, denoted by $C \cupplus D$, is the multiset that contains all elements of $C$ and $D$ such that the multiplicity of an element in the resulting multiset is equal to the sum of multiplicities of this element in $C$ and $D$.
For example, $[b] \cupplus B_2 = [a^2,b^2]$.
Also note that $\mathcal{L}_2$ in \autoref{fig:intro_logs} is the union of $\mathcal{L}_1$ and the multiset of three sequences with two instances of sequence $\sequence{a,f,e}$; more info on sequences is provided below.
The difference of two multisets $C$ and $D$, denoted by $C \setminus D$, is the multiset that for each element $x \in C$ contains $\mathit{max}(0,C(x)-D(x))$ occurrences of $x$.
For example, it holds that $B_3 \setminus B_2 = B_1$, and $B_3 \setminus [b] = [a,a]$.
Given a multiset $B \in \mathcal{B}(A)$ over set $A$, by $\mathit{Set}(B)$ we refer to the set that contains all and only elements in $B$, \ie $\mathit{Set}(B):=\set{b \in A}{b \in B}$.

A \emph{sequence} is an ordered collection of elements.
By $\sigma:=\sequence{a_1,a_2,\ldots,a_n} \in A^*$, we denote a sequence over some set $A$ of length $n \in \Nzero$, $a_i \in A$, $i \in [1..\,n]$, where $[j..\,k]:=\set{x \in \Nzero}{j \leq x \leq k}$, $j,k \in \Nzero$.\footnote{By $\mathbb{N}$ and $\mathbb{N}_0$, we denote the set of all natural numbers excluding and including zero, respectively.}
By $|\sigma|:=n$, we denote the length of the sequence.
By $\sigma_{[i]}$, $i \in [1..\,n]$, we refer to the $i$-th element of $\sigma$, \ie $\sigma_{[i]} = a_i$.
Given a sequence $\sigma$ and a set $K$, by $\sigma|_K$, we denote a sequence obtained from $\sigma$ by deleting all elements of $\sigma$ that are not members of $K$ without changing the order of the remaining elements.
For example, it holds that $\sequence{\texttt{a},\texttt{b},\texttt{d},\texttt{c},\texttt{a}}\!|_{\{\texttt{b},\texttt{c}\}}=\sequence{\texttt{b},\texttt{c}}$.
Given two sequences $\sigma$ and $\sigma'$, by $\sigma \circ \sigma'$, we denote the \emph{concatenation} of $\sigma$ and $\sigma'$, \ie the sequence obtained by appending $\sigma'$ to the end of $\sigma$.
For example, $\sequence{\texttt{a},\texttt{b},\texttt{a}} \circ \sequence{} \circ \sequence{\texttt{b},\texttt{a}} = \sequence{\texttt{a},\texttt{b},\texttt{a},\texttt{b},\texttt{a}}$, where $\sequence{}$ is the empty sequence.
For two sets of sequences $X_1$ and $X_2$ over $A$, $X_1 \circ X_2:=\set{\sigma \in A^*}{\exists\, \sigma_1 \in X_1 \exists\, \sigma_2 \in X_2 : \sigma = \sigma_1 \circ \sigma_2}$.
By $\mathit{suffix}(\sigma,i)$, $i \in \mathbb{N}$, we denote the suffix of $\sigma$ starting from and including position $i$.
\artemDONE{
For example, $\mathcal{L}_1$ in \autoref{fig:intro_logs} contains 
sequences 
$\sigma_1:=\sequence{\texttt{a},\texttt{b},\texttt{d},\texttt{e}}$
and
$\sigma_2:=\sequence{\texttt{a},\texttt{b},\texttt{c},\texttt{b},\texttt{c},\texttt{d},\texttt{e}}$.
It holds that
$\mathit{suffix}(\sigma_1,3)=\sequence{\texttt{d},\texttt{e}}$
 and
$\mathit{suffix}(\sigma_2,6)=\sequence{\texttt{d},\texttt{e}}$.
}

If $\sigma:=\sequence{a_1,a_2,\ldots,a_n} \in A^*$ is a sequence over $A$ and $f$ is a function over $A$, then $f(\sigma):=\sequence{f(a_1),f(a_2),\ldots,f(a_n)}$.
Similarly, if $A' \subseteq A$, then $f(A'):=\set{f(a)}{a \in A'}$.

An \emph{alphabet} is any nonempty finite set. The elements of an alphabet are its \emph{labels}, or \emph{symbols}. By $\Xi$, we denote a universe of symbols.
\artemDONE{
For example, \autoref{fig:intro_example_actions} specifies alphabet 
$\Sigma := \{\texttt{a}, \texttt{b}, \texttt{c}, \texttt{d}, \texttt{e}, 
\texttt{f} \}$.
}
A \emph{word} over an alphabet is a finite sequence of its symbols. 
A (formal) \emph{language} over an alphabet $\Sigma$ is a set of words over $\Sigma$.

%%%%%%%%%%%%%%%%%%%%%%%%%%%%%%%%%%%%%%%%%%%%%%%%%%%%%%%%%%%%%%%%%%%%%%%%%%%%%%%
\subsection{Finite Automata}
\label{sec:automata}
%%%%%%%%%%%%%%%%%%%%%%%%%%%%%%%%%%%%%%%%%%%%%%%%%%%%%%%%%%%%%%%%%%%%%%%%%%%%%%%

We deal with a common notion of a finite automaton~\cite{Hopcroft2007}.
Let $\Xi$ be a universe of labels and let $\tau \in \Xi$ be a special \emph{silent} label.

\begin{define}{Nondeterministic finite automaton}{def:NFA}{\quad\\}
A \emph{nondeterministic finite automaton} (NFA) is a 5-tuple $(Q,\Lambda,\delta,q_0,A)$, where 
$Q$ is a finite nonempty set of \emph{states},
$\Lambda \subset \Xi$ is a set of \emph{labels}, such that $Q$ and $\Xi$ are disjoint,
$\delta : Q \times (\Lambda \cup \{\tau\}) \rightarrow \powerset(Q)$ is the \emph{transition function}, where $\tau \not\in Q \cup \Lambda$,
$q_0 \in Q$ is the \emph{start state}, and 
$A \subseteq Q$ is the \emph{set of accept states}.\footnote{Given a set $A$, by $\powerset(A)$, we denote the powerset of $A$.}
\end{define}

\noindent
An NFA induces a set of computations.

\begin{define}{Computation}{def:computation}{\quad\\}
A \emph{computation} of an NFA $(Q,\Lambda,\delta,q_0,A)$ is either the empty word or a word $s:= \sequence{s_1,s_2,\ldots,s_n}$, $n \in \mathbb{N}$, where every $s_i$ is a member of $\Lambda \cup \{\tau\}$, $i \in [1\, ..\, n]$, and there exists a sequence of states $q:=\left\langle q_0,q_1,\ldots,q_n\right\rangle$, where every $q_j$ is a member of the set of states $Q$, $j \in [1\, ..\, n]$, such that for every $k \in [1\, ..\, n]$ it holds that $q_{k} \in \delta(q_{k-1},s_{k})$.
\end{define}

\noindent
We say that \emph{$s$ leads to $q_n$}. 
By convention, the empty word leads to the start state.
An NFA $B:=(Q,\Lambda,\delta,q_0,A)$ \emph{accepts} a word $s$ \ifaof $s$ is a computation of $B$ that leads to an accept state $q$ of $B$.

\begin{define}{Language of an NFA}{def:NFA:language}{\quad\\}
The \emph{language} of an NFA $B:=(Q,\Lambda,\delta,q_0,A)$, is denoted by $\lang{B}$, and is the set of words that $B$ accepts, \ie $\lang{B}:=\set{s \in \Lambda^*}{\exists\, r \in (\Lambda \cup \{\tau\})^* : \left( (B\,\hspace{1mm}\mathit{accepts}\,\hspace{1mm}r)\,\land\,(s=r|_{\Lambda}) \right)}$.
\end{define}

\noindent
We say that $B$ \emph{recognises} $\lang{B}$.
In an NFA, the transition function takes a state and label to produce the set of possible next states, while in a deterministic finite automaton the transition function takes a state and label and produces the next state. 

\begin{define}{Deterministic finite automaton}{def:DFA}{\quad\\}
A \emph{deterministic finite automaton} (DFA) is an NFA $(Q,\Lambda,\delta,q_0,A)$ such that 
for every state $q \in Q$ it holds that $\delta(q,\tau)=\emptyset$ and 
for every state $q \in Q$ and for every label $s \in \Lambda$ it holds that $|\delta(q,s)| \leq 1$.
\end{define}

% ergodic automaton
An NFA $(Q,\Lambda,\delta,q_0,A)$ is \emph{ergodic} if its underlying graph is strongly irreducible, \ie for all $(x,y) \in Q \times Q$ there exists a sequence of states $\sequence{q_1,\ldots,q_n} \in Q^*$, $n \in \mathbb{N}$, for which it holds that for every $k \in [1\,..\,n-1]$ there exists $\lambda \in \Lambda \cup \{\tau\}$ such that $q_{k+1} \in \delta(q_k,\lambda)$, $q_1=x$, and $q_n=y$. 

A language $L \subseteq \Xi^*$ is \emph{regular} \ifaof it is the language of an NFA.
A language $L \subseteq \Xi^*$ is \emph{irreducible} if, given two words $w_1,w_2 \in L$, there exists a word $w \in \Xi^*$ such that the concatenation $w_1 \circ w \circ w_2$ is in $L$.
A regular language $L$ is irreducible \ifaof it is the language of an ergodic NFA \cite{Ceccherini-SilbersteinMS03}.

An NFA $B:=(Q,\Lambda,\delta,q_0,A)$ is \emph{$\tau$-free} \ifaof for all $q \in Q$ it holds that $\delta(q,\tau)=\emptyset$.
By definition, every DFA is \emph{$\tau$-free}.
Given an NFA $B$, one can always construct a DFA $B'$ that recognises the language of $B$~\cite{Hopcroft2007}.

\begin{example}
We illustrate the above notions using the automaton in 
\autoref{fig:intro_example_aut_2}, which  is defined according to 
our model as $\mathcal{S}_2:=(Q,\Lambda,\delta,q_0,A)$, with states $Q := 
\{\texttt{A},\texttt{B},\texttt{C},\texttt{D},\texttt{E}\}$, labels $\Lambda 
:= \{\texttt{a},\texttt{b},\texttt{c},\texttt{d},\texttt{e}\}$, transition 
function $\delta := \{ ((\texttt{A},\texttt{a}),\{\texttt{B}\}), 
((\texttt{B},\texttt{b}),\{\texttt{C},\texttt{D}\}),
((\texttt{C},\texttt{c}),\{\texttt{B}\}), 
((\texttt{D},\texttt{d}),\{\texttt{E}\}), 
((\texttt{E},\texttt{e}),\{\texttt{A}\}) \}$, 
start state $q_0:=\texttt{A}$, and accept states $A:=\{\texttt{A}\}$. This 
automaton is a $\tau$-free NFA. However, a DFA that recognises the language 
of $\mathcal{S}_2$ may be constructed, as illustrated by \autoref{fig:dfa}. 
\hfill\ensuremath{\lrcorner}
\end{example}

The discussions in~\autoref{sec:framework} and~\autoref{sec:precision_recall} rely on the use of DFAs.
However, software systems and their executions may induce NFA with silent transitions, as observed in the dataset used in the evaluation reported in~\autoref{sec:scalability:evaluation}. 
The transformation of an NFA into an equivalent DFA is an inherent step of the approach which impacts its performance.
Hence, we introduce NFAs here and refer to the transformation from NFAs into DFA explicitly in~\autoref{sec:implementation}.

\begin{figure}
  \centering
\includegraphics[scale=1] {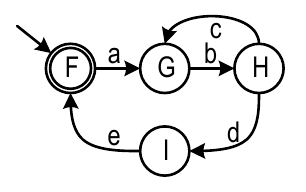}
\vspace{-3mm}
\caption{DFA $\mathcal{S}_4$ that recognises the language of $\mathcal{S}_2$ in~\autoref{fig:intro_example_aut_2}.}
\vspace{-2mm}
\label{fig:dfa}
\end{figure}

%%%%%%%%%%%%%%%%%%%%%%%%%%%%%%%%%%%%%%%%%%%%%%%%%%%%%%%%%%%%%%%%%%%%%%%%%%%%%%%
\section{A Framework for Language Quotients}
\label{sec:framework}
%%%%%%%%%%%%%%%%%%%%%%%%%%%%%%%%%%%%%%%%%%%%%%%%%%%%%%%%%%%%%%%%%%%%%%%%%%%%%%%

This section introduces a framework for behavioural comparison of systems using language quotients.
As detailed in \autoref{sec:framework_definition}, a language quotient is defined based on a
measurement function over the languages of systems.
In \autoref{sec:framework_properties},
we demonstrate that the proposed quotients satisfy desirable properties for
behavioural comparison of systems.
Finally, in \autoref{sec:framework_instantiations}, we propose two
measurement functions for instantiating language quotients, one based on the cardinality of a language and one based on its topological entropy.

%%%%%%%%%%%%%%%%%%%%%%%%%%%%%%%%%%%%%%%%%%%%%%%%%%%%%%%%%%%%%%%%%%%%%%%%%%%%%%%
\subsection{Framework Definition}
\label{sec:framework_definition}
%%%%%%%%%%%%%%%%%%%%%%%%%%%%%%%%%%%%%%%%%%%%%%%%%%%%%%%%%%%%%%%%%%%%%%%%%%%%%%%

A behavioural comparison of systems is usually carried out based on aspects of their languages.
An aspect of a language can be captured by a measure $m : \powerset(\Xi^*)
\rightarrow \mathbb{R}_0^+$, which is a (set) function from the set of all languages over $\Xi$ to non-negative real
numbers.\footnote{By $\mathbb{R}_0^+$, we denote the set of all non-negative real
numbers.}
Two desirable properties of a measure are:
\begin{compactitem}
\item
A measure can be monotonic.
A measure $m$ \emph{is (strictly monotonically) increasing}
\ifaof for all $U \subset \Xi^*$ and $V \subseteq \Xi^*$ such that $U
\subset V$, it holds that $m(U) < m(V)$.
\item
A measure can map the infimum of its domain to the infimum of its codomain.
In this line, we define that a measure $m$ \emph{starts at zero}
\ifaof $m(\emptyset)=0$.
\end{compactitem}
We say that a measure over languages is a \emph{language measure} \ifaof it is increasing and starts
at zero.%
\footnote{Thus, a language measure satisfies the properties of \emph{non-negativity} and defines the empty set to be a \emph{null set} (see~\cite{tao2013introduction} for details). However, it is not required to be \emph{countable} or \emph{finite additive}, as these properties are not exploited in the subsequent analysis of this article. Note that if a language measure $m$ is \emph{countably additive}, $(\Xi^*,\powerset(\Xi^*),m)$ defines a \emph{measure space}, as it is studied in mathematical analysis.}

A language quotient sets aspects of languages into relation as follows:

\begin{define}{Language quotient}{def:language:quotient}{\quad\\}
Given two languages $L_1$ and $L_2$, and a language measure $m$, the
\emph{language quotient} of $L_1$ over $L_2$ induced by $m$ is the fraction of the measure of $L_1$ over the measure of $L_2$:
$$
\mathit{quotient}_m(L_1,L_2):=\frac{m(L_1)}{m(L_2)}.
\vspace{-3mm}
$$
\end{define}

\noindent
\emph{Nomen est omen}, a language quotient is defined over languages, not
systems. The rationale behind this formalisation is that the framework of
language quotients, once instantiated with a specific
measure, may be applied for diverse algebraic operations;
examples include quotients that are defined over the intersection, union, or difference of languages,
(see the notion of \emph{language coverage} in \autoref{sec:introduction} for illustration).
In \autoref{sec:precision_recall_definition}, we provide further examples of quotients over the intersection of languages that are useful in the context of process mining.

%%%%%%%%%%%%%%%%%%%%%%%%%%%%%%%%%%%%%%%%%%%%%%%%%%%%%%%%%%%%%%%%%%%%%%%%%%%%%%%
\subsection{Properties of Language Quotients}
\label{sec:framework_properties}
%%%%%%%%%%%%%%%%%%%%%%%%%%%%%%%%%%%%%%%%%%%%%%%%%%%%%%%%%%%%%%%%%%%%%%%%%%%%%%%

Language quotients enjoy useful properties that rest on the properties of a language measure.
One can compare quotients with the same numerators as follows.

\begin{lem}{Fixed numerator quotients}{lem:fixed:num}{\quad\\}
If $L_1, L_2, L_3 \!\subseteq\! \Xi^*$ are languages such that $L_1$ is nonempty, $L_1 \subset L_2$, and $L_2 \subset L_3$, then it holds that $\mathit{quotient}_m(L_1,L_3) < \mathit{quotient}_m(L_1,L_2)$, where $m$ is a language measure.
\end{lem}
\begin{proof}
Let us assume that $L_1 \neq \emptyset$, $L_1 \subset L_2$, and $L_2 \subset L_3$, but it holds that $\mathit{quotient}_m(L_1,L_2) \leq \mathit{quotient}_m(L_1,L_3)$.
Because $m$ starts at zero and is increasing, it holds that $0 < m(L_2) < m(L_3)$.
Because $m(L_1)>0$, we reach a contradiction.
\end{proof}

\noindent
The statement of the lemma is shown schematically in \autoref{fig:quotients:properties} (top row).
If $L_2$ and $L_3$ are languages of two systems that extend the behaviour of a third system that recognises language $L_1$, then, using the quotients, one can conclude that the system that recognises $L_3$ \emph{extends} the behaviour of the system that recognises $L_1$ more than does the system that recognizes $L_2$.
The difference between the extension behaviours is captured by $\mathit{quotient}_m(L_1,L_3) - \mathit{quotient}_m(L_1,L_2)$.
The meaning of the difference depends on the meaning of language measure $m$ used to instantiate the quotients.
If $m$ measures the cardinality of a language, then the difference stands for the fraction of the behaviour with which $L_3$ extends $L_1$ more than does $L_2$.

\begin{figure}[t!]
	\begin{center}
		\includegraphics[scale=0.25]{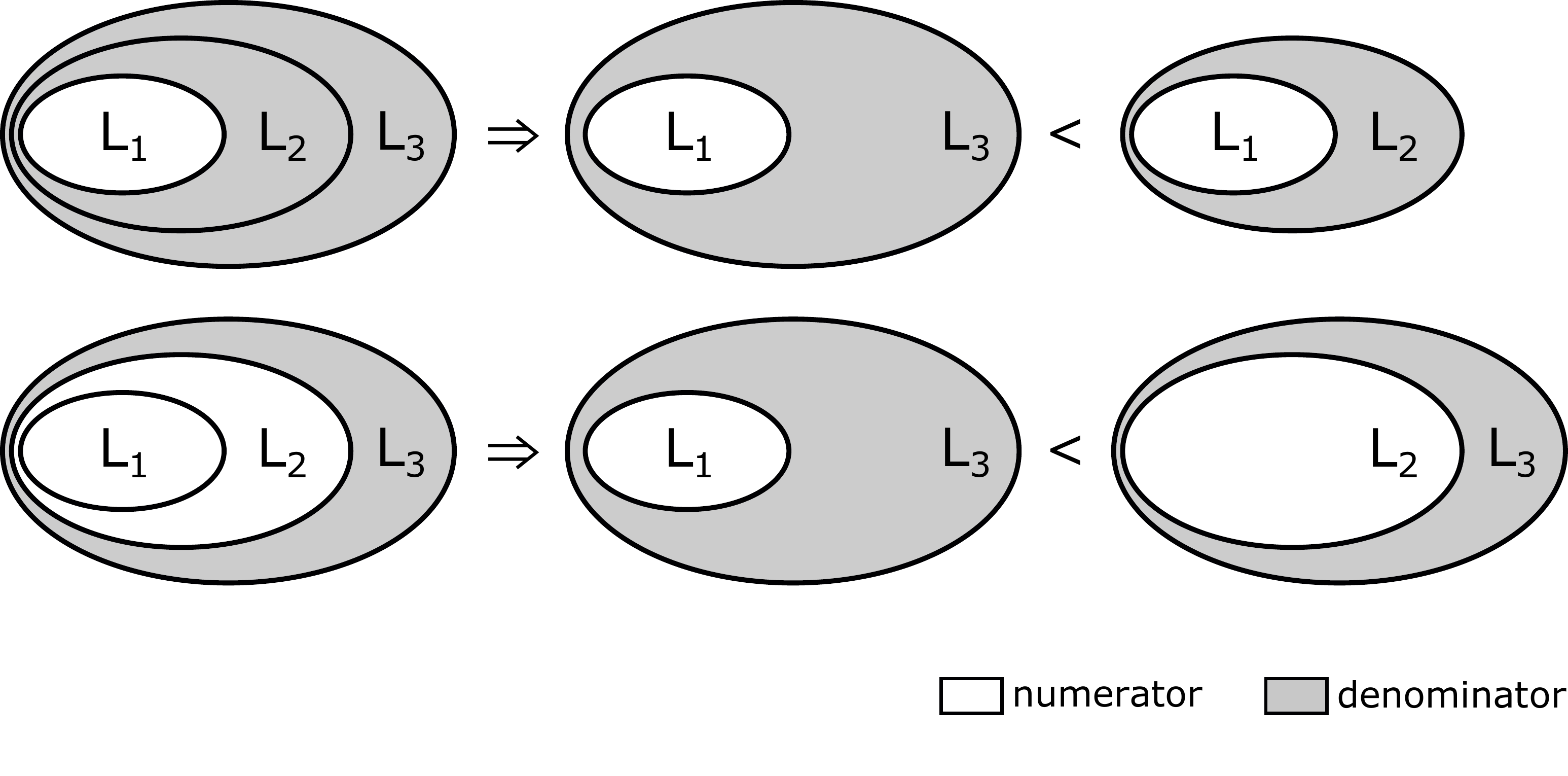}
	\vspace{-3mm}
		\caption{Schematic representation of: \lemmaname~\ref{lem:fixed:num} (top row) and \lemmaname~\ref{lem:fixed:den} (bottom row).}
		\label{fig:quotients:properties}
	\end{center}
	\vspace{-2mm}
\end{figure}

Moreover, language quotients with the same denominators can be compared as below.

\begin{lem}{Fixed denominator quotients}{lem:fixed:den}{\quad\\}
If $L_1, L_2, L_3 \subseteq \Xi^*$ are languages such that $L_1 \subset L_2$ and $L_2 \subset L_3$, then it holds that $\mathit{quotient}_m(L_1,L_3) < \mathit{quotient}_m(L_2,L_3)$, where $m$ is a language measure.
\end{lem}
\begin{proof}
Assume that $L_1 \subset L_2$ and $L_2 \subset L_3$ but it holds that $\mathit{quotient}_m(L_2,L_3) \leq \mathit{quotient}_m(L_1,L_3)$.
Because $m$ starts at zero and is increasing, it holds that $0 \leq m(L_1) < m(L_2)$.
Because $m(L_3)>0$, we reach a contradiction.
\end{proof}

\noindent
The statement of the lemma is visualized schematically in \autoref{fig:quotients:properties} (bottom row).
For example, if $L_3$ is a language of a specification of a system, and $L_1$ and $L_2$ are languages of its two implementations, then, based on the quotients, one can conclude that the implementation that recognises $L_2$ is more complete than the implementation that recognizes $L_1$.
In other words, $L_2$ has better \emph{coverage} of the specification than $L_1$.
The extent to which the implementation that recognises $L_2$ is more complete can be quantified by $\mathit{quotient}_m(L_2,L_3) - \mathit{quotient}_m(L_1,L_3)$.
The meaning of the difference, again, depends on the meaning of language measure $m$ used to instantiate the quotients.

If one fixes the numerator, like in the case of comparing the amounts to which various systems extend a given behaviour, the quotients are bounded below.

\begin{corol}{Fixed numerator quotients}{cor:fixed:num}
If $L_1, L_2 \subset \Xi^*$ are languages such that $L_1 \subset L_2$, then it holds that $\mathit{quotient}_m(L_1,\Xi^*) < \mathit{quotient}_m(L_1,L_2)$, where $m$ is a language measure.
\end{corol}

\noindent
\corolname~\ref{cor:fixed:num} follows immediately from \lemmaname~\ref{lem:fixed:num}, as it holds that $L_1 \subset L_2$ and $L_2 \subset \Xi^*$.

%%%%%%%%%%%%%%%%%%%%%%%%%%%%%%%%%%%%%%%%%%%%%%%%%%%%%%%%%%%%%%%%%%%%%%%%%%%%%%%
\subsection{Framework Instantiations}
\label{sec:framework_instantiations}
%%%%%%%%%%%%%%%%%%%%%%%%%%%%%%%%%%%%%%%%%%%%%%%%%%%%%%%%%%%%%%%%%%%%%%%%%%%%%%%

This section proposes two language quotients, as instantiations of \definitionname~\ref{def:language:quotient}
using specific measurement functions.
Thus, these quotients have all the properties proposed in~\autoref{sec:framework_properties}.
The first quotient is based on the cardinality of a language, whereas the
other one is grounded in the notion of topological entropy.

%%%%%%%%%%%%%%%%%%%%%%%%%%%%%%%%%%%%%%
\mypar{Cardinality quotient}
%%%%%%%%%%%%%%%%%%%%%%%%%%%%%%%%%%%%%%
As language $L$ is a set of words, its \emph{cardinality}, denoted by $|L|$, is a property that can serve as the basis for behavioural comparison.
Clearly, cardinality is a language measure, \ie it \emph{is increasing} and \emph{starts at zero}.
By defining a language quotient based on this measure, we obtain the
cardinality quotient:

\begin{define}{Cardinality quotient}{def:cardinality:quotient}%{\quad\\}
The \emph{cardinality quotient} of language $L_1$ over language $L_2$ is the fraction of the cardinality of $L_1$ over the cardinality of $L_2$, \ie
$
\mathit{quotient}_{\mathit{car}}(L_1,L_2):=\frac{|L_1|}{|L_2|}.
$
\end{define}

\noindent
The cardinality quotient captures the ratio of the sizes of two
languages.
It is well-defined only for $L_2\neq \emptyset$.
Note that this is a definitional issue that may be addressed explicitly (e.g., defining
$\mathit{quotient}_{\mathit{car}}(L_1,L_2):=0$ if $L_2=\emptyset$). A more severe problem is
the computation of the quotient for infinite languages.
Given an alphabet, finite by definition, a regular language may define a
countably infinite set of words~\cite{Sipser2012}. For example,
the cardinality of
an irreducible regular language is infinity. Again, one may address the
resulting definitional issues explicitly, e.g., by adopting that a constant
divided by infinity is equal to zero and that infinity divided by infinity is
equal to one. However, any such convention is not useful for behavioural
comparison in the context of regular languages. For instance, the
\emph{language extension} and \emph{language coverage}, see
\autoref{sec:back}, would be equal to one for any pair of ergodic
automata, such as those in \autoref{fig:intro_example_aut_1}
and \autoref{fig:intro_example_aut_2}.
We thus conclude that cardinality quotients provide a suitable means for behavioural comparison solely for finite languages.

%%%%%%%%%%%%%%%%%%%%%%%%%%%%%%%%%%%%%%
\mypar{Eigenvalue quotient}
%%%%%%%%%%%%%%%%%%%%%%%%%%%%%%%%%%%%%%
To obtain language quotients that are useful for comparing infinite languages, we instantiate them with a measure based on the topological entropy.
Intuitively, the topological entropy of a language captures the increase in variability of the words of the language as their length goes to infinity.

Given a language $L$, let $C_n(L)$, $n \in \Nzero$, be the set of all the words in $L$ of length $n$, \ie $C_n(L):=\set{x \in L}{|x|=n}$. 
Then, the \emph{topological entropy} of $L$ is defined as follows (see~\cite{Parry64,Ceccherini-SilbersteinMS03} for details)\footnote{Given a sequence $(x_n)$, $\limsup\limits_{n \to \infty}{x_n}$ is the \emph{limit superior} of $(x_n)$ and is defined by 
$\inf \set{\sup \set{x_m}{m \geq n}}{n \geq 0}$.}:
% entropy of a language
\[
\ent{L} := \limsup\limits_{n \to \infty}{\frac{\log{|C_n(L)|}}{n}}.
\]

\renewcommand{\arraystretch}{.6}

\begin{figure*}[h]
\vspace{-3mm}
\centering
\subfloat[$B_1$]{
   \includegraphics[scale=1] {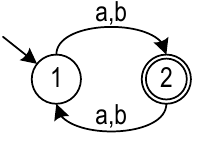}
   \label{fig:counter:example:B1}
 }
\hspace{15mm}
\subfloat[$B_2$]{
   \includegraphics[scale=1] {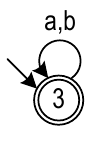}
   \label{fig:counter:example:B2}
 }
\hspace{15mm}
\subfloat[$\hat{B_1}$]{
   \includegraphics[scale=1] {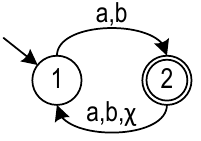}
   \label{fig:counter:example:B1:fix}
 }
\hspace{15mm}
\subfloat[$\hat{B_2}$]{
   \includegraphics[scale=1] {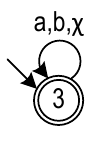} % File `fig/fix/B2_fix' not found.
   \label{fig:counter:example:B2:fix}
 }
\vspace{-1mm}
\caption{Four DFAs.}
\label{fig:fix}
\vspace{-2mm}
\end{figure*}

\renewcommand{\arraystretch}{1}

\noindent
Topological entropy characterises the complexity of a language and is closely related to the properties of the DFAs that recognise this language.
For a DFA $B:=(Q,\Lambda,\delta,q_0,A)$, 
with $C_n(B)$, $n \in \Nzero$, we denote the set of all the words in $\lang{B}$ of length $n$,
\ie $C_n(B):=\set{x \in \lang{B}}{|x|=n}$.
Then, the topological entropy of $B$ is defined as the topological entropy of the language that it recognises~\cite{Ceccherini-SilbersteinMS03}:
\[
\ent{L(B)} = \ent{B} := \limsup\limits_{n \to \infty}{\frac{\log{|C_n(B)|}}{n}}.
\]

% Adjacency matrix of a DFA
\noindent
The topological entropy of a DFA, and thus of its language, is further related to the
structure of the automaton.
Below, we shall deal with square non-negative matrices $G:=\{g_{ij}\}$, $i,j \in [1\,..\,n]$, $n \in \mathbb{N}$, \ie $g_{ij} \geq 0$ for all $i,j \in [1\,..\,n]$.
The adjacency matrix of a DFA $(Q,\Lambda,\delta,q_0,A)$, where $Q:=\{q_0,q_1,...,q_n\}$, $n \in \Nzero$, is a square matrix $G:=\{g_{ij}\}$, $i,j \in [1\,..\,|Q|]$, such that
$g_{ij}:=|\set{q_j \in \delta(q_i,\lambda)}{\lambda \in \Lambda}|$, for all $i,j \in [1\,..\,|Q|]$.\footnote{Recall from \autoref{sec:automata} that every DFA is $\tau$-free.}
The topological entropy of an ergodic DFA $B$, \ie $\ent{B}$, is given by the logarithm
of the Perron-Frobenius eigenvalue of its adjacency matrix, which is a unique
largest real eigenvalue of the adjacency matrix of
$B$~\cite{Ceccherini-SilbersteinMS03}.
Note that an adjacency matrix of an ergodic DFA $B$ has an eigenvalue $r$ such
that $r$ is real, $r>0$, and $r \geq |\lambda|$ for any eigenvalue $\lambda$
of the adjacency matrix of $B$~\cite[\theoremname~1.5]{Seneta2006Non-NegativeMatricesand}.
The relation between the entropy of a language and the entropy of an ergodic
DFA recognising this language, as outlined above, is important for
computational reasons, as it provides us with a straight-forward approach
to compute the entropy of a language, via the Perron-Frobenius theory.

Topological entropy is \emph{not} an increasing measure over regular languages. 
It is neither an increasing measure over irreducible regular languages. 
Indeed, for two ergodic automata $B_1$ and $B_2$ shown in \autoref{fig:counter:example:B1} and \autoref{fig:counter:example:B2}, respectively,
it holds that 
$\lang{B_1} \subset \lang{B_2}$ and $\ent{\lang{B_1}}=1.0=\ent{\lang{B_2}}$; note that logarithm base two was used to compute the entropy.

Let $U \subset V$ be two regular languages over alphabet $\Psi \subset \Xi$ such that $U \subset V$.
Let $\hat{U}$ and $\hat{V}$ denote the languages $(U \circ \{\sequence{\chi}\})^* \circ U$ and $(V \circ \{\sequence{\chi}\})^* \circ V$, $\chi \in \Xi \setminus \Psi$, respectively. 
We say that \emph{$\hat{U}$ and $\hat{V}$ are the results of short-circuiting $U$ and $V$ with $\chi$}.
Note that given an automaton $B$ it is straight-forward to construct an automaton that recognizes the short-circuited version of $\lang{B}$.
This can be achieved by inserting fresh transitions in $B$, each labelled with $\chi$, from each accept state of $B$ to its start state to obtain automaton $\hat{B}$.
For example, automata $\hat{B_1}$ and $\hat{B_2}$ from \autoref{fig:counter:example:B1:fix} and \autoref{fig:counter:example:B2:fix}, respectively, recognise the short-circuited versions of languages $\lang{B_1}$ and $\lang{B_2}$, where $B_1$ and $B_2$ are shown in \autoref{fig:counter:example:B1} and \autoref{fig:counter:example:B2}.
Note that any augmented in this way automaton is guaranteed to be ergodic and, thus, the language such an automaton recognises is irreducible.

Let $(u_n)_{n=1}^{\infty}$ and $(v_n)_{n=1}^{\infty}$ be two sequences such that:
\[
u_n := \frac{\log{|C_n(\hat{U})|}}{n} \;\;\;\;\;\;\;\;\text{and}\;\;\;\;\;\;\;\; v_n := \frac{\log{|C_n(\hat{V})|}}{n}.
\]

\noindent
For every $n \in \mathbb{N}_0$ it holds that $C_n(\hat{U}) \subseteq C_n(\hat{V})$ because $\hat{U} \subset \hat{V}$.
Let $w_u \in U$ and $w_v \in V \setminus U$ be two words.
Let $(\alpha_i)_{i=1}^{\infty}$ and $(\beta_i)_{i=1}^{\infty}$ be two sequences such that $\alpha_1:=|w_v|$, $\beta_1:=|w_u|$,
\[
\alpha_j:=\alpha_{j-1}+\frac{\mathit{LCM}(|w_u|+1,|w_v|+1)}{|w_u|+1},\;\text{and}\;\beta_j:=\beta_{j-1}+\frac{\mathit{LCM}(|w_u|+1,|w_v|+1)}{|w_v|+1},\;j>1.\footnote{By $\mathit{LCM}(a,b)$ we denote the the \emph{least common multiple} of two integers $a$ and $b$.}
\]

\noindent
Then, for any $k\in\mathbb{N}$ it holds by induction that $|(w_u \circ \sequence{\chi})^{\alpha_k} \circ w_u|=|(w_v \circ \sequence{\chi})^{\beta_k} \circ w_v|$.\footnote{Given a word $w$, by $w^k$, $k \in \mathbb{N}$, we denote concatenation of $k$ instances of $w$, \eg $(\texttt{ab})^3=\texttt{ababab}$.}
Note that $w_u \in V$ and $w_v \not\in U$.
Hence, for any $n \in \mathbb{N}$, $\sup \set{u_m}{m \geq n} < \sup \set{v_m}{m \geq n}$ and, consequently:
\[
\inf \set{\sup \set{u_m}{m \geq n}}{n \geq 0} \;<\; \inf \set{\sup \set{v_m}{m \geq n}}{n \geq 0}.
\]

\noindent
That is, it holds that $\ent{\hat{U}} < \ent{\hat{V}}$.
As a consequence of this fact, it holds that $\ent{\hat{B_1}}=1.2925$ and $\ent{\hat{B_2}}=1.585$, \ie $\ent{\hat{B_1}}<\ent{\hat{B_2}}$, where $\hat{B_1}$ and $\hat{B_2}$ are automata shown in \autoref{fig:counter:example:B1:fix} and \autoref{fig:counter:example:B2:fix}, respectively; again, logarithm base two was used to compute the entropy.

Given a measure over languages, one can obtain a corresponding short-circuit measure as follows.

\begin{define}{Short-circuit measure}{def:short:circuit:measure}{\quad\\}
A \emph{short-circuit measure} over languages over alphabet $\Psi \subset \Xi$
induced by a measure over languages $m : \powerset(\Xi^*) \rightarrow \mathbb{R}_0^+$ is
the (set) function $m^\bullet : \powerset(\Psi^*) \rightarrow \mathbb{R}_0^+$ defined by
$m^\bullet(L):=m((L \circ \{\sequence{\chi}\})^* \circ L)$, where $L$ is a language over $\Psi$, \ie $L \subseteq \Psi^*$, and $\chi \in \Xi \setminus \Psi$ is a short-circuit symbol.
\end{define}

\noindent
By $\eig{L}$, where $L$ is a language, we denote the Perron-Frobenius eigenvalue of the adjacency matrix of a DFA that recognises $L$, and call it the \emph{eigenvalue measure} of $L$.
We also say that $\eig{L}$ is the \emph{eigenvalue} of $L$.
We accept that the adjacency matrix of a DFA that induces the empty language is the zero square matrix of order one.
Hence, the above definitions and observations lead to the next conclusion.

\begin{lem}{Language measures}{lem:language:measures}{\quad\\} %
The short-circuit topological entropy ($\mathit{ent}^\bullet$) and the short-circuit eigenvalue measure ($\mathit{eig}^\bullet$) are language measures, \ie they are increasing and start at zero.
\end{lem}

\noindent
The fact that the short-circuit eigenvalue measure is increasing follows immediately from the facts that 
(i) the short-circuit topological entropy is a language measure and (ii) the logarithm is a strictly increasing function.
Finally, to avoid decisions of which logarithm base to use when computing the entropy, in what follows, we define and use language quotients induced by the eigenvalue measure.

\begin{define}{Eigenvalue quotient}{def:entropy:quotient}{\quad\\}
Given two regular languages $L_1$ and $L_2$, the \emph{eigenvalue quotient} of $L_1$ over $L_2$ is 
the fraction of 
the short-circuit eigenvalue measure of $L_1$ over the short-circuit eigenvalue measure of $L_2$, \ie
\[
\mathit{quotient}_{\mathit{eig}^\bullet}(L_1,L_2):=\frac{\mathit{eig}^\bullet(L_1)}{\mathit{eig}^\bullet(L_2)}.
\vspace{-5mm}
\]
\end{define}

\begin{figure}[t]
  \vspace{-2mm}
  \centering
  \includegraphics[scale=1] {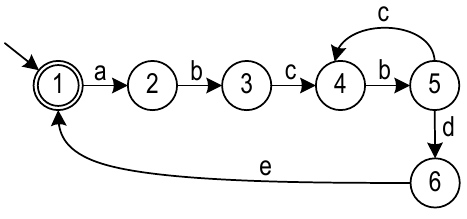}
	\vspace{-3mm}
  \caption{DFA $\mathcal{S}_5$.}
  \label{fig:sect4:S5}
	\vspace{-2mm}
\end{figure}

\noindent
\begin{example}
Consider automata $\mathcal{S}_1$, $\mathcal{S}_4$, and $\mathcal{S}_5$ in \autoref{fig:intro_example_aut_1}, \autoref{fig:dfa}, and \autoref{fig:sect4:S5}.
These three automata are ergodic and it holds that $L(\mathcal{S}_5) \subset L(\mathcal{S}_4)$ and $L(\mathcal{S}_4) \subset L(\mathcal{S}_1)$.
It holds that
$\mathit{quotient}_{\mathit{eig}^\bullet}(L(\mathcal{S}_4),\linebreak L(\mathcal{S}_1))=0.6$,
$\mathit{quotient}_{\mathit{eig}^\bullet}(L(\mathcal{S}_5),L(\mathcal{S}_1))=0.5525$,
$\mathit{quotient}_{\mathit{eig}^\bullet}(L(\mathcal{S}_5), L(\mathcal{S}_4))=0.9208$, and\linebreak
$\mathit{quotient}_{\mathit{eig}^\bullet}(L(\mathcal{S}_5),\Phi^*)=0.2321$, where $\Phi$ is the set $\{\texttt{a},\texttt{b},\texttt{c},\texttt{d},\texttt{e}\}$.
Indeed, it holds that:
\begin{compactenum}[(i)]
\item $\mathit{quotient}_{\mathit{eig}^\bullet}(L(\mathcal{S}_5),L(\mathcal{S}_1)) < \mathit{quotient}_{\mathit{eig}^\bullet}(L(\mathcal{S}_5),L(\mathcal{S}_4))$ --- see \lemmaname~\ref{lem:fixed:num};
\item $\mathit{quotient}_{\mathit{eig}^\bullet}(L(\mathcal{S}_5),L(\mathcal{S}_1)) < \mathit{quotient}_{\mathit{eig}^\bullet}(L(\mathcal{S}_4),L(\mathcal{S}_1))$ --- see \lemmaname~\ref{lem:fixed:den}; and
\item $\mathit{quotient}_{\mathit{eig}^\bullet}(L(\mathcal{S}_5),\Phi^*) < \mathit{quotient}_{\mathit{eig}^\bullet}(L(\mathcal{S}_5),L(\mathcal{S}_4))$ --- see \corolname~\ref{cor:fixed:num}.
\end{compactenum}
\medskip

\begin{figure}[t]
\vspace{-3mm}
\centering
\subfloat[]{
$
\footnotesize
\renewcommand\arraystretch{1.2}
\begin{array}{c|cccc}
               & \mathrm{F} & \mathrm{G} & \mathrm{H} & \mathrm{I} \\
\hline
\mathrm{F}   & 1 & 1 & 0 & 0 \\
\mathrm{G}   & 0 & 0 & 1 & 0 \\
\mathrm{H}   & 0 & 1 & 0 & 1 \\
\mathrm{I}   & 1 & 0 & 0 & 0
\end{array}
$
\label{fig:matrix:S4}
}
\hspace{9mm}
\subfloat[]{
$
\footnotesize
\renewcommand\arraystretch{1.2}
\begin{array}{c|cccccc}
             & \mathrm{1} & \mathrm{2} & \mathrm{3} & \mathrm{4} & \mathrm{5} & \mathrm{6} \\
\hline
\mathrm{1}   & 1 & 1 & 0 & 0 & 0 & 0 \\
\mathrm{2}   & 0 & 0 & 1 & 0 & 0 & 0 \\
\mathrm{3}   & 0 & 0 & 0 & 1 & 0 & 0 \\
\mathrm{4}   & 0 & 0 & 0 & 0 & 1 & 0 \\
\mathrm{5}   & 0 & 0 & 0 & 1 & 0 & 1 \\
\mathrm{6}   & 1 & 0 & 0 & 0 & 0 & 0
\end{array}
$
\label{fig:matrix:S5}
}
\vspace{-3mm}
\caption{Adjacency matrices of the short-circuit, \ie after insertions of the $\chi$ transitions, versions of DFAs from (a) \autoref{fig:dfa} and (b) \autoref{fig:sect4:S5}.}
\label{fig:two:matrices}
\vspace{-2mm}
\end{figure}

\noindent
To show that $\mathit{quotient}_{\mathit{eig}^\bullet}(L(\mathcal{S}_5),L(\mathcal{S}_4))$ indeed equals to $0.9208$, \autoref{fig:matrix:S4} and \autoref{fig:matrix:S5} show adjacency matrices of $\mathcal{S}_4$ and $\mathcal{S}_5$, respectively.
Note that the Perron-Frobenius eigenvalue of the matrix in \autoref{fig:matrix:S4} is $1.5129$, 
while the Perron-Frobenius eigenvalue of the matrix in \autoref{fig:matrix:S5} is $1.3931$. 
\hfill\ensuremath{\lrcorner}
\end{example}

%%%%%%%%%%%%%%%%%%%%%%%%%%%%%%%%%%%%%%%%%%%%%%%%%%%%%%%%%%%%%%%%%%%%%%%%%%%%%%%
%%%%%%%%%%%%%%%%%%%%%%%%%%%%%%%%%%%%%%%%%%%%%%%%%%%%%%%%%%%%%%%%%%%%%%%%%%%%%%%
\section{Precision and Recall}
\label{sec:precision_recall}
%%%%%%%%%%%%%%%%%%%%%%%%%%%%%%%%%%%%%%%%%%%%%%%%%%%%%%%%%%%%%%%%%%%%%%%%%%%%%%%

Language quotients provide a general means for behavioural comparison.
To demonstrate the use of the quotients, this section proposes and discusses their application in process mining~\cite{Aalst16}.
One of the problems studied in process mining is the problem of \emph{process discovery}.
Given a log of recorded executions of a system, a discovery technique constructs a specification of the system that ``represents'' the behaviour captured in the log.
As a system may execute same sequences of actions multiple times, its log is a multiset of words that encode the executions.

\begin{define}{Log}{def:event:log}{\quad\\}
A \emph{log} is a finite multiset over a language.
\end{define}

\noindent
An element of a log is a \emph{trace}, whereas an element of a trace is an
\emph{event} of the trace. Given a log $\mathcal{L}$, 
$L(\mathcal{L}):=\mathit{Set}(\mathcal{L})$ is the
\emph{language} of $\mathcal{L}$.

\begin{example}
\artemDONE{For example, logs $\mathcal{L}_1$, $\mathcal{L}_2$, and 
$\mathcal{L}_3$, listed
	in \autoref{fig:intro_logs} contain two, five, and three traces, 
	respectively.
	Note that $\mathcal{L}_2$ contains trace $\sequence{a,f,e}$ twice, which
	denotes that this sequence of actions was recorded in the log two times.}
\hfill\ensuremath{\lrcorner}
\end{example}

The quality of a generated process specification is typically evaluated using \emph{precision}, \emph{fitness} (a specific type of recall), \emph{simplicity}, and \emph{generalization}~\cite{Aalst16}.
We use the framework of behavioural quotients to define precision and recall of specifications w.r.t. logs
(\autoref{sec:precision_recall_definition}).
We demonstrate that our precision and recall quotients satisfy important requirements for precision and recall measures 
(\autoref{sec:precision_recall_properties}).

%%%%%%%%%%%%%%%%%%%%%%%%%%%%%%%%%%%%%%%%%%%%%%%%%%%%%%%%%%%%%%%%%%%%%%%%%%%%%%%
\subsection{Definition of Precision and Recall}
\label{sec:precision_recall_definition}
%%%%%%%%%%%%%%%%%%%%%%%%%%%%%%%%%%%%%%%%%%%%%%%%%%%%%%%%%%%%%%%%%%%%%%%%%%%%%%%

This section proposes two quotients for comparing behaviours captured in a given log and DFA, namely \emph{precision} and \emph{recall} of the DFA w.r.t. the log.
These quotients are inspired by the precision and recall measures that have proved to be useful in information retrieval, binary classification, and pattern recognition.
The precision and recall measures proposed here can be used to measure
precision and fitness, respectively, of specifications discovered from logs.

In information retrieval, given a set of relevant documents and a set of retrieved documents, \emph{precision} is the fraction of relevant retrieved documents over the retrieved documents.
Given a log and a DFA, we propose to measure how precisely a DFA
(specification)
describes a log as the fraction of executions recorded in the log and specified
in the DFA over all the executions (of which there can be infinitely many)
specified in the DFA.

\begin{define}{Precision of DFA w.r.t. log}{def:pm:precision}{\quad\\}
Given a log $\mathcal{L}$ and a DFA $B$,
the \emph{precision} of $B$ w.r.t. $\mathcal{L}$ induced by a language measure $m$ is denoted by $\mathit{precision}_m(B,\mathcal{L})$ and is the language quotient induced by $m$ of the intersection of the languages of $B$ and $\mathcal{L}$ over the language of $B$, \ie
$
\mathit{precision}_m(B,\mathcal{L}):=\mathit{quotient}_m(L(B) \cap L(\mathcal{L}),L(B)).
$
\end{define}

\noindent
Precision is the ratio of the measure of traces of the log that are also computations of the DFA (specified and recorded behaviour) to the measure of all the computations of the DFA (specified behaviour).

\begin{example}
The precision of automaton $\mathcal{S}_3$ in
\autoref{fig:intro_example_aut_3} w.r.t. log $\mathcal{L}_2$ in
\autoref{fig:intro_logs} induced by the cardinality of a language is computed
as follows:
\smash{$\mathit{precision}_\mathit{car}(\mathcal{S}_3,\mathcal{L}_2) = \frac{|L(\mathcal{S}_3) \cap L(\mathcal{L}_2)|}{|L(\mathcal{S}_3)|} = \frac{1}{2}$};
the languages of $\mathcal{S}_3$ and $\mathcal{L}_2$ share one word, $\sigma:=\sequence{\texttt{a},\texttt{b},\texttt{d},\texttt{e}}$, while the language of $\mathcal{S}_3$ has two words: $\sigma$ and $\sequence{\texttt{a},\texttt{b},\texttt{c},\texttt{d},\texttt{e}}$.
\hfill\ensuremath{\lrcorner}
\end{example}

In information retrieval, given a set of relevant documents and a set of retrieved documents, \emph{recall} is the fraction of relevant retrieved documents over the relevant documents.
Given a log and a DFA, we measure how well the DFA captures the
behaviour of the log as the fraction of executions recorded in the
log and specified in the DFA over all the behaviour recorded in the log.

\begin{define}{Recall of DFA w.r.t. log}{def:pm:recall}{\quad\\}
Given a log $\mathcal{L}$ and a DFA $B$,
the \emph{recall} of $B$ w.r.t. $\mathcal{L}$ induced by a language measure $m$ is denoted by $\mathit{recall}_m(B,\mathcal{L})$ and is the language quotient induced by $m$ of the intersection of the languages of $B$ and $\mathcal{L}$ over the language of $\mathcal{L}$, \ie
$
\mathit{recall}_m(B,\mathcal{L}):=\mathit{quotient}_m(L(B) \cap L(\mathcal{L}),L(\mathcal{L})).
$
\end{define}

%\noindent
Recall is therefore the ratio of the measure of traces of the log that are also
computations of the DFA (specified and recorded behaviour) to the measure of
the
traces of the log (recorded behaviour).

\begin{example}
For example, the recall of automaton $\mathcal{S}_3$ in
\autoref{fig:intro_example_aut_3} w.r.t. log $\mathcal{L}_2$ in
\autoref{fig:intro_logs} induced by the cardinality of a language is computed
as follows:
\smash{$\mathit{recall}_\mathit{car}(\mathcal{S}_3,\mathcal{L}_2) = \frac{|L(\mathcal{S}_3) \cap L(\mathcal{L}_2)|}{|L(\mathcal{L}_2)|} = \frac{1}{4}.$}
This result is easy to verify by checking that the language of $\mathcal{L}_2$ 
consists of four words.
\hfill\ensuremath{\lrcorner}
\end{example}

The notions of fitness and recall of a DFA w.r.t. a log  take a language measure as a parameter.
The language of a log is finite.
If the language of the DFA is also finite, one can instantiate the precision and recall with the cardinality of a language, as proposed above.
The language of a DFA is, however, often infinite.
To overcome this limitation, we suggest instantiating the precision and recall with the short-circuit eigenvalue measure, as illustrated below; see \autoref{sec:framework_instantiations} for details on the short-circuit eigenvalue measure.

\begin{figure*}[t]
\vspace{-2mm}
\centering
\subfloat[]{
   \includegraphics[scale=0.95,trim=0 -1mm 0 0] {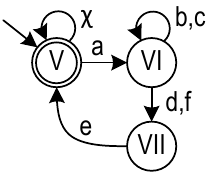}
   \label{fig:sect5:S1}
}
\hspace{13mm}
\subfloat[]{
   \includegraphics[scale=0.95] {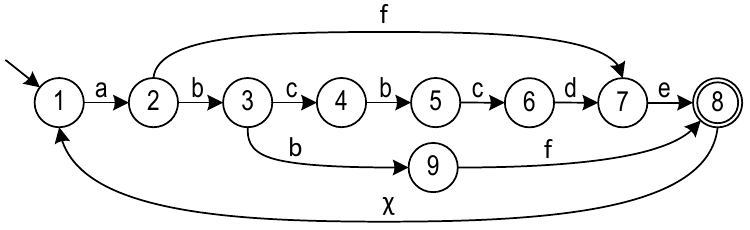}
   \label{fig:sect5:L3}
}
\hspace{10mm}
\subfloat[]{
\begin{minipage}[adjusting]{.2\linewidth}
\centering
\vspace{-16mm}
$
\footnotesize
\renewcommand\arraystretch{1.2}
\begin{array}{c|ccc}
& \mathrm{V} & \mathrm{VI} & \mathrm{VII} \\
\hline
\mathrm{V}   & 1 & 1 & 0 \\
\mathrm{VI}  & 0 & 2 & 2 \\
\mathrm{VII} & 1 & 0 & 0
\end{array}
$
\label{fig:sect5:S1:matrix}
\end{minipage}
}
\hspace{13mm}
\subfloat[]{
   \includegraphics[scale=0.95,trim=0 -2mm 0 3mm]{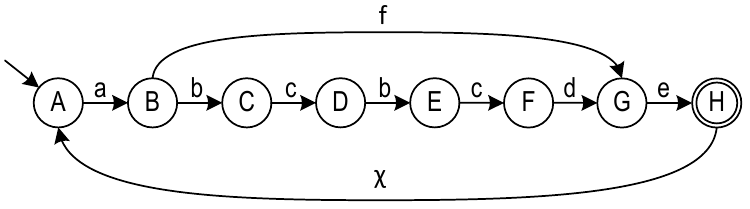}
   \label{fig:sect5:S1:cap:L3}
}
\vspace{-2mm}
\caption{Three DFAs and the adjacency matrix of the automaton in (a).}
\label{fig:automata:short-circuiting}
\vspace{-4mm}
\end{figure*}

\begin{example}
Consider automaton $\mathcal{S}_1$ in \autoref{fig:intro_example_aut_1} and log $\mathcal{L}_3$ in \autoref{fig:intro_logs}.
\autoref{fig:sect5:S1} and \autoref{fig:sect5:L3} show automata with languages
$(L(\mathcal{S}_1) \circ \{\sequence{\chi}\})^* \circ L(\mathcal{S}_1)$ and $(L(\mathcal{L}_3) \circ \{\sequence{\chi}\})^* \circ L(\mathcal{L}_3)$, respectively,
whereas \autoref{fig:sect5:S1:cap:L3} shows an automaton with language $((L(\mathcal{S}_1) \cap L(\mathcal{L}_3)) \circ \{\sequence{\chi}\})^* \circ (L(\mathcal{S}_1) \cap L(\mathcal{L}_3))$.
It is easy to see that given a DFA $B:=(Q,\Lambda,\delta,q_0,A)$, it holds that $L(B')=(L(B) \circ \{\sequence{\chi}\})^* \circ L(B)$, where $B':=(Q,\Lambda \cup \{\chi\},\delta \cup (A \times \{q_0\}),q_0,A)$.
Note that the automaton in \autoref{fig:sect5:S1} was obtained from the automaton in \autoref{fig:intro_example_aut_1} using this simple transformation and subsequent minimization~\cite{Hopcroft71}.
Such minimization is possible because any automaton with the language of interest, in this case $(L(\mathcal{S}_1) \circ \{\sequence{\chi}\})^* \circ L(\mathcal{S}_1)$, suffices.
\autoref{fig:sect5:S1:matrix} shows the adjacency matrix of the automaton in \autoref{fig:sect5:S1}. The Perron-Frobenius eigenvalue of this matrix is 2.521.
The Perron-Frobenius eigenvalues of the adjacency matrices of automata in \autoref{fig:sect5:L3} and \autoref{fig:sect5:S1:cap:L3} are 1.226 and 1.128, respectively.
Thus, it holds that
$\mathit{precision}_{\mathit{eig}^\bullet}(\mathcal{S}_1,\mathcal{L}_3)=0.447$ 
and $\mathit{recall}_{\mathit{eig}^\bullet}(\mathcal{S}_1,\mathcal{L}_3)=0.92$.

\autoref{tab:prec:and:recall:values} lists all the precision and recall values induced by $\mathit{eig}^\bullet$ for each of the three DFAs in
\figurenames~\ref{fig:intro_example_aut_1}--\ref{fig:intro_example_aut_3} w.r.t. every log in \autoref{fig:intro_logs}. 
The presented values obey all the properties discussed in \autoref{sec:framework_properties} and in the subsequent section.

\begin{table}[h]
\centering
\caption{Precision and recall values.}
\vspace{-3mm}
\begin{tabular}{c c c c} 
\toprule
Automaton & Log & {Precision} & {Recall} \\
\midrule
$\mathcal{S}_1$ & $\mathcal{L}_1$ & 0.442 & 1.0 \\
$\mathcal{S}_1$ & $\mathcal{L}_2$ & 0.506 & 1.0 \\
$\mathcal{S}_1$ & $\mathcal{L}_3$ & 0.447 & 0.92 \\
$\mathcal{S}_2$ & $\mathcal{L}_1$ & 0.661 & 0.897 \\
$\mathcal{S}_2$ & $\mathcal{L}_2$ & 0.661 & 0.784 \\
$\mathcal{S}_2$ & $\mathcal{L}_3$ & 0 & 0 \\
$\mathcal{S}_3$ & $\mathcal{L}_1$ & 0.881 & 0.897 \\
$\mathcal{S}_3$ & $\mathcal{L}_2$ & 0.881 & 0.784 \\
$\mathcal{S}_3$ & $\mathcal{L}_3$ & 0 & 0 \\
\bottomrule
\end{tabular}
\label{tab:prec:and:recall:values}
\end{table}
\vspace{-5mm}
\hfill\ensuremath{\lrcorner}
\end{example}

%%%%%%%%%%%%%%%%%%%%%%%%%%%%%%%%%%%%%%%%%%%%%%%%%%%%%%%%%%%%%%%%%%%%%%%%%%%%%%%
\subsection{Properties of Precision and Recall}
\label{sec:precision_recall_properties}
%%%%%%%%%%%%%%%%%%%%%%%%%%%%%%%%%%%%%%%%%%%%%%%%%%%%%%%%%%%%%%%%%%%%%%%%%%%%%%%

This section discusses important properties of precision and recall quotients, as defined in \autoref{sec:precision_recall_definition}, instantiated with some language measure, \eg the short-circuit eigenvalue measure ($\mathit{eig}^\bullet$) from \autoref{sec:framework_instantiations}.
As the precision and recall quotients are language quotients, they inherit all the properties discussed in \autoref{sec:framework_properties}.
Next, we present these properties for the precision quotient; note that one can derive corresponding properties for the recall quotient analogously.

\begin{propose}{Precision monotonicity over designs}{prp:precision:1}%{\quad\\}
Let $\mathcal{L}$ be an event log and let $B_1$ and $B_2$ be two DFAs such that $\lang{\mathcal{L}} \subset \lang{B_1}$ and $\lang{B_1} \subset \lang{B_2}$.
Let $m$ be a language measure over regular languages.
Then, it holds that $\mathit{precision}_m(B_2,\mathcal{L}) < \mathit{precision}_m(B_1,\mathcal{L})$.
\end{propose}

\noindent
\propositionname~\ref{prp:precision:1} follows from \lemmaname~\ref{lem:fixed:num} by accepting that $L_1=\lang{\mathcal{L}}$, $L_2=\lang{B_1}$, and $L_3=\lang{B_2}$.

\begin{propose}{Precision monotonicity over executions}{prp:precision:2}%{\quad\\}
Let $\mathcal{L}_1$ and $\mathcal{L}_2$ be two event logs and let $B$ be a DFA such that $\lang{\mathcal{L}_1} \subset \lang{\mathcal{L}_2}$ and $\lang{\mathcal{L}_2} \subset \lang{B}$.
Let $m$ be a language measure over regular languages.
Then, it holds that $\mathit{precision}_m(B,\mathcal{L}_1) < \mathit{precision}_m(B,\mathcal{L}_2)$.
\end{propose}

\noindent
\propositionname~\ref{prp:precision:2} follows from \lemmaname~\ref{lem:fixed:den} by accepting that $L_1=\lang{\mathcal{L}_1}$, $L_2=\lang{\mathcal{L}_2}$, and $L_3=\lang{B}$.

\begin{propose}{Precision minimality}{prp:precision:3}%{\quad\\}
Let $\mathcal{L}$ be an event log and let $B$ be a DFA such that $\lang{\mathcal{L}} \subset \lang{B}$ and $\lang{B} \subset \Phi^*$, $\Phi^* \subset \Xi^*$.
Let $m$ be a language measure over regular languages.
Then, it holds that $\mathit{precision}_m(\Phi^*,\mathcal{L}) < \mathit{precision}_m(B,\mathcal{L})$.
\end{propose}

\noindent
\propositionname~\ref{prp:precision:3} follows from \corolname~\ref{cor:fixed:num} by accepting that $L_1=\lang{\mathcal{L}}$ and $L_2=\lang{B}$.

Next, we present further properties specific for the precision and recall of a DFA w.r.t. a log.
First, precision and recall take values from the interval that contains zero and one.

\begin{propose}{Precision interval}{prp:precision:interval}%{\quad\\}
Given a log $\mathcal{L}$, a DFA $B$, such that $L(B) \neq \emptyset$, and a
language measure $m$ over regular languages, it holds that $0 \leq
\mathit{precision}_m(B,\mathcal{L}) \leq 1$.
\end{propose}

\noindent
\propositionname~\ref{prp:precision:interval} follows from \definitionname~\ref{def:pm:precision} and the fact that $m$ is a language measure over regular languages.
Indeed, it holds that $L(B) \cap L(\mathcal{L}) \subseteq L(B)$ and, thus, $m(L(B) \cap L(\mathcal{L})) \leq m(L(B))$; note that $m$ is an increasing measure.

\begin{propose}{Recall interval}{prp:recall:interval}{\quad\\}
Given a log $\mathcal{L}$, such that $L(\mathcal{L}) \neq \emptyset$, a DFA
$B$, and a language measure $m$ over regular languages, it holds that $0 \leq
\mathit{recall}_m(B,\mathcal{L}) \leq 1$.
\end{propose}

\noindent
\propositionname~\ref{prp:recall:interval} holds because of \definitionname~\ref{def:pm:recall}, the facts that $L(B) \cap L(\mathcal{L}) \subseteq L(\mathcal{L})$, and, again, because $m$ is an increasing measure.

Second, precision and recall equal to one when the languages of a given DFA and log are in the containment relation.

\begin{propose}{Maximal precision}{prp:max:precision}{\quad\\}
Given a log $\mathcal{L}$, a DFA $B$, such that $L(B) \neq \emptyset$, and a
language measure $m$ over regular languages,
$L(B) \subseteq L(\mathcal{L})$ \ifaof $\mathit{precision}_m(B,\mathcal{L})=1$.
\end{propose}

\noindent
If $L(B) \subseteq L(\mathcal{L})$, then it holds that $\mathit{precision}_m(B,\mathcal{L}) = {m(L(B))}/{m(L(B))}=1$.
Conversely, if $\mathit{precision}_m(B,\mathcal{L})=1$, then $m(L(B) \cap L(\mathcal{L})) = m(L(B))$.
Then, it holds that $L(B) \subseteq L(\mathcal{L})$.

\begin{propose}{Maximal recall}{prp:max:recall}{\quad\\}
Given a log $\mathcal{L}$, such that $L(\mathcal{L}) \neq \emptyset$, a DFA
$B$, and a language measure $m$ over regular languages, $L(\mathcal{L})
\subseteq L(B)$ \ifaof $\mathit{recall}_m(B,\mathcal{L})=1$.
\end{propose}

\noindent
The proof of \propositionname~\ref{prp:max:recall} follows the structure of the proof of \propositionname~\ref{prp:max:precision} but swaps the roles of the languages of $B$ and $\mathcal{L}$.

Third, precision and recall both equal to one \ifaof the languages of the DFA and log are identical.

\begin{corol}{Maximal precision and recall}{cor:max}{\quad\\}
Given a log $\mathcal{L}$, $L(\mathcal{L}) \neq \emptyset$, a DFA $B$, $L(B)
\neq \emptyset$, and a language measure $m$ over regular languages,
$L(B) \!=\! L(\mathcal{L})$ \ifaof $\mathit{precision}_m(B,\mathcal{L})\!=\!1$ and $\mathit{recall}_m(B,\mathcal{L})\!=\!1$.
\end{corol}

\noindent
\corolname~\ref{cor:max} follows immediately from \propositionname~\ref{prp:max:precision} and \propositionname~\ref{prp:max:recall}.

Finally, precision and recall equal to zero when the languages of the DFA and log do not overlap.

\begin{propose}{Minimal precision}{prp:min:precision}{\quad\\}
Given a log $\mathcal{L}$, a DFA $B$, such that  $L(B) \neq \emptyset$, and a
language measure $m$ over regular languages,
$L(B) \cap L(\mathcal{L}) = \emptyset$ \ifaof $\mathit{precision}_m(B,\mathcal{L})=0$.
\end{propose}

\noindent
If $L(B) \cap L(\mathcal{L}) = \emptyset$, then $\mathit{precision}_m = m(\emptyset) / m(B) = 0$, as $m$ starts at zero.
Conversely, if $\mathit{precision}_m(B,\mathcal{L})=0$, then $m(L(B) \cap L(\mathcal{L}))=0$.
Then, $L(B) \cap L(\mathcal{L})=\emptyset$ because $m$ starts at zero and is increasing.

\begin{propose}{Minimal recall}{prp:min:recall}{\quad\\}
Given a log $\mathcal{L}$, such that $L(\mathcal{L}) \neq \emptyset$, a DFA
$B$, and a language measure $m$ over regular languages,
$L(B) \cap L(\mathcal{L}) = \emptyset$ \ifaof $\mathit{recall}_m(B,\mathcal{L})=0$.
\end{propose}

\noindent
The proof of \propositionname~\ref{prp:min:recall} follows the structure of the proof of \propositionname~\ref{prp:max:recall} but swaps the roles of the languages of $B$ and $\mathcal{L}$.

%%%%%%%%%%%%%%%%%%%%%%%%%%%%%%%%%%%%%%%%%%%%%%%%%%%%%%%%%%%%%%%%%%%%%%%%%%%%%%%

%%%%%%%%%%%%%%%%%%%%%%%%%%%%%%%%%%%%%%%%%%%%%%%%%%%%%%%%%%%%%%%%%%%%%%%%%%%%%%%
\section{Implementation}
\label{sec:implementation}
%%%%%%%%%%%%%%%%%%%%%%%%%%%%%%%%%%%%%%%%%%%%%%%%%%%%%%%%%%%%%%%%%%%%%%%%%%%%%%%

The eigenvalue-based recall and precision measures for comparing specification and executions have been implemented and are publicly available.\footnote{The source code for computing the measures and for performing the experiments reported in~\autoref{sec:precision_recall_evaluation} is available at~\url{https://github.com/andreas-solti/eigen-measure}. In addition, as part of the jBPT library~\cite{Polyvyanyy2013a} (refer to~\url{https://github.com/jbpt/codebase}, \texttt{jbpt-pm} module), we develop and maintain a tool with a command-line interface to compute measures that compare specifications of dynamic systems. As of today, the tool supports the computation of measures presented in this work and the measures based on the partial matching of executions presented in~\cite{Polyvyanyy2019}.}
Algorithm~\ref{alg:precision:and:recall} summarizes the steps for computing the measures in pseudocode.
As input, the algorithm takes two specifications.
One specification describes a collection of ``retrieved'' executions, while the other captures ``relevant'' executions.
For example, in the context of process mining, one can see executions encoded in an event log as relevant, \ie those that encode valuable information, while executions captured by the model discovered from the event log as retrieved, \ie constructed from the event log.

{
\scriptsize
\begin{algorithm}
\caption{PrecisionAndRecallForSpecificationsOfDynamicSystems\label{alg:precision:and:recall}}
\KwIn{Two NFAs $\mathit{ret}$ and $\mathit{rel}$ describing retrieved and relevant executions, respectively.}
\KwOut{A pair $(\mathit{prec},\mathit{rec})$, where $\mathit{prec}$ and $\mathit{rec}$ are, respectively, precision and recall for $\mathit{ret}$ and $\mathit{rel}$.}
\BlankLine
\tcp{Construct deterministic versions of $\mathit{ret}$ and $\mathit{rel}$}
\leIf{$\neg \textrm{\emph{IsDeterministic}}(\mathit{ret})$}{$\mathit{dRet} \leftarrow \textrm{Determinize}(\mathit{ret})$}{$\mathit{dRet} \leftarrow \mathit{ret}$}
\leIf{$\neg \textrm{\emph{IsDeterministic}}(\mathit{rel})$}{$\mathit{dRel} \leftarrow \textrm{Determinize}(\mathit{rel})$}{$\mathit{dRel} \leftarrow \mathit{rel}$}
% minimize
$\mathit{mRet} \leftarrow \textrm{Minimize}(\mathit{dRet})$\tcc*[r]{Minimize automaton $\mathit{dRet}$}
$\mathit{mRel} \leftarrow \textrm{Minimize}(\mathit{dRel})$\tcc*[r]{Minimize automaton $\mathit{dRel}$}
% short-circuit
$\mathit{scRet} \leftarrow \textrm{ShortCircuit}(\mathit{mRet})$\tcc*[r]{Short-circuit automaton $\mathit{mRet}$}
$\mathit{scRel} \leftarrow \textrm{ShortCircuit}(\mathit{mRel})$\tcc*[r]{Short-circuit automaton $\mathit{mRel}$}
$\mathit{intersection}   \leftarrow \textrm{Intersection}(\mathit{mRet},\mathit{mRel})$\tcc*[r]{Construct intersection of $\mathit{mRet}$ and $\mathit{mRel}$}
$\mathit{scIntersection} \leftarrow \textrm{ShortCircuit}(\mathit{intersection})$\tcc*[r]{Short-circuit automaton $\mathit{intersection}$}
\tcp{Compute Perron-Frobenius eigenvalues of the adjacency matrices of automata}
$\mathit{eigRet} \leftarrow \textrm{PerronFrobenius}(\mathit{scRet})$\;
$\mathit{eigRel} \leftarrow \textrm{PerronFrobenius}(\mathit{scRel})$\;
$\mathit{eigIntersection} \leftarrow \textrm{PerronFrobenius}(\mathit{scIntersection})$\;
\Return{$(\frac{\mathit{eigInersection}}{\mathit{eigRet}},\frac{\mathit{eigInersection}}{\mathit{eigRel}})$}\;
\end{algorithm}
}

\noindent
Lines~2 and~3 of the algorithm ensure that $\mathit{dRet}$ and $\mathit{dRel}$ are \emph{deterministic versions} of $\mathit{ret}$ and $\mathit{rel}$, respectively; that is $\mathit{dRet}$ and $\mathit{dRel}$ are DFAs that recognise the languages $L(\mathit{ret})$ and $L(\mathit{rel})$, respectively.
Given a $\tau$-free NFA, a DFA that recognises the language of the NFA always exists~\cite{Hopcroft2007} and can be constructed using the Rabin-Scott powerset construction method~\cite{RabinS59}, which has the worst-case time complexity of $O(2^n)$ where $n$ is the number of states in the NFA~\cite{Moore71}.
However, in practice, the DFA constructed from the NFA has about as many states as the NFA, but often more transitions~\cite{Hopcroft2007}.
If an NFA is not $\tau$-free, one still can always construct a DFA that recognises the language of the NFA~\cite{Hopcroft2007}.
The construction extends the powerset construction to account for silent transitions.
Hence, at lines~2 and 3 of Algorithm~\ref{alg:precision:and:recall}, both functions IsDeterministic and Determinize take an NFA as input.
Function IsDeterministic returns \emph{true} if the input NFA is a DFA, and otherwise, returns \emph{false}, while function Determinize constructs and returns a deterministic version of the input NFA.
In our tools, function Determinize implements the extended version of the Rabin-Scott powerset construction from~\cite{Hopcroft2007} that constructs $\tau$-free automata.

Lines~4 and~5 of the algorithm construct the minimal versions of the DFAs $\mathit{dRet}$ and $\mathit{dRel}$, respectively.
For every DFA $A$, there exists a unique (up to isomorphism) DFA with a minimum number of states that recognises the language of $A$, called the \emph{minimal version} of $A$~\cite{Hopcroft2007}.
There exist several algorithms that, given a DFA, construct its minimal version.
For example, the worst-case time complexity of the algorithm by Hopcroft~\cite{Hopcroft71} is $O(nm \log{(m)})$, where $n$ is the number of states and $m$ is the size of the alphabet.
Function Minimize used at lines~4 and~5 takes a DFA as input and returns its minimal version.
In our tools, function Minimize implements the algorithm by Hopcroft.
Note that the minimization step can be skipped as the computation of the topological entropy does not require a DFA to be minimal (see \autoref{sec:framework_instantiations}).
While performing the scalability evaluation reported in \autoref{sec:scalability:evaluation}, we noticed that it is faster to minimize a DFA and then compute the Perron-Frobenius eigenvalue of its adjacency matrix than to compute the eigenvalue of the original DFA.
A detailed study of this phenomenon, despite important, is out of the scope of this paper.
The computation times reported in \tablename~\ref{tab:evaluation:scalability} include the minimization times.

Next, lines 6--9 of Algorithm~\ref{alg:precision:and:recall} construct \emph{short-circuit versions} of automata $\mathit{mRet}$ (line~6), $\mathit{mRel}$ (line~7), and the intersection $\mathit{intersection}$ of $\mathit{mRet}$ and $\mathit{mRel}$ (line~9) constructed at line~8.
Given a DFA, its short-circuit version is obtained by adding, for each accept state $a$, a short-circuit transition from $a$ to the start state.
All the short-circuit transitions have a dedicated short-circuit label ($\chi$), which is not in the set of labels of the original automaton.
It is easy to see that the short-circuit version of a DFA is deterministic.
Function ShortCircuit at lines 6, 7, and 9 of the algorithm implements the above construction.
The intersection of regular languages is a well-known operation in automata theory and is implemented in function Intersection at line~8 of the algorithm.
Its time complexity is $O(nm)$, where $n$ and $m$ are the numbers of states in the intersected automata~\cite{Hopcroft2007}.

Lines 11--13 of Algorithm~\ref{alg:precision:and:recall} compute Perron-Frobenius eigenvalues~\cite{Ceccherini-SilbersteinMS03} of the adjacency matrices of $\mathit{scRet}$ (line~11), $\mathit{scRel}$ (line~12), and $\mathit{scIntersection}$ (line~13).
Function PerronFrobenius used at lines 11--13 of the algorithm takes a DFA as input, constructs its adjacency matrix, and returns a largest eigenvalue of the matrix.
Note that, typically, the adjacency matrix of a DFA is rather sparse.
Thus, we can handle very large DFAs on personal computers and compute eigenvalues of their adjacency matrices with the help of memory-friendly sparse data structures.

We use the Java library \texttt{Matrix Toolkit Java (MTJ)} that relies on the low level numerical Fortran-based libraries in \texttt{ARPACK}~\cite{lehoucq1998arpack} to compute eigenvalues.
The available version of the \texttt{MTJ} library was capable of handling symmetric matrices in \texttt{ARPACK}.
Note that the adjacency matrices of automata are usually not symmetric.
Thus, we extended \texttt{MTJ} to expose \texttt{dnaupd} and \texttt{dneupd} routines of \texttt{ARPACK} to compute eigenvalues of general matrices.
Our extension for computing a largest eigenvalue of a non-symmetric matrix is publicly available.\footnote{The source code is available at \url{https://github.com/andreas-solti/matrix-toolkits-java} and the Maven Central Repository.}
% Computation details
The underlying technique for computing a largest eigenvalue of a matrix is called ``implicitly restarted Arnoldi iterations''~\cite{lehoucq1998arpack} and has a low polynomial time complexity.
Note that the numerical methods for computing an eigenvalue of a general matrix converge but provide no guarantees of convergence in a fixed number of iterations.
Thus, we set the threshold of 300\,000 as the maximum number of allowed iterations for practical reasons.
The author of the software package states that: ``The question of determining a shift strategy that leads to a provable rapid rate of convergence is a difficult problem that continues to be researched''~\cite{lehoucq2001implicitly}.
In the rare cases of non-convergence of the computation, we use the estimated value of the eigenvalue obtained at the end of the computation.
The proposed method is not tied to the ARPACK implementation for computing a largest eigenvalue of a matrix.
Thus, the use of more recent results in eigenvalue computation~\cite{Dookhitram2009} and parallel algorithms for eigenvalue computation~\cite{Imakura2018} may improve the overall performance of our tools.

Finally, line~14 of Algorithm~\ref{alg:precision:and:recall} returns a pair with the precision and recall for the input automata as the first and second element, respectively, computed as quotients of the corresponding eigenvalues.
To compute precision and recall of a system with respect to a given event log, \ie the quotients presented in~\autoref{sec:precision_recall}, one can invoke Algorithm~\ref{alg:precision:and:recall} with NFA $\mathit{ret}$ encoding the system and NFA $\mathit{rel}$ encoding the event log.
Indeed, one can also invoke Algorithm~\ref{alg:precision:and:recall} with two NFAs that describe two systems to estimate the language coverage.

The worst time complexity of Algorithm~\ref{alg:precision:and:recall} is dominated by the exponential worst time complexity of the NFA determinization at lines 2 and 3.
Note, however, that in practice the automata generated from software code or business process models are readily deterministic.
Besides, it is easy to construct a DFA that encodes an event log, for example, as a prefix tree. 
%%%%%%%%%%%%%%%%%%%%%%%%%%%%%%%%%%%%%%%%%%%%%%%%%%%%%%%%%%%%%%%%%%%%%%%%%%%%%%%
\section{Experimental Evaluation}
\label{sec:precision_recall_evaluation}
%%%%%%%%%%%%%%%%%%%%%%%%%%%%%%%%%%%%%%%%%%%%%%%%%%%%%%%%%%%%%%%%%%%%%%%%%%%%%%%

The goal of the evaluation reported in this section is to demonstrate that the proposed eigenvalue-based measures for comparing specifications and collections of executions of systems advance the state-of-the-art and can be readily applied in practice. This is achieved by answering the below research questions; Tables~\ref{tab:approaches:labelsandrefs:precision}~and~\ref{tab:approaches:labelsandrefs:recall} list the approaches considered in our comparative evaluation, they are discussed in detail in~\autoref{sec:related_work}. 

\begin{table*}[h]
\vspace{-2mm}
\centering
\subfloat{
\begin{minipage}[t]{0.56\textwidth}
\centering
\begin{scriptsize}
\begin{tabular}{@{}ll@{}}
\toprule
Short label              & Full name and reference                                                                       \\ \midrule
\textsf{advBehAppropriateness}    & Advanced behavioural appropriateness~\cite{RozinatA08IS}                                      \\[0.5ex]
\textsf{alignmentPrecision}       & Alignment-based precision~\cite{AdriansyahDA11EDOC}                                           \\[0.5ex]
\textsf{antiAlignPrecision}       & Anti-alignments precision~\cite{DongenCC16}                                                   \\[0.5ex]
\textsf{bestAlignPrecision}       & Best optimal-alignments precision~\cite{Adriansyah.etal/ISEM2015:OneAndBestAlignPrecision}    \\[0.5ex]
\textsf{negativeEventPrecision}   & AGNEs specificity~\cite{Goedertier.etal/JoMLR2009:ProcessDiscoveryArtificialNegativeEvents}   \\[0.5ex]
\textsf{oneAlignPrecision}        & One optimal-alignment precision~\cite{Adriansyah.etal/ISEM2015:OneAndBestAlignPrecision}      \\[0.5ex]
\textsf{precisionEig}             & Eigenvalue-based precision (this paper)                                                       \\[0.5ex]
\textsf{precisionETC}             & ETC precision~\cite{Munoz-Gama.Carmona/CIDM2011:ETConformanceExt:StabilityConfidenceSeverity} \\[0.5ex]
\textsf{projectedPrecision}       & PCC precision~\cite{Leemans2016}                                                              \\[0.5ex]
\textsf{simpleBehAppropriateness} & Simple behavioural appropriateness~\cite{RozinatA08IS}                                        \\ \bottomrule
\end{tabular}

\end{scriptsize}
\vspace{2mm}
\caption{Precision measures.}
\label{tab:approaches:labelsandrefs:precision}
\end{minipage}
}
\subfloat{
\begin{minipage}[t]{0.42\textwidth}
\centering
\begin{scriptsize}
\begin{tabular}{@{}ll@{}}
\toprule
Short label         & Full name and reference                                                                \\ \midrule
\textsf{alignmentFitness}    & Alignment-based fitness~\cite{AalstAD12WIDM}                                           \\[0.5ex]
\textsf{negativeEventRecall} & AGNEs recall~\cite{Goedertier.etal/JoMLR2009:ProcessDiscoveryArtificialNegativeEvents} \\[0.5ex]
\textsf{tokenBasedFitness}   & Token-based fitness~\cite{RozinatA08IS}                                                \\[0.5ex]
\textsf{parsingMeasure}      & Continued parsing measure~\cite{Weijters2006}                                          \\[0.5ex]
\textsf{projectedRecall}     & PCC recall~\cite{Leemans2016}                                                          \\[0.5ex]
\textsf{properCompletion}    & Proper completion~\cite{RozinatA08IS}                                                  \\[0.5ex]
\textsf{recallEig}           & Eigenvalue-based recall (this paper)                                                   \\ \bottomrule
\end{tabular}

\end{scriptsize}
\caption{Fitness (recall) measures.}
\label{tab:approaches:labelsandrefs:recall}
\end{minipage}
}
\vspace{-10mm}
\end{table*}

\medskip
\begin{compactitem}
  \item[RQ1:] Are state-of-the-art precision and recall measures for process mining monotone?
	\item[RQ2:] Are the eigenvalue-based quotients applicable for comparing the behaviours of two systems?
	\item[RQ3:] Is the computation of eigenvalue-based quotients feasible for practical applications?
\end{compactitem}
\medskip

% RQ 1
To answer RQ1, we studied which state-of-the-art precision and recall measures in process mining fulfill Lemmata~\ref{lem:fixed:num} and~\ref{lem:fixed:den}; we assume for this evaluation that specifications used in comparisons describe bounded systems, \ie systems that induce finite collections of reachable states.
In~\autoref{subsec:monotonicity:experiments}, we give a negative answer to RQ1 for the state-of-the-art precision measures; for each evaluated measure, we present at least one example that violates the property of monotonicity.

% RQ 2
To answer RQ2, we compare specifications of software systems studied in~\cite{Gabel.Su/FSE2008:Javert} with specifications discovered from their sample executions by computing language coverage values.
We argue that such comparisons can be used to reveal insights on the quality of automated algorithms for discovering system specifications, see~\autoref{sec:comparing:specifications}, within reasonable time bounds, see~\autoref{sec:scalability:evaluation}.

% RQ 3
Finally, to answer RQ3, we measured the wall-clock time of computing eigenvalue-based quotients for logs and specifications from real-world and synthetic datasets.
First, we computed the eigenvalue-based precision and recall for fifteen real-world event logs and specifications automatically discovered from these logs. 
The logs are publicly available\footnote{Event logs are published at: \url{https://data.4tu.nl/repository/collection:event_logs_real}} and encode executions of real-world IT systems executing 
business processes with real customers.
We observed that for most input pairs, each composed of a specification and a log, computations of both precision and recall measures are accomplished within ten minutes, and often much faster.
Second, we report the values and computation times of the eigenvalue-based quotients for inputs taken from a collection of real-world specifications, corresponding collections of sample executions of controlled sizes, and specifications automatically discovered from the sample executions.\footnote{The collection of the specifications used in the experiment (which also includes specifications that due to space considerations were not discussed in this article), as well as tools to generate sample executions and discovered models, and scripts to reproduce the experiment is available at: \url{https://github.com/andreas-solti/monotone-precision}.}

To perform the experiments, we used our implementation of the eigenvalue-based quotients, described in \autoref{sec:implementation}, and relied on the
Comprehensive Benchmark Framework (CoBeFra)~\cite{DBLP:conf/cidm/BrouckeWVB13} to compute other precision and recall measures.

%%%%%%%%%%%%%%%%%%%%%%%%%%%%%%%%%%%%%%%%%%%%%%%%%%%%%%%%%%%%%%%%%%%%%%%%%%%%%%%
\subsection{Comparing Executions and Specifications: Monotonicity of Precision and Recall}
\label{subsec:monotonicity:experiments}
%%%%%%%%%%%%%%%%%%%%%%%%%%%%%%%%%%%%%%%%%%%%%%%%%%%%%%%%%%%%%%%%%%%%%%%%%%%%%%%

For a given log, a monotonic precision measure should \emph{always} decrease when fresh behaviour is added to the specification. 
Conversely, a monotonic precision measure should \emph{always} increase when the excess behaviour is removed from the specification. 
We use three experimental setups that address these two phenomena and show that all the precision measures listed in \autoref{tab:approaches:labelsandrefs:precision}, except the eigenvalue-based measure, fail to demonstrate monotonicity for at least one of the setups.
We also compare and discuss precision and recall measurements computed for several real-world and synthetic datasets using the techniques listed in Tables~\ref{tab:approaches:labelsandrefs:precision} and~\ref{tab:approaches:labelsandrefs:recall}.

%%%%%%%%%%%%%%%%%%%%%%%%%%%%%%%%%%%%%%%%%%%%%%%%%%%%%%%%%%%%%%%%%%%%%%%%%%%%%%%
\subsubsection{Monotonicity of Precision Measures}
%%%%%%%%%%%%%%%%%%%%%%%%%%%%%%%%%%%%%%%%%%%%%%%%%%%%%%%%%%%%%%%%%%%%%%%%%%%%%%%

Next, we discuss the results of the three experimental setups that aim to study the monotonicity of the existing precision measures.
In the first setup, given the log that contains traces with up to two events \texttt{a} before event \texttt{b} and a perfectly fitting specification, we gradually add behaviour to the specification.
We use regular expressions\footnote{Notation \texttt{a\{<min>,<max>\}} is a short-hand for enumerating the minimal and maximal number of repetitions of symbol \texttt{a}. For example, \texttt{a\{0,2\}}$\,\circ\,$\texttt{b} specifies the language \texttt{\{<b>,<a,b>,<a,a,b>\}}.} to describe the languages of the log and specifications:

\smallskip
\begin{compactitem}
	\item[$L$] is the \emph{log} with the language \texttt{a\{0,2\}}$\,\circ\,$\texttt{b};
	\item[$M_x$] are the \emph{specifications} with language \texttt{a\{0,x\}$\,\circ\,$b}, $\texttt{x} \in [2\, ..\,20]$;
	\item[$M_\star$] is the \emph{specification} with language \texttt{a$^*\circ\,$b}.
\end{compactitem}
\smallskip

\begin{figure}[b]
\vspace{-3mm}
	\begin{center}
		\includegraphics[width=.85\columnwidth]{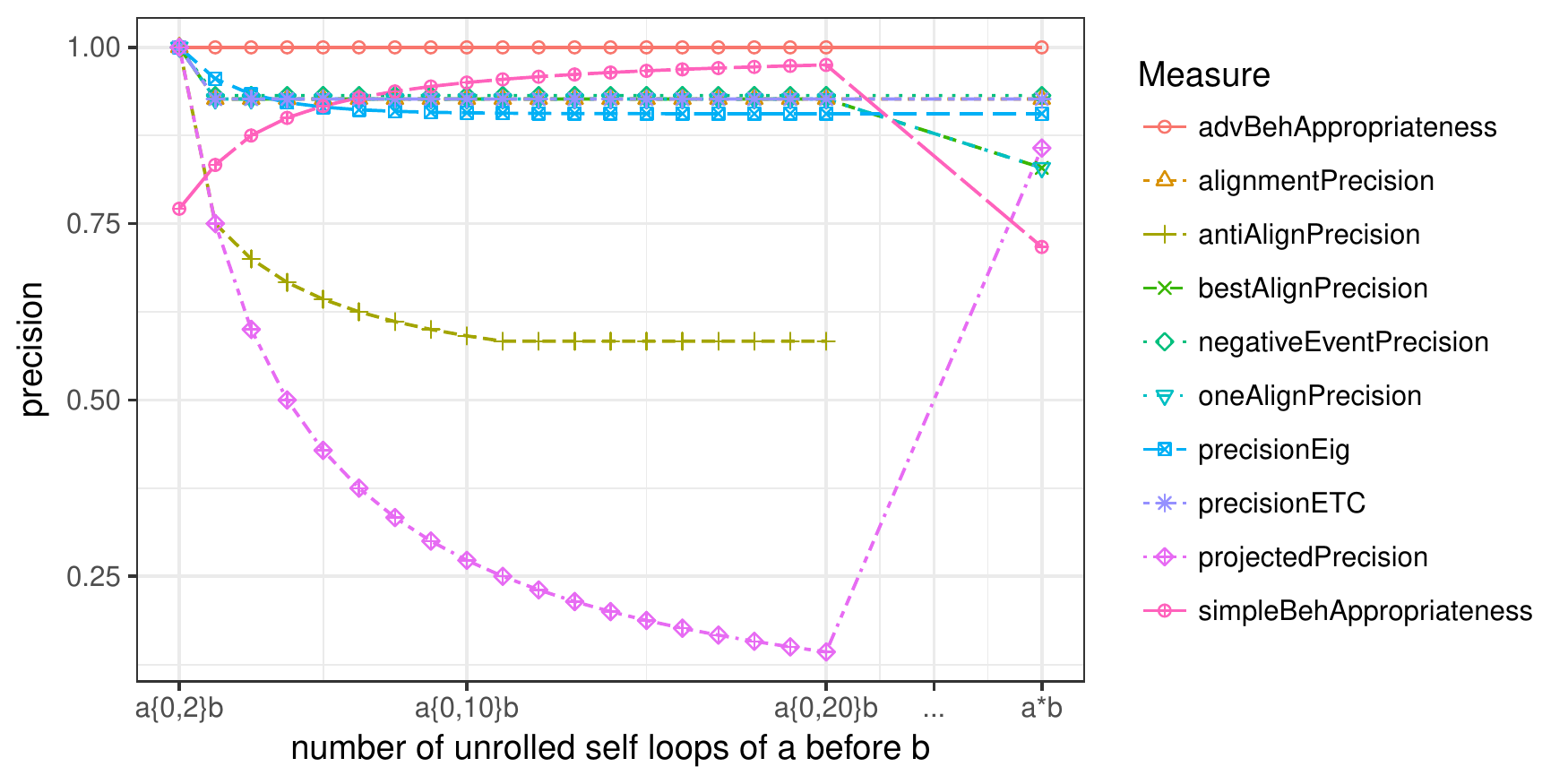}
		\vspace{-3mm}
		\caption{Increasing number of optional \texttt{a}'s before \texttt{b}. Starting with up to two \texttt{a}'s before \texttt{b}, stepwise allow more \texttt{a}'s up to the closure that allows an arbitrary number of \texttt{a}'s before \texttt{b}.}
		\label{fig:precision_recall_evaluation_unrolling_precision}
		\vspace{-2mm}
	\end{center}
\end{figure}

\autoref{fig:precision_recall_evaluation_unrolling_precision} shows the values of various precision measures on the y-axis, plotted for the different
specification languages reported on the x-axis: from $0$--$2$ possible repetitions of \texttt{a} before \texttt{b} up to $0$--$20$ repetitions. 
The last measurement on the end of the x-axis denotes the precision w.r.t. the specification that accepts an arbitrary number (i.e., $0$--$\infty$) of \texttt{a}'s before \texttt{b}, that is \texttt{a$^*\circ\,$b}. 
The values were recorded only if they were computed under the threshold of ten minutes.

The simple behavioural appropriateness measure~\cite{RozinatA08IS} shows a 
trend opposite to the other measures, as the precision values increase for more 
permissive specifications. 
Advanced behavioural appropriateness~\cite{RozinatA08IS} fails to recognise the ``growth'' of the specifications' languages. 
The anti-alignments precision~\cite{DongenCC16} demonstrates the correct trend, but has been unable to compute the precision for the \texttt{a$^*\circ\,$b} specification within the threshold time. 
PCC precision~\cite{Leemans2016} is strictly monotone in the region between up to 2 and up to 20 \texttt{a}'s before \texttt{b}, but violates the monotonicity in the step from the \texttt{a\{0,20\}$\,\circ\,$b} to \texttt{a$^*\circ\,$b} specification. 
The other measures show a similar trend starting at 1.00 for the perfectly fitting specification and decreasing but stabilizing quickly.
These measures, however, do not distinguish between the specifications \texttt{a\{0,y\}$\,\circ\,$b}, where $y \in [3\, ..\, 20]$. 
Our eigenvalue-based precision measure shows a steady stabilizing decline, \ie the more possible repetitions of \texttt{a} before \texttt{b} are allowed by the specification the smaller the precision value is.

Besides iteration, parallelism, captured via interleavings of actions, is another dimension that we investigate. 
We vary the number of permutations over a fixed alphabet of size 5. 
Each word is constructed by drawing five out of five available symbols without replacement, where the order matters. 
Hence, there are $5! = 120$ distinct permutations of symbols, i.e., 120 distinct words. 
A process specification that allows parallel execution of five activities also permits exactly 120 different executions. 
We expect a specification that enumerates all 120 permutations to be equally precise as another specification that uses a parallel building block that says that the same five activities can be done in any order. Thus, the second experimental setup uses the following log and specifications.

\smallskip
\begin{compactitem}
\item[$L_{5||}$] is the \emph{log} with language $\texttt{\{abcde,\,abced,\,abdec,\,abdce,\,abecd\}}$;
	\item[$M_{x||}$] is the collection of \emph{specifications} such that each specification describes all the five traces in $L_{5||}$, and further permutations of symbols \texttt{a}, \texttt{b}, \texttt{c}, \texttt{d}, and \texttt{e}, such that specification $M_{x||}$, $5 \leq x \leq 120$, describes $x$ distinct permutations, and for all $5 \leq x < y \leq 120$ it holds that $M_{y||}$ describes all the permutations described by $M_{x||}$;
	\item[$M_{||}$] is the \emph{specification} with the language of all 120 permutations of symbols \texttt{a}, \texttt{b}, \texttt{c}, \texttt{d}, and \texttt{e} implemented as a parallel block of five simultaneously enabled activities.
\end{compactitem}
\smallskip

\begin{figure}[t]
\vspace{-3mm}
	\begin{center}
		\includegraphics[width=.85\columnwidth]{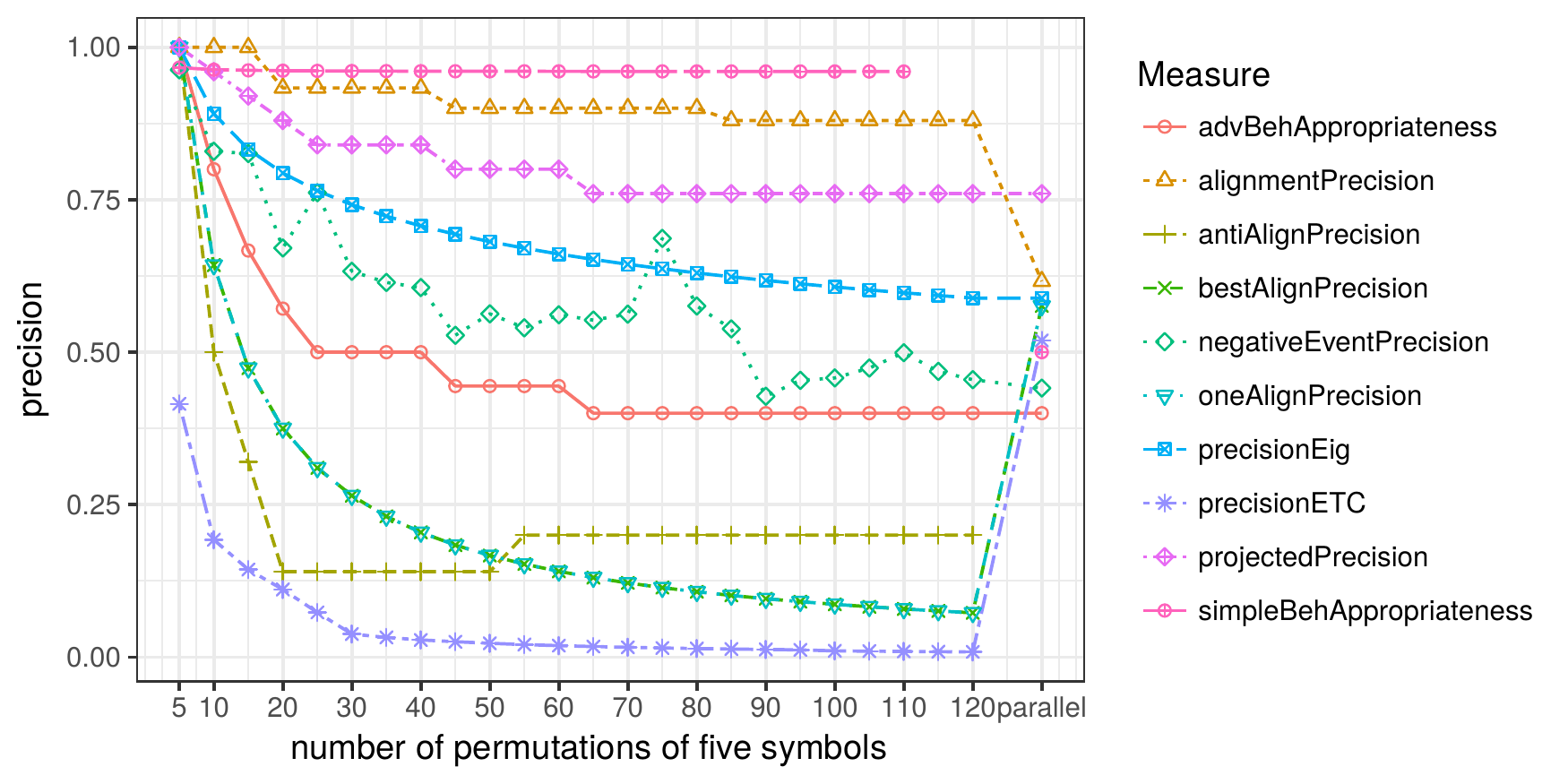}
\vspace{-3mm}
		\caption{The trends of precision measures for the log $\mathit{L}$ w.r.t. specifications that allow permutations of the same five symbols, including the precision of all the explicit $5! = 120$ permutations, and the precision of the language equivalent parallel specification with 120 implicitly allowed permutations.}
		\label{fig:precision_permutations}
	\end{center}
	\vspace{-4mm}
\end{figure}

Most existing precision measures listed in \autoref{tab:approaches:labelsandrefs:precision} show decreasing trends for log $L_{5||}$ and the collection of specifications $M_{x||}$, $5 \leq x \leq 120$, as it can be noticed in \autoref{fig:precision_permutations}. 
However, the specification that explicitly encodes all the 120 permutations often has a different precision value than the specification with five activities in 
parallel, although these two specifications describe the same language. 
Note that only three measures reported the same precision values for both these specifications, namely advanced behavioural appropriateness~\cite{RozinatA08IS}, PCC precision~\cite{Leemans2016}, and our eigenvalue-based measure.
The monotonicity in the second experimental setup is violated by 
the ETC precision~\cite{Munoz-Gama.Carmona/CIDM2011:ETConformanceExt:StabilityConfidenceSeverity},
one optimal-alignment precision~\cite{Adriansyah.etal/ISEM2015:OneAndBestAlignPrecision}, 
best optimal-alignments precision~\cite{Adriansyah.etal/ISEM2015:OneAndBestAlignPrecision}, 
anti-alignments precision~\cite{DongenCC16}, and 
AGNEs specificity~\cite{Goedertier.etal/JoMLR2009:ProcessDiscoveryArtificialNegativeEvents}.
We were unable to compute anti-alignments precision~\cite{DongenCC16} for the fully parallel specification within ten minutes. 
Also, we were unable, using available tools and the ten minutes time threshold, to compute simple behavioural appropriateness~\cite{RozinatA08IS} for the specifications that explicitly capture more than 110 permutations. 
Note that the value of simple behavioural appropriateness drops significantly for the parallel specification.

\begin{figure}[t]
	\vspace{-1mm}
	\begin{center}
		\includegraphics[width=.85\columnwidth]{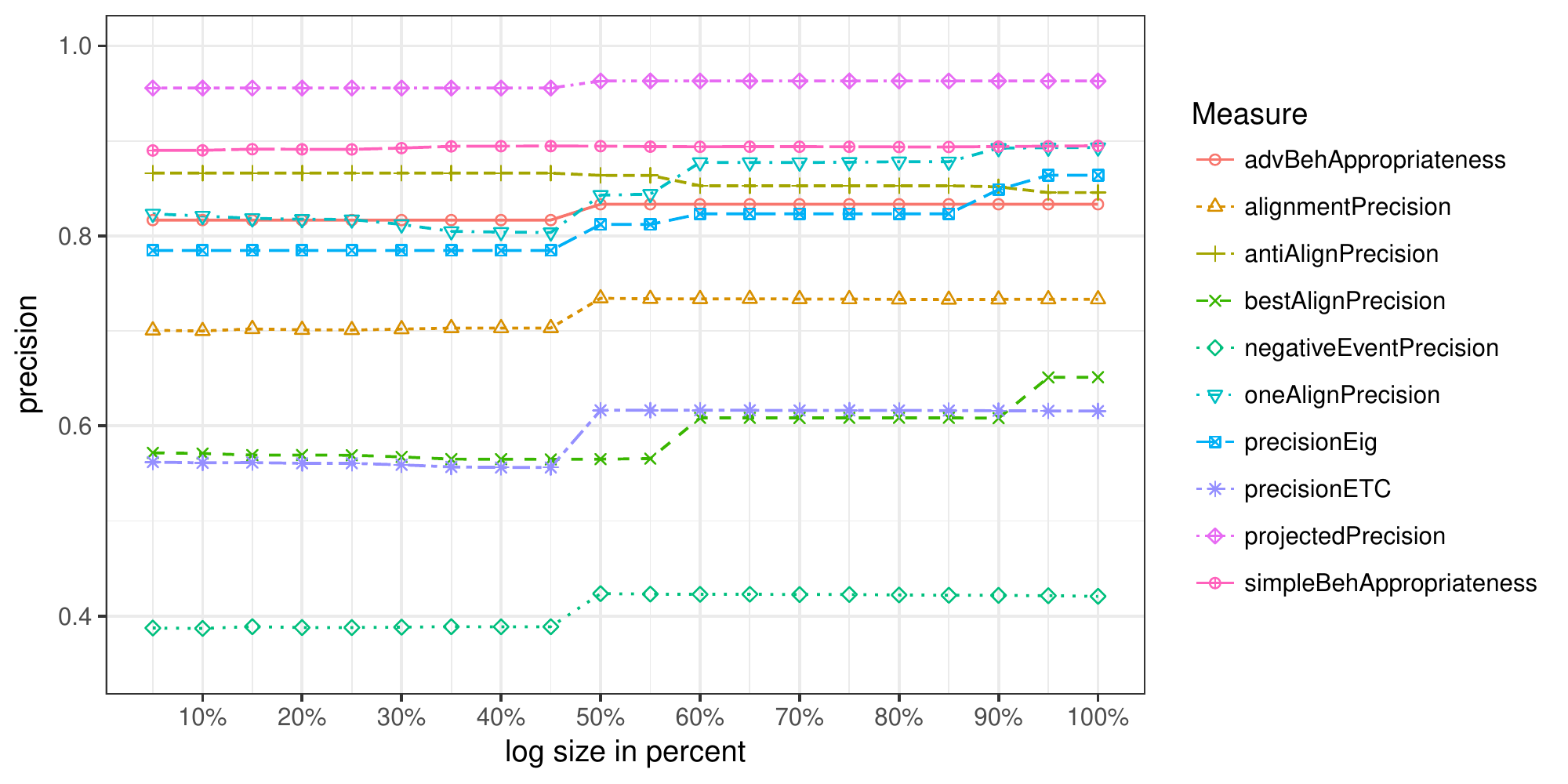}
		\vspace{-3mm}
		\caption{Expected increase in precision with more traces in the log of the main process (A) in \cite{bpic2012Data}.}
		\label{fig:precision_percent}
	\end{center}
	\vspace{-4mm}
\end{figure}

In the third experimental setup, we use the real-world log of the BPI Challenge 2012~\cite{bpic2012Data}. 
We discover a specification $M$ that can replay all the traces in the log using Inductive Miner~\cite{Leemans2016} with the infrequent and noise threshold parameters set to $0$. 
Then, we select five percent of random traces from the log to obtain sub-log $L_{5\%}$ and compute the precision of the specification w.r.t. the sub-log. 
We repeat this process for other sub-logs, each obtained by adding an additional five percent of random traces from the original log, such that $L_{5\%} \subset L_{10\%} \subset \ldots \subset L_{100\%}$; $L_{x\%}$ is a sub-log that contains $x$ percent of traces from the BPI Challenge 2012 log.
Because the specification fits the original log perfectly, it holds that $M$ describes all the traces in all the studied sub-logs. 
The measured precision values are reported in \autoref{fig:precision_percent}.
Note that when we increase the number of traces in the sub-log, new behaviour is not necessarily added. 
At each step, we can potentially end up adding only traces that the previous sub-log already contains.
To make the results accessible, we also created the plot shown in \autoref{fig:precision_percent_diff}. 
It depicts the differences between two consecutive precision values, \eg if the precision value increased by 0.1 when adding five percent of traces, we add a mark at 0.1. 
For a monotonically increasing measure, one should observe only non-negative differences; negative differences are emphasized with red triangles in the figure.

Only three precision measures demonstrate monotonicity for the third experimental setup.
These measures are advanced behavioural appropriateness~\cite{RozinatA08IS}, PCC precision~\cite{Leemans2016}, and our eigenvalue-based precision. 
Some negative values are due to the non-deterministic nature of corresponding precision measures, as discussed in~\cite{TaxLSFA17}. 
Also, there is a systematic error in the anti-alignments precision values~\cite{DongenCC16} that show an unexpected downward trend, despite the fact that the specification is fixed and the number of traces in the sub-logs increases in this experimental setup.
Note that the eigenvalue-based precision measure, as guaranteed by its properties, successfully recognises all the four changes in the behaviours encoded in the sub-logs.

To conclude, for each of the precision measures listed in \autoref{tab:approaches:labelsandrefs:precision} except the eigenvalue-based precision, we were able to construct at least one example that violates the monotonicity property captured in Lemmata~\ref{lem:fixed:num} and~\ref{lem:fixed:den}.

\begin{figure}[t]
  \vspace{-1mm}
	\begin{center}
		\includegraphics[width=.85\columnwidth]{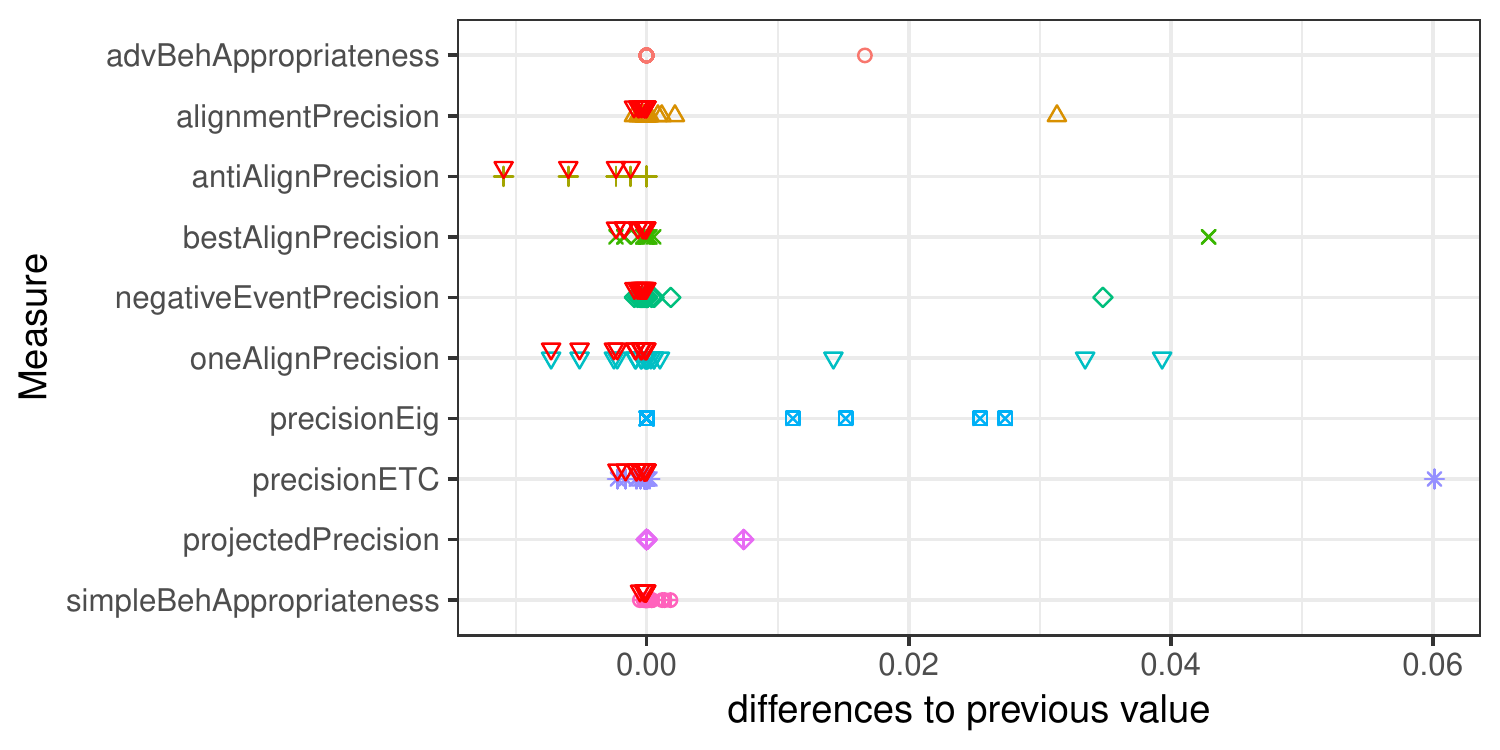}
		\vspace{-3mm}
		\caption{Each dot represents the relative increase or decrease in \autoref{fig:precision_percent} at each subsequent measurement step as the size of the log increases; red triangles encode relative decrease.}
		\label{fig:precision_percent_diff}
	\end{center}
	\vspace{-4mm}
\end{figure}

%%%%%%%%%%%%%%%%%%%%%%%%%%%%%%%%%%%%%%%%%%%%%%%%%%%%%%%%%%%%%%%%%%%%%%%%%%%%%%%
\subsubsection{Monotonicity of Recall Measures}
%%%%%%%%%%%%%%%%%%%%%%%%%%%%%%%%%%%%%%%%%%%%%%%%%%%%%%%%%%%%%%%%%%%%%%%%%%%%%%%

The recall of a specification w.r.t.\ a log is defined as the fraction of a measurement of the shared behaviour by a measurement of the behaviour in the log. 
In this case, both measurements capture finite behaviour, which makes the problem of computing the fraction less challenging than in the case of measuring precision.

\begin{figure}[b]
\vspace{-4mm}
	\begin{center}
		\includegraphics[width=.85\columnwidth]{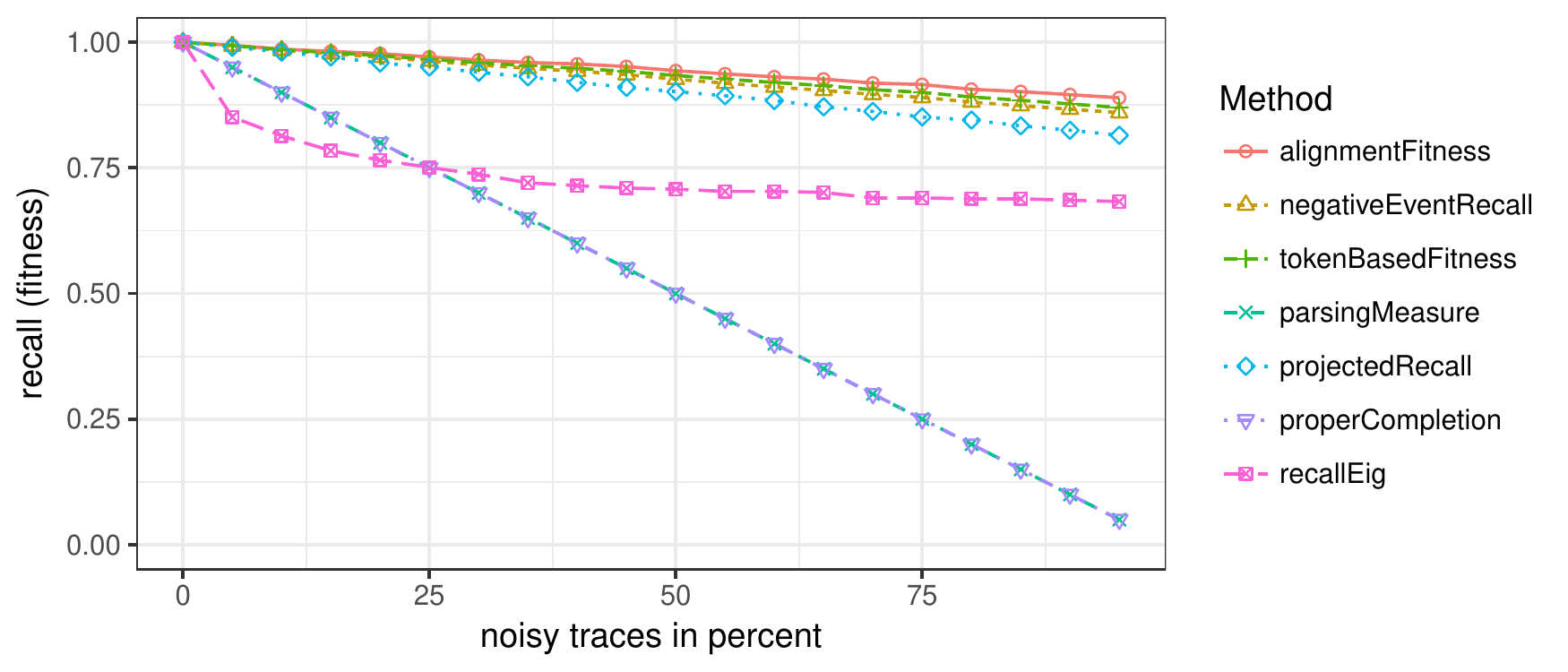}
		\vspace{-3mm}
		\caption{Recall measures for a sequential specification and increasing amount of noise in a fitting log.}
		\label{fig:recall_noise}
	\end{center}
	\vspace{-1mm}
\end{figure}

\autoref{fig:recall_noise} plots recall values for the measures listed in \autoref{tab:approaches:labelsandrefs:recall}.
The values were obtained in the following experimental setup. 
Given a sequential specification of ten activities and a fitting log with no noise, we start increasing the amount of noisy traces in the log. 
Here, noise is defined as removing, adding, or swapping events in the log, and the percentage shown on the x-axis reflects the relative number of traces affected by noise. 
Continued parsing measure~\cite{Weijters2006} and proper completion~\cite{RozinatA08IS} simply count the fraction of traces that are entirely fitting.
For example, continued parsing measure is based on a binary decision for each trace. 
Hence, small deviations between a trace and execution have the same impact on the measured value as significant differences.
In contrast, alignment-based fitness~\cite{AalstAD12WIDM}, AGNEs recall~\cite{Goedertier.etal/JoMLR2009:ProcessDiscoveryArtificialNegativeEvents}, and token-based fitness~\cite{RozinatA08IS} are sensitive to small discrepancies in traces in the log and executions of the specification, \ie they penalize minor deviations between traces and executions only slightly. 

The eigenvalue-based measures depend on the ``sizes'' of the languages of the compared log and specification. 
The above experimental setup shows that the behaviour in the log changes significantly with the insertion of first noisy traces, when the noise level is low. 
This leads to a rapid drop in recall as, indeed, the specification fails to capture the fresh behaviour, but only captures its deterministic sequential part. 
The increase in the number of noisy traces does not change the behaviour of the log at higher noise levels that much, as the probability that a new noisy trace has already been seen increases with the number of noisy traces.
In contrast, the other measures show a linear trend, as they do not take into account the \emph{size of the behaviour}, but ``count'' the \emph{number of fitting traces} w.r.t. the size of the log. 
As a consequence, traditional approaches treat the two cases listed in \autoref{tab:recall} equivalently, while our measure judges the recall for the situation described in the first row lower than that for the situation described in the second row, as the variance in the log is lower even though it has the same number of deviating traces.

\begin{table}[h]
\vspace{-1mm}
\centering
\caption{Precision and recall for specification that describes one execution $\sequence{\texttt{a},\texttt{b},\texttt{c}}$ and two logs $L_{\mathit{abc}(d|e)?}$ and $L_{\mathit{abc}(d)?}$. 
Log $L_{\mathit{abc}(d|e)?}$ consists of three traces: $\sequence{\texttt{a},\texttt{b},\texttt{c}}$, $\sequence{\texttt{a},\texttt{b},\texttt{c}, \texttt{d}}$, and $\sequence{\texttt{a},\texttt{b},\texttt{c}, \texttt{e}}$.
Log $L_{\mathit{abc}(d)?}$ consists of five traces: three occurrences of trace $\sequence{\texttt{a},\texttt{b},\texttt{c}}$ and two occurrences of trace $\sequence{\texttt{a},\texttt{b},\texttt{c},\texttt{d}}$.
}
	\vspace{-2mm}
	\label{tab:recall}
	\begin{tabular}{llll}
	\toprule
	Specification & Log & Precision & Recall\\
	\midrule
	$S_{\mathit{abc}}$ & $L_{\mathit{abc}(d|e)?}$ & 1.0 & 0.789 \\
	$S_{\mathit{abc}}$ & $L_{\mathit{abc}(d)?}$ & 1.0 & 0.856 \\
	\bottomrule
	\end{tabular}
	\vspace{-2mm}
\end{table}

While the amount of noisy traces increases linearly in this experiment, we are interested in the behaviour that is in both specification and log versus the behaviour in the log only. 
Our eigenvalue-based recall captures this non-linearity in the behaviour of the log.
Thus, we conclude that if one is interested in the measure of how much behaviour of a log is captured in a specification, our measure is more suitable.
However, if one is interested only in the fitting part of the log and does not need to distinguish between different deviations, the traditional fitness/recall measures are preferable. 
Latter linearly capture a decreasing number of fitting traces w.r.t. a given specification.

%%%%%%%%%%%%%%%%%%%%%%%%%%%%%%%%%%%%%%%%%%%%%%%%%%%%%%%%%%%%%%%%%%%%%%%%%%%%%%%%
\subsection{Comparing Specifications: Coverage}
\label{sec:comparing:specifications}
%%%%%%%%%%%%%%%%%%%%%%%%%%%%%%%%%%%%%%%%%%%%%%%%%%%%%%%%%%%%%%%%%%%%%%%%%%%%%%%%

The language coverage measure for two systems $\mathcal{S}_x$ and $\mathcal{S}_y$ was introduced in \autoref{sec:back} using the language cardinality measure.
To overcome the problem of measuring infinite languages of systems, according to the framework presented in \autoref{sec:framework_definition}, we instantiate the language coverage quotient with the short-circuit measure induced by the eigenvalue measure, \ie $\mathit{eig}^\bullet$, as follows:
$$
\funcCall{\mathit{coverage}_{\mathit{eig}^\bullet}}{\mathcal{S}_x,\mathcal{S}_y} := 
\frac{\funcCall{{\mathit{eig}^\bullet}}{\lang{\mathcal{S}_x} \,\cap\, \lang{\mathcal{S}_y}}}{\funcCall{{\mathit{eig}^\bullet}}{\lang{\mathcal{S}_x}}}.
$$

We demonstrate the use of $\smash{\mathit{coverage}_{\mathit{eig}^\bullet}}$ for measuring how well the behaviour of one software system \emph{covers} the behaviour of some other software system using the following experiment.
Given a specification of a software system $\mathcal{S}$, we simulate a collection of its executions, \ie a log, $\mathcal{L}$.
Next, we discover a specification $\mathcal{D}$ from $\mathcal{L}$.
Finally, we compute
$\smash{\mathit{coverage}_{\mathit{eig}^\bullet}(\mathcal{S},\mathcal{D})}$ and $\smash{\mathit{coverage}_{\mathit{eig}^\bullet}(\mathcal{D},\mathcal{S})}$.

\begin{figure*}[h]
	\vspace{-3mm}
	\centering
	\subfloat[Socket API~\cite{Ammons2002}]{
		\includegraphics[scale=.92, trim=0 -7mm 0 0] {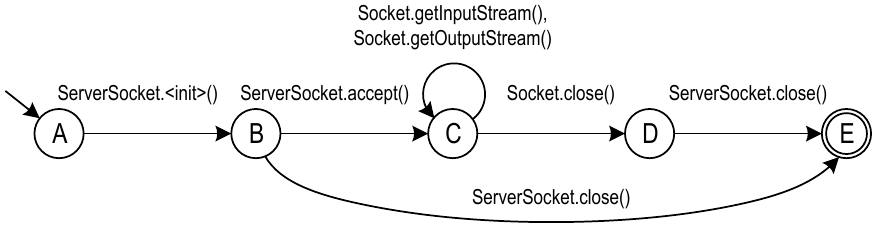}
		\label{fig:two:software:systems:5}
	}
	\hspace{2mm}
	\subfloat[Hibernate~\cite{Gabel.Su/FSE2008:Javert}]{
		\includegraphics[scale=.92] {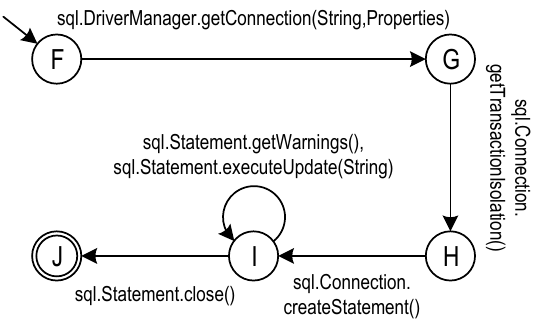}
		\label{fig:two:software:systems:13a}
	}
	\vspace{-3mm}
	\caption{Two specifications of software systems captured as DFAs.}
	\label{fig:two:software:systems}
	\vspace{-2mm}
\end{figure*}

\autoref{fig:comparing:specifications:5} plots language coverage values between the specification in \autoref{fig:two:software:systems:5} and the discovered specifications from various collections of its simulated traces.
Similarly, \autoref{fig:comparing:specifications:13a} plots language coverage values between the specification in \autoref{fig:two:software:systems:13a} and the corresponding discovered specifications.
Both specifications in \autoref{fig:two:software:systems} were used in~\cite{Gabel.Su/FSE2008:Javert} in the context of evaluating an algorithm for automatically discovering specifications from executions of software systems.
In particular, \autoref{fig:two:software:systems:5} captures the Socket API reproduced from~\cite{Ammons2002}, while \autoref{fig:two:software:systems:13a} describes a part of Hibernate functionality (see~\cite{Gabel.Su/FSE2008:Javert} for details).
The specifications were discovered using Inductive Miner~\cite{Leemans2016} with the infrequent and noise threshold parameters set to $0.2$.

In both plots in \autoref{fig:comparing:specifications}, each blue circle encodes $\smash{\funcCall{\mathit{coverage}_{\mathit{eig}^\bullet}}{\mathcal{S},\mathcal{D}}}$ for the corresponding specification $\mathcal{S}$ from \autoref{fig:two:software:systems} and some specification $\mathcal{D}$ automatically discovered from a log whose size (measured as the number of, not necessarily distinct, traces) is reflected on the x-axis; note that the x-axis uses a logarithmic scale.
Similarly, red diamonds encode the corresponding $\smash{\funcCall{\mathit{coverage}_{\mathit{eig}^\bullet}}{\mathcal{D},\mathcal{S}}}$ values.
Note that the simulated logs are in the subset relation, \ie each log $\mathcal{L}'$ is strictly contained in every log $\mathcal{L}''$ that has more traces than $\mathcal{L}'$.

\begin{figure*}[h]
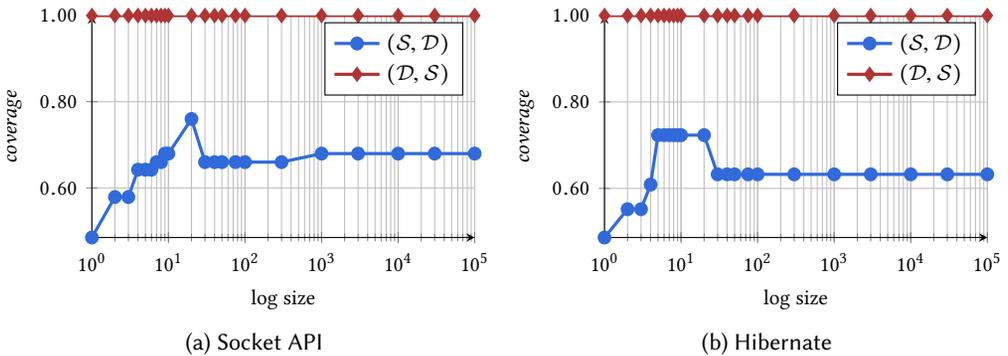

	\vspace{-3mm}
	\centering
	\subfloat[Socket API]
	{\plotdata{$\mathit{coverage}$}{mybluecolor}{myredcolor}
		{$(\mathcal{S},\mathcal{D})$}{$(\mathcal{D},\mathcal{S})$}{1}{2}{north east}
		\label{fig:comparing:specifications:5}}
	\subfloat[Hibernate]{\plotdata{$\mathit{coverage}$}{mybluecolor}{myredcolor}{$(\mathcal{S},\mathcal{D})$}{$(\mathcal{D},\mathcal{S})$}{5}{6}{north east}
		\label{fig:comparing:specifications:13a}}
	\vspace{-3mm}
	\caption{Language coverage between designed and discovered specifications of software systems.}
	\label{fig:comparing:specifications}
	\vspace{-2mm}
\end{figure*}

The fact that all the measured $\smash{\funcCall{\mathit{coverage}_{\mathit{eig}^\bullet}}{\mathcal{D},\mathcal{S}}}$ values are equal to 1.0 (see the red diamonds in \autoref{fig:comparing:specifications}), tells that the behaviours of all the discovered specifications are sub-behaviours of the corresponding specifications from \autoref{fig:two:software:systems}.
In general, one should expect that specifications discovered from more traces of a system should cover the behaviour of the system better, \ie the values denoted by blue circles in \autoref{fig:comparing:specifications} should increase with the increase in the size of the log. Any other trend suggests that the discovery algorithm ignores, or even loses, knowledge about the behaviour of the system with more observed traces available.

%%%%%%%%%%%%%%%%%%%%%%%%%%%%%%%%%%%%%%%%%%%%%%%%%%%%%%%%%%%%%%%%%%%%%%%%%%%%%%%
\subsection{Scalability}
\label{sec:scalability:evaluation}
%%%%%%%%%%%%%%%%%%%%%%%%%%%%%%%%%%%%%%%%%%%%%%%%%%%%%%%%%%%%%%%%%%%%%%%%%%%%%%%

Next, we study the scalability of the eigenvalue-based precision and recall measures using real-world and synthetic datasets.

%%%%%%%%%%%%%%%%%%%%%%%%%%%%%%%%%%%%%%%%%%%%%%%%%%%%%%%%%%%%%%%%%%%%%%%%%%%%%%%
\subsubsection{Scalability Evaluation on Real-World Data}
%%%%%%%%%%%%%%%%%%%%%%%%%%%%%%%%%%%%%%%%%%%%%%%%%%%%%%%%%%%%%%%%%%%%%%%%%%%%%%%

Practical language measures and quotients must be able to handle large languages. 
Hence, we measured the wall-clock time of the eigenvalue-based precision and recall computation for fifteen real-world logs and corresponding discovered specifications. 
The logs encode executions of real-world IT systems executing genuine business processes with real customers.
The logs are publicly available\footnote{Logs are available at: \url{https://data.4tu.nl/repository/collection:event_logs_real}} and of different complexities. 
The log with the least variation in traces is the BPI Challenge (BPIC) 2013 (open cases) log.
It can be encoded in a finite acyclic automaton with only 116 states. 
The BPIC 2017 log, on the other hand, translates to an automaton with 105\,387 states.

\begin{table*}[h]
\vspace{-1mm}
\begin{center}
\caption{Measurements on Ryzen 5 2600X with 64GB of RAM.}
\vspace{-2mm}
\label{tab:evaluation:scalability}
\begin{scriptsize}
	\begin{tabular}{@{}lrrrrrrrrrrrr@{}}
\toprule
&\multicolumn{3}{c}{Automaton size (\# states)}&\multicolumn{3}{c}{Largest eigenvalue}&Recall&Precision&\multicolumn{4}{c}{Wallclock-time (minutes)} \\
Log name&L&M&L $\cap$ M &L&M&L $\cap$ M& & & L&M&L $\cap$ M&total \\
\midrule
\textsf{BPIC'12}						&27\,943	&4		&27\,943	&1.40 & 22.00		& 1.40 		&1.000&0.063&5.69&0.00&5.69&11.38 \\
\textsf{BPIC'13-closed}			&280			&9		&57				&2.09 & 2.70		& 1.85 		&0.837&0.685&0.00&0.00&0.00&0.00 \\
\textsf{BPIC'13-incidents}	&4\,426		&4		&67				&2.20 & 3.19		& 1.78 		&0.731&0.558&0.23&0.00&0.02&0.25 \\
\textsf{BPIC'13-open}				&116			&8		&5				&2.71 & 2.08		& 1.75 		&0.559&0.840&0.00&0.00&0.00&0.00 \\
\textsf{BPIC'15-1}					&33\,090	&25		&12\,815	&1.62 & 390.24	& 1.45 		&0.771&0.004&8.65&0.00&1.08&9.72 \\
\textsf{BPIC'15-2}					&32\,060	&26		&25\,863	&1.64 & 389.67	& 1.63 		&0.983&0.004&4.45&0.00&3.66&8.11 \\
\textsf{BPIC'15-3}					&33\,353	&16		&33\,200	&1.68 & 374.26	& 1.68 		&1.000&0.004&0.86&0.00&0.87&1.73 \\
\textsf{BPIC'15-4}					&27\,566  &28   &27\,466  &1.71 & 322.20  & 1.71    &1.000&0.005&1.30&0.01&1.56&2.87 \\
\textsf{BPIC'15-5}					&36\,221  &14   &26\,609  &1.37 & 369.17  & 1.36    &0.996&0.004&3.00&0.00&0.85&3.85 \\
\textsf{BPIC'17}						&105\,387	&6		&105\,387	&1.39 & 22.45		& 1.39 		&1.000&0.062&80.71&0.00&80.71&161.43 \\
\textsf{WABO-1}							&23\,416	&17		&10\,585	&1.63 & 367.98	& 1.51 		&0.844&0.004&0.85&0.00&1.00&1.85 \\
\textsf{WABO-2}							&23\,930	&14		&23\,312	&1.49 & 369.26	& 1.39 		&0.820&0.004&0.87&0.00&0.65&1.52 \\
\textsf{WABO-3}							&23\,519	&35		&23\,519	&1.61 & 361.57	& 1.57 		&0.954&0.004&0.60&0.00&0.42&1.03 \\
\textsf{WABO-4}							&19\,984	&46		&19\,984	&1.54 & 297.59	& 1.54 		&1.000&0.005&2.97&0.00&2.97&5.93 \\
\textsf{WABO-5}							&25\,060  &4    &24\,981  &1.37 & 346.03  & 1.37    &1.000&0.004&0.43&0.00&1.43&1.86 \\
\bottomrule
\end{tabular}

\end{scriptsize}	
\end{center}
\vspace{-1mm}
\end{table*}

For each log, we discovered a specification using Inductive Miner~\cite{Leemans2016} configured with the default noise threshold of $0.2$. 
For each log and corresponding discovered specification, we applied our method by first constructing the respective finite automata and computing the eigenvalues of their short-circuited representations. 
The observed wall-clock times of the computations of the largest eigenvalues for the log $\textit{L}$, the specification $M$, and their intersection automaton $L \cap M$ are shown in \autoref{tab:evaluation:scalability}. 
As an indicator of the complexity, the number of states of the respective automata are listed in the table. 
Note that the specification automata are considerably smaller than the corresponding log automata.
Presumably, this is because the employed discovery algorithm constructs specifications that do not contain duplicate actions. 
The adjacency matrix of an automaton has size that is quadratic in the number of states in the automaton, which can pose practical difficulties when storing it on a computer.
However, adjacency matrices are usually sparse, which allowed us to use their memory-efficient representations.

The variance in measured wall-clock times is notable. 
The longest time to compute the eigenvalue-based precision and recall was taken for the BPIC 2017 log, whereas for most of the experimented logs, both precision and recall values were computed under ten minutes and often much faster.
The technique used for computing largest eigenvalues is called ``implicitly restarted Arnoldi iterations''~\cite{lehoucq1998arpack}.
Note that this numerical method for computing a largest eigenvalue of a general matrix always converges, but provides no guarantees of convergence in a fixed number of iterations. 
Thus, in our implementation, for practical reasons, we use the threshold of 300\,000 iterations for the maximum number of iterations.
For all the experimented logs, this threshold was sufficient to ensure the convergence of the computations.

\subsubsection{Scalability Evaluation on Synthetic Data}

In the evaluation on the real-world data, the characteristics of logs in terms of the variety and number of traces differed a lot. 
To perform a consistent analysis, next, we report on the values and computation times of the eigenvalue-based quotients for the simulated logs and specifications, both designed and discovered, discussed in \autoref{sec:comparing:specifications}.

\begin{figure*}[t]
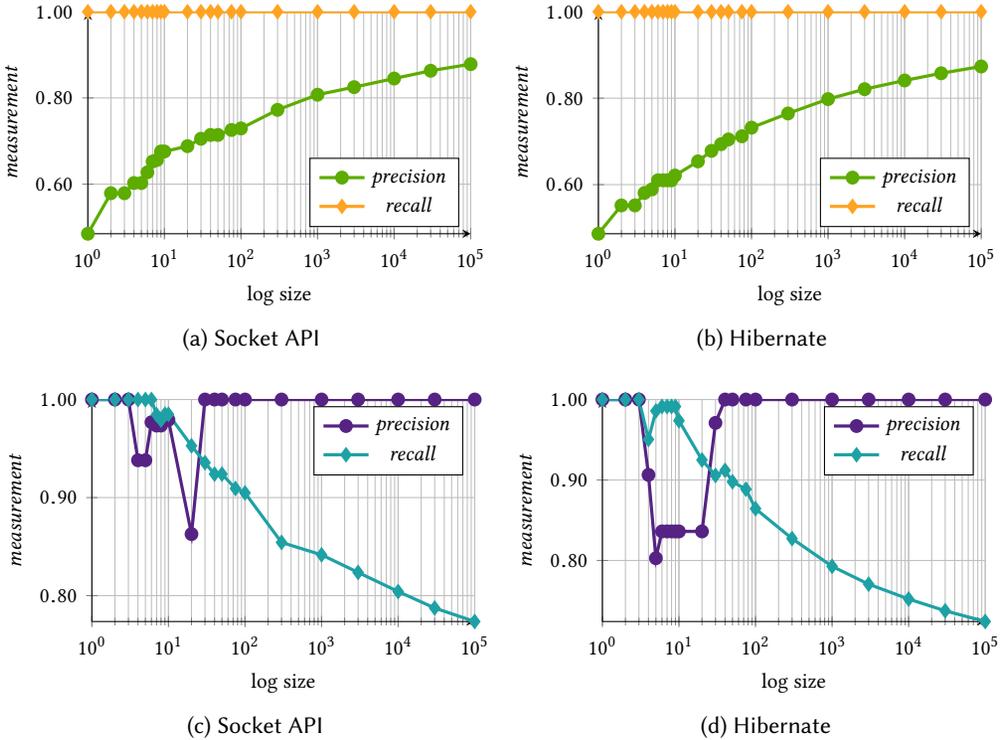

\vspace{-3mm}
\centering
\subfloat[Socket API]{\plotdata{$\mathit{measurement}$}{mygreencolor}{myyellowcolor}{$\mathit{precision}$}{$\mathit{recall}$}{7}{8}{south east}
\label{fig:comparing:model:and:log:5}}
\subfloat[Hibernate]{\plotdata{$\mathit{measurement}$}{mygreencolor}{myyellowcolor}{$\mathit{precision}$}{$\mathit{recall}$}{11}{12}{south east}
\label{fig:comparing:model:and:log:13a}}
\newline
\subfloat[Socket API]{\plotdata{$\mathit{measurement}$}{mypurplecolor}{mytealcolor}{$\mathit{precision}$}{$\mathit{recall}$}{13}{14}{north east}
\label{fig:comparing:dmodels:and:log:5}}
\subfloat[Hibernate]{\plotdata{$\mathit{measurement}$}{mypurplecolor}{mytealcolor}{$\mathit{precision}$}{$\mathit{recall}$}{17}{18}{north east}
\label{fig:comparing:dmodels:and:log:13a}}
\vspace{-3mm}
\caption{Eigenvalue-based precision and recall of: designed specifications w.r.t. their executions, (a) and (b), and discovered specifications w.r.t. the executions they were discovered from, (c) and (d).}
\label{fig:comparing:model:and:log}
\vspace{-2mm}
\end{figure*}

\autoref{fig:comparing:model:and:log} plots the measured eigenvalue-based precision and recall values.
\autoref{fig:comparing:model:and:log:5} and \autoref{fig:comparing:model:and:log:13a} show the values for, respectively, the Socket API and Hibernate specification w.r.t. the simulated logs.
Because all the logs are composed of executions of the specifications, all the recall values are equal to 1.0.
As can be seen from the plots, with the increase of the number of traces in logs, the precision values increase, which is consistent with the fact that the eigenvalue-based precision is monotone.
\autoref{fig:comparing:dmodels:and:log:5} and \autoref{fig:comparing:dmodels:and:log:13a} show precision and recall values for the corresponding discovered specifications w.r.t. the simulated logs.
As can be observed from the plots, the recall values tend to decrease with the increase of the log size.
The fact that precision values tend to be 1.0 suggests that this particular configuration of the discovery technique constructs specifications that do not generalize beyond the behaviour seen in the logs.

\begin{figure*}[h]
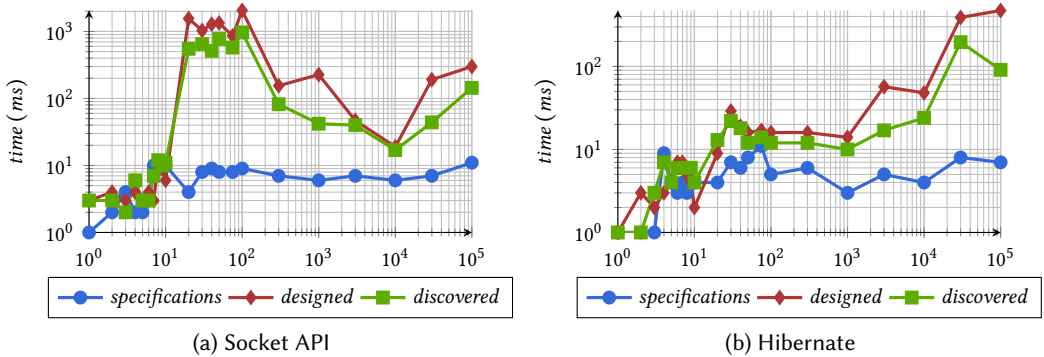

\vspace{-3mm}
\centering
\subfloat[Socket 
API]{\plottime{mybluecolor}{myredcolor}{mygreencolor}{$\mathit{specifications}$}
	{$\mathit{designed}$}{$\mathit{discovered}$}{1}{2}{3}
\label{fig:time:5}}
\subfloat[Hibernate]{\plottime{mybluecolor}{myredcolor}{mygreencolor}{$\mathit{specifications}$}{$\mathit{designed}$}{$\mathit{discovered}$}{7}{8}{9}
\label{fig:time:11}}
\vspace{-3mm}
\caption{Wallclock-time of computing the eigenvalue-based quotients on Ryzen 5 2600X with 64GB of RAM.}
\label{fig:time}
\vspace{-2mm}
\end{figure*}

\autoref{fig:time} plots times of computing the eigenvalue-based quotients reported in \autoref{fig:comparing:specifications} and \autoref{fig:comparing:model:and:log}; note the use of logarithmic scales for both axis.
Each plotted value reports the overall time of computing two quotients for the corresponding event log shown on the x-axis.
The blue circles denote the times of computing the coverage quotients for \emph{specifications} from \autoref{fig:comparing:specifications}.
The red diamonds show the times of computing precision and recall for the \emph{designed} specifications reported in \autoref{fig:comparing:model:and:log:5} and \autoref{fig:comparing:model:and:log:13a}.
Finally, the green squares report the times of computing precision and recall for the \emph{discovered} specifications reported in \autoref{fig:comparing:dmodels:and:log:5} and \autoref{fig:comparing:dmodels:and:log:13a}.

The experimental results, both on the real-world and synthetic datasets, tell us that the size of the input is not enough to determine the runtime of the method.
Instead, the rate of convergence of a largest eigenvalue computation depends on other properties of the adjacency matrices of the underlying automata, e.g., on the difference between the largest and the second-largest eigenvalue.
We conclude the scalability experiments with the insight that our current implementation of the method shows variability in performance depending on the convergence of the underlying eigenvalue computation.

%%%%%%%%%%%%%%%%%%%%%%%%%%%%%%%%%%%%%%%%%%%%%%%%%%%%%%%%%%%%%%%%%%%%%%%%%%%%%%%
\section{Related Work}
\label{sec:related_work}
%%%%%%%%%%%%%%%%%%%%%%%%%%%%%%%%%%%%%%%%%%%%%%%%%%%%%%%%%%%%%%%%%%%%%%%%%%%%%%%

The comparison of behaviours has played a major role in the verification of software and hardware artefacts across several areas of computer science and software engineering, including the theory of concurrent systems
\cite{GorrieriVersari/2015:IntroductiontoConcurrencyTheory-TSandCCS},
reactive systems \cite{Clarke.etal01:ModelCheckingBook},
and agent programming 
\cite{DeGiacomo.etal/AI2016:Agentplanningprograms},
to mention but a few.
\autoref{sec:related_work:behavioural_equivalence} outlines noticeable
notions of behavioural equivalence and behavioural comparison, including inheritance and similarity.
Then, \autoref{sec:related_work:process_mining} describes the evolution of the
precision and recall measures for behavioural comparison in the field of
process mining, along with highlights on commonalities and dissimilarities to our
approach.
Finally, \autoref{sec:related_work:software_engineering} reports on previous research on behavioural comparison in software engineering, again emphasising the similarities and differences with our technique.

\subsection{Behavioural Equivalence}\label{sec:related_work:behavioural_equivalence}
In the context of dynamic systems, there are several notions of behavioural equivalence, which are broadly classified into two categories: equivalences that are based on the interleaving semantics and those based on the true concurrency semantics~\cite{Glabbeek.Goltz/AI2001:RefinementofActionsandEquivalenceNotionsforConcurrentSystems}.
We remark that the systems under analysis in this paper fall under the class of
finite-state, assume the presence of final/accepting states, and operate with
interleaving semantics.
Probably the most important behavioural equivalence between two systems of
computation in this context is the one that guarantees that any step performed
in one system can be mimicked by the other one, and vice
versa~\cite{Demri/2016:TemporalLogicsinCSFiniteStateSystems}. This idea is the
basis for the notion of
\emph{bisimulation}
~\cite{Milner/1982:ACalculusofCommunicatingSystems}.
On rooted labelled transition systems (a super-class of the systems we
analyse), bisimulation imposes that from the initial state onward, possible
actions must coincide between the systems and inductively lead to states that
are bisimilar as well.
\emph{Weak bisimulation}~\cite{Milner/1982:ACalculusofCommunicatingSystems}
relaxes bisimulation in that it considers only observable actions, i.e., it is
permitted that systems guarantee bisimulation on non-$\tau$ transitions only,
as $\tau$ transitions can be added as prefix- or suffix-moves to that extent.
\emph{Branching bisimulation} enforces weak bisimulation by requiring that the same set of choices is offered before and after each unobservable action~\cite{Glabbeek/CONCUR1993:TheLinearTimeBranchingTimeSpectrumII}.

Bisimulation exerts less strict conditions than graph isomorphism, which is a
bijection between all states preserving transitions. However, it is also more
specific than \emph{trace equivalence}, solely ascertaining that observable
actions match, thus being insensitive to non-determinism, internal actions,
choices, and
deadlocks~\cite{GorrieriVersari/2015:IntroductiontoConcurrencyTheory-TSandCCS}.
\emph{Completed} trace equivalence adds the condition that, if systems have
sink states from which no further action is possible, they must be reachable in
systems by replaying the same traces.
Our research benefits from the multiple notions of behavioural equivalence and
investigations conducted on the matter so far, yet it abstracts from the
decision problem on the matching of behaviours and rather aims at assessing
\emph{how much} the behaviour of a first system is extended by a second one.

\citet{KunzeWeske/2016:BehaviouralModelsforBusinessProcesses} declare not only behavioural equivalence, but also behavioural similarity and inheritance, as main challenges pertaining to behavioural comparison.
In particular, the authors introduce a property for the latter, namely
\emph{trace inheritance}, which is satisfied only if the language of a system is
included in the language of another system at the same level of abstraction.
In the light of that definition, our research thus focuses on behavioural
inheritance~\cite{Basten/1998:SystemDesignWithPetriNetsAndProcessAlgebra}, and
specifically trace inheritance, between dynamic systems. However, we aim at
providing a measure assessing in how far languages extend one another, rather
than checking whether the property holds true or not.
This quantitative aspect typically pertains more to behavioural similarity.
To measure it, applying na\"ive approaches based on set-similarity measures
such as the Jaccard coefficient~\cite{Deza.Deza/2006:DictionaryofDistances} to
the set of systems' traces proves infeasible: Loops lead to trace sets of
infinite cardinality.

To overcome that problem, approaches to behaviour similarity were introduced that restricted the analysis to local relations between traces' events~\cite{Kunze/PhD2013:Searchingbusinessprocessmodelsbyexample}.
Noticeable examples include the $n$-gram
similarity~\cite{Mahleko.etal/EEE2005:ProcessAnnotatedServiceDiscovery-n-gram-basedindex},
 comparing systems by the shared allowed $n$-long sub-sequences in systems'
respective traces. Despite the efficiency of the solution, the issue is that even if $n$-grams coincide, not necessarily do the traces as
well. Nevertheless, the best results are reportedly achieved with the least
strict parameter, namely $n=2$. Later on, behavioural profiles similarity was
introduced in~\cite{DBLP:journals/sosym/0001WW15}.
The idea is to compare ``footprints'' of systems, obtained by matrices
connecting pairs of event labels with mutually exclusive relations. Those
relations are exclusiveness, strict order, and interleaving order, i.e., the
fundamental relations of behavioural profiles as of~\cite{WeidlichMW11}.
Despite being semantically rich, 
\citet{PolyvyanyyADG16} show that the expressive power of behavioural profiles is strictly less than regular languages, thus entailing that they cannot be used to decide trace equivalence of finite state automata.
Our approach abstracts from the local perspective on traces or relations between events in that it resorts on the topological entropy to compare the variability of languages.
We reflect the comparison of dynamic systems into precision and recall.

\subsection{Precision and Recall in Process Mining}\label{sec:related_work:process_mining}
Process mining aims at extracting knowledge about
processes from the digital data stored by organisations' IT systems~\cite{Dumas.etal/2018:FundamentalsofBPM}.
Process mining is adopted to discover new facts, including process
specifications themselves that were not documented before, compare the expected
process behaviour with reported reality and detect deviations between the
former and the latter~\cite{Aalst16}.
It shows thus the inherent aim of finding and assessing the match between the
behaviours of a dynamic system, in terms of to-be process specifications versus
as-is
process data. Therefore, the identification of quotients that allow for a
comparative measurement of behaviours naturally suits the matter.
In particular, 
\citet{Buijs.etal/IJCIS2014:QualityDimensionsinProcessDiscovery} identify (replay) fitness, precision, generalisation, and simplicity as
the four main quality dimensions for assessing the quality of process mining results~\cite{JanssenswillenDJD17}.

A first precision measure called ``behavioural appropriateness'' is introduced
in the seminal work of \citet{RozinatA08IS} as the degree of how much behaviour is permitted by the
specification although not recorded in the log.
The \emph{simple behavioural appropriateness} builds on the observation that an increase of alternatives or parallelism entails a higher number of enabled transitions during log replay, while the \emph{advanced behavioural appropriateness} uses long-distance precedence dependencies between pairs of activities. 
In this way, it is higher when sometimes-forward and sometimes-backward
relation pairs shared between specification and log approximate the total
amount of the specification. Conversely, it is lower if the specification
allows for more variability. The assumption of total fitness of the log entails
that the log cannot show more variability than the specification.
Our approach also compares the availability of actions at given states, but abstracts from the exact replay of traces by considering the entropy of the languages.

The ETConformance approach avoids the complete exploration of the specification
behaviour by traversal of the specification to solely reflect the traces
recorded in the log~\cite{Munoz-Gama.Carmona/CIDM2011:ETConformanceExt:StabilityConfidenceSeverity}.
To that extent, a finite (acyclic) rooted deterministic labelled transition system named \emph{prefix automaton} is generated by folding traces based on prefix trace-equivalence of the generated states. The assumption of total fitness entails that the set of available transitions contains the ones permitted by the prefix automaton.
The locality of the approach allows for efficient computation, with the downside that only behaviour close to the log is considered.
Similarly, our approach assesses precision by quantifying the behavioural differences among states of a finite-state rooted labelled transitions system.
However, it abstracts from the recorded runs of the involved specifications.
Remarkably, \citet{Munoz-Gama.Carmona/CIDM2011:ETConformanceExt:StabilityConfidenceSeverity} also introduce advanced diagnostic measures to assess the severity of imprecisions and their stability factor with respect to small perturbations in the log.

An approach combining the concept of prefix automaton with the one of \emph{alignments}~\cite{AdriansyahDA11EDOC} is proposed by \citet{AalstAD12WIDM} to deal with non-entirely fitting logs.
The proposed \emph{alignment-based precision} is the arithmetic mean over all
events in the log of the ratio between the activities that were allowed by the
specification and the ones that were actually executed as per the prefix
automaton, given the replay history.
\citet{Adriansyah.etal/ISEM2015:OneAndBestAlignPrecision} propose different precision measures based on the nature of the alignments to be considered. The underlying structure remains a prefix automaton 
as in \cite{Munoz-Gama.Carmona/CIDM2011:ETConformanceExt:StabilityConfidenceSeverity}, here augmented by associating weights to states.
As in the approaches of
\cite{AalstAD12WIDM}~and~\cite{Adriansyah.etal/ISEM2015:OneAndBestAlignPrecision},
 the precision measure proposed in this paper does not take into account
diverging behaviours. To that extent, the log repair given by alignments could
be beneficial to a pre-processing phase. Because our solution resorts on the
entropy of specifications' languages, it abstracts from the replay and counting
of events.

More recently, \citet{Leemans2016} introduced precision and recall measures to
compare the behaviour of specifications or logs, requiring a finite state
automaton as
the underlying structure for a state-to-state comparison
as in~\cite{Adriansyah.etal/ISEM2015:OneAndBestAlignPrecision,Munoz-Gama.Carmona/CIDM2011:ETConformanceExt:StabilityConfidenceSeverity}.
To cope with the high computational effort required by the intersection
operations, a projection of both specifications is pre-computed for every
subset of $k$ actions in the joint alphabet. Resulting automata contained
silent transitions and presented non-determinism. The resulting \emph{Projected
Conformance Checking (PCC) precision} and a corresponding recall measure build
then on $k$-subsets projections.
As in \cite{Leemans2016}, we benefit from minimisation of the underlying structure and provide dual definitions for precision and recall. However, the computation of measures based on eigenvalues does not require the approximation via $k$-projections.

The anti-alignment based precision is defined by \citet{DongenCC16} using the
concept of anti-alignment first proposed in \cite{ChatainC16}. An
anti-alignment is a finite trace of a given length which is accepted by the
process specification, yet not in the log and sufficiently distant from any
trace therein (where the trace distance can be computed by using edit
distance~\cite{Levenshtein/SPD1966:EditDistance}, e.g.).
To assess precision, every distinct trace is removed from the log and an anti-alignment of equal length is generated with maximum distance. These are averaged. %
Likewise, we reason on language properties of analysed specifications, thus
abstracting from the number of occurrences of a trace. However, our approach
does not require the iterative scan and comparison of specifications excluding
parts of the behaviour, thus saving on computation time.

The Artificially Generated Negative Events technique (AGNEs) discovers process
specifications out of logs enriched with artificially injected negative
events~\cite{Goedertier.etal/JoMLR2009:ProcessDiscoveryArtificialNegativeEvents}.
The assumption is that the log includes the complete set of behavioural
patterns, which means that events can only be missing in a log because they are
not permitted by the process. The notion of \emph{recall} can then be defined
as the rate of true positives over all events classified as positive, and
\emph{specificity} accordingly. Before the computation, a preliminary reduction
of matching event sequences to single traces is conducted such that traces do
not add up to the overall amount.
Our definitions of precision and recall are also dual and do not depend on the number of occurrences of the same trace.
However, no artificial injection of noise is required in our approach, thus reducing the bias that the alteration of the input behaviour with negative information may cause.

To evaluate their discovery algorithm, namely the Heuristic Miner, 
\citet{Weijters2006} introduce the so-called Parsing Measure (PM), which is
based on the fraction of correctly parsed traces over all traces in the input
log. 
As a derivative, the Continued Parsing Measure (CPM) provides a more
fine-granular analysis, at the price of being bound to the specification of the
underlying Heuristic Miner.
Our notion of recall for a specification is also based on the measuring of the
part of language not covering another behaviour.
Noticeably, PM and CPM weigh the amount of incorrectly parsed traces, thus quantitatively assessing to which extent the divergences occur in the event log.
Owing to our level of abstraction, we do not account for this assessment.
However, the measure we propose is less dependent on the recorded traces and is not based on the count of events.

The fitness measure proposed by 
\citet{RozinatA08IS} counts the number of tokens consumed and produced during
the replay of traces over the Petri net specification, and puts them into
relation with
missing tokens and tokens remaining after completion. It extends a simpler
measure computed as the ratio of traces causing missing or remaining tokens
defined in the same paper and named \emph{proper completion} in
\cite{JanssenswillenDJD17}. Another token-based fitness measure, used in
genetic process mining, 
accounts for trace frequency~\cite{Medeiros.etal/2007:GeneticProcessMining}.
In contrast to \cite{RozinatA08IS,Medeiros.etal/2007:GeneticProcessMining}, we aim at defining measures that are
not tailored to specific behaviour specification language, thus we do not rely
on Petri net semantics to define recall.

The concept of alignment-based fitness introduced by \citet{AalstAD12WIDM}
relies on a cost function to be specified by the user, indicating the penalty
for non-synchronous moves in the replay of traces on the specification. Fitness
is then computed for every trace as the total cost of the optimal alignment,
divided by a worst-case alignment, indicated as the one consisting of moves in
the trace for every event, followed by moves in the specification from the
start to the end of a shortest run. Log fitness is then calculated by averaging
the trace fitness values over all traces.
Alignments are a valuable means to make the approach independent on the
specification language, as in the rationale of our investigation. Our technique
does not allow the user to indicate costs. Providing this feature in our
approach is an intriguing problem that could be addressed in future work.
On the other hand, our approach does not resort on the computationally
expensive finding of optimal runs on the input specifications.

We remark that especially the approaches described in
\cite{AalstAD12WIDM,Adriansyah.etal/ISEM2015:OneAndBestAlignPrecision,Leemans2016,RozinatA08IS}
 not only propose precision and recall measures and algorithms for their
computation, but provide also techniques to illustrate where and in how far
deviations occur between the log and the specification.
The integration of those powerful diagnostic tools with our approach delineates interesting plans for future research.

To conclude, 
\cite{TaxLSFA17} recently defined five requirements (there named \emph{axioms})
that a precision measure should guarantee, in a strive for the general
definition of fundamental properties that should be satisfied by process mining
quality measures. The authors show that neither of aforementioned simple
behavioural appropriateness~\cite{RozinatA08IS}, advanced behavioural
appropriateness~\cite{RozinatA08IS}, ETC
precision~\cite{Munoz-Gama.Carmona/CIDM2011:ETConformanceExt:StabilityConfidenceSeverity},
AGNEs specificity~\cite{Goedertier.etal/JoMLR2009:ProcessDiscoveryArtificialNegativeEvents}, or PCC precision~\cite{Leemans2016} comply with their requirements for precision.
By design, our approach fulfils all those requirements instead, as shown in \autoref{sec:precision_recall_evaluation}.

\subsection{Behavioural Comparison in Software Engineering}
\label{sec:related_work:software_engineering}
In software engineering, a noticeable body of literature on automaton-based specification mining~\cite{Ammons2002,Lo.etal/2011:MiningSoftwareSpecifications} have proposed highly relevant contributions towards the behavioural comparison of state machines.

Javert~\cite{Gabel.Su/FSE2008:Javert} generates complex system specifications stemming from mined patterns. To that end, the technique applies sound composition rules of branching and sequencing on discovered simple patterns, thus achieving good scalability. Similarly to our solution, Javert resorts on automata theory for the composition steps and the representation of models. Our approach thus complements Javert in that it can measure the precision and recall of those returned models.

\citet{Shoham.etal/IEEETSE2008:StaticSpecificationMiningUsingAutomata} adopt an automata-based approach to automatically mine the specification of client interactions with APIs for object-oriented libraries. Their approach resorts on the notion of quotient automata to abstract on the representation of behaviour through an equivalence relation over states. Notice that the notion of behavioural quotient we propose applies to languages for obtaining a measurable comparison of systems behaviour regardless of their model's structure. Our language quotient framework is thus separate and integrable with the technique of \citet{Shoham.etal/IEEETSE2008:StaticSpecificationMiningUsingAutomata}, which could be employed to take advantage of their effective removal of spurious patterns.

\citet{Lo2006} propose a framework called QUARK (QUality Assurance framewoRK) for empirically assessing the automata generated by different miners.
Their assumption is that two models have to be compared: one reference and one
reverse-engineered from API interactions. This context is similar to ours in
that we also compare a reference process specification with another behavioural
abstraction, in our case stemmed from a set of execution traces of a process.
In their approach, they compute accuracy in terms of trace similarity. 
They first collect two samples of randomly generated traces, one per model. The precision is the proportion of samples generated by the reverse-engineered model that are accepted by the reference automaton. Dually, the recall is the proportion of traces that are generated by the reference automaton, and are accepted by the reverse-engineered one.
Our approach moves in the opposite direction: we abstract from traces and
compare systems, rather than comparing traces generated by the systems.
Remarkably, \citet{Lo2006} also propose measures that deal with probabilistic finite automata, based upon the Hidden Markov Models comparison.
Their study suggests the extension of our approach toward the analysis of probabilistic models as an opportunity for future research.

The use of simulated traces for system comparison, first reported in
\cite{Lang.etal/ICGI1998:DFALearningCompetition} and applied in
QUARK~\cite{Lo2006}, has been later criticised by \citet{Walkinshaw2008}.
A problem is that it is virtually impossible to cover the whole behaviour of a system by random walks.
This problem is of high severity especially because some faulty executions 
might remain unexplored by a random sample, which is of high relevance in 
software testing~\cite{Walkinshaw.etal/FM2009,Weyuker/ATPLS1983}.
To address this bias, \citet{Walkinshaw.etal/FM2009} propose an adaptation of 
the original Vasilevski/Chow 
W-Method~\cite{Chow/IEEETSE1978:TestingSoftwareDesignModeledbyFSMs,Bogdanov.etal/FAC2006:TestingMethodsforXMachinesReview}.
 Their technique is aimed at generating test sets that cover all 
distinguishable runs of the model.
Furthermore, they refine the notions of precision and recall to account for not only the traces that are mutually accepted by the compared models, but also to inspect the capability of the two to reject traces that are not compliant with the target behaviour.
In our context, to-be-rejected traces are not considered as we assume the log to stem from registered correct system runs.
However, we see in this aspect an endeavour for future work: an extension of our language-quotient based approach that accounts for the semantic discrimination of runs that ended up in positive outcomes from those that do not, similarly to what was done by \citet{PoncedeLeon.etal/InfSci2018:Incorporatingnegativeinformationtoprocessdiscoveryofcomplexsystems} and \citet{Chesani.etal/JPNOMC2009:ExploitingInductiveLogic}.

\citeauthor{Walkinshaw2013} extend their seminal work~\cite{Walkinshaw2008} in two directions.
First, they expand the comparison measures with classical data mining ones such as specificity and balanced classification rate.
Second, they introduce the LTSDiff algorithm, which compares models under a structural perspective, rather than a behavioural one.
In this paper, we do not consider the structural similarity, thus being
model-agnostic and not imposing requirements on the determinism or minimality
of input systems.
However, our technique could be improved by integrating the cognitive-like, iterative approach of the LTSDiff algorithm, based on an intermediate results expansion starting from landmarks~\cite{Sorrows.Hirtle/COSIT1999:TheNatureofLandmarksforRealandElectronicSpaces} (i.e., matching subsets of the inputs).

\citet{Quante.Koschke/WCRE2007:DynamicProtocolRecovery} first consider a measure for model comparison taking into account the language of involved automata without the analysis of generated traces.
They devise to that extent an approach similar to that of edit distance.
A minimised union automaton is first created between the input ones.
Thereupon, a concurrent synchronous run is executed on each of the models and the union automaton.
It determines the number of edits, that is, the transitions to be removed from the union (never traversed) or added to the input model (unfolded self-loops).
The final measure is computed by averaging the distances in terms of edits of the models from the union automaton.
Our approach revolves around language comparison based on the analysis of automata as well.
However, it discriminates between precision and recall, thus giving a more precise picture of the accuracy of the mined model with respect to the reference of the log.

\citet{Pradel.etal/ICSM2010:FrameworkForEvaluationOfSpecificationMinersFSMs}
use a variant of the $k$-tails
algorithm~\cite{Biermann.Feldman/IEEETC1972:OntheSynthesisofFSMsfromSamplesofTheirBehavior} to compare mined and reference models. To that extent, they first generate
the union of the finite automata given as input models. Then, they adapt the
$k$-tails algorithm to approximate the matching of those states from which
common (sub)sequences of length $k$ can be generated. Such states are then
merged. Precision is computed based on the number of shared transitions between
the mined model and the intersection of the reference model with the automaton
subject to $k$-tails merging. Recall is computed analogously but switching
mined and reference model.
The usage of $k$-tails to merge states allows for the processing of models mined from noisy or incomplete traces.
On the other hand, the fact that matches are not exact and subject to a proper choice of $k$ may lead to an inaccuracy of results, as emphasised by~\citet{Walkinshaw2013}.
As in
\cite{Pradel.etal/ICSM2010:FrameworkForEvaluationOfSpecificationMinersFSMs},
our approach considers a language abstraction of systems for comparison
purposes, without generating trace sets.
In contrast to it, we do not resort to structural approximations over the input
specifications.
On the one hand, it favours accuracy.
On the other hand, an adaptation of our approach to account for noise, as in \cite{Pradel.etal/ICSM2010:FrameworkForEvaluationOfSpecificationMinersFSMs}, is an interesting direction for future work.

Interesting research avenues for future work stem from the extension of language measures to cater for more expressive models than automata-based behaviours labelled by the sole activity name.
\citet{Berg.etal/FASE2006:RegularInferenceforStateMachineswithParameters} present an algorithm that extends the transition labels of inferred automata with propositional guards on function parameter values, based on queries over the observed runs of protocol implementations.
Later, \citet{Lorenzoli.etal/ICSE2008:AutomaticGenerationofSwBehavioralModels} with GK-tail and   \citet{Walkinshaw.etal/ESE2016:InferringExtendedFSMModelsfromSwExecutions} with MINT (Model Inference Technique) propose techniques to infer Extended Finite State Machines (EFMSs), namely automata with guards on data stored in the program memory of the program, from a set of program traces.
\citet{Emam.Miller/ACMTOSEM2018:ExtendedProbabilisticFSAfromSwExecutions} improve on the existing EFMS inference algorithms with a stochastic-based approach to include in the discovered models of behaviour the probabilities that determine the likelihood of transitions.
The techniques proposed by \citet{Narayan.etal/ACMTECS2018:MiningTimedRegularSpecificationsfromSystemTraces} discover behavioural rules based on Timed Regular Expressions (TREs), which are equivalent to timed automata~\cite{Asarin.etal/JA2002:TimedRegularExpressions}, to cater for constraints related to real-time.
\citet{Krismayer.etal/CAiSE2019:ConstraintMiningCyberPhysicalSystems} illustrate a technique to mine constraints out of event logs that store the information of software systems operating in the cyber-physical domain. The analyzed constraints express rules on the sequence of actions, exert limitations on time spans, and predicate on attribute values of the events.
These works inspire interesting future endeavours for our research to measure precision and recall of models of behaviour including data and time aspects. 
%%%%%%%%%%%%%%%%%%%%%%%%%%%%%%%%%%%%%%%%%%%%%%%%%%%%%%%%%%%%%%%%%%%%%%%%%%%%%%%
\section{Discussions}
\label{sec:discussions}
%%%%%%%%%%%%%%%%%%%%%%%%%%%%%%%%%%%%%%%%%%%%%%%%%%%%%%%%%%%%%%%%%%%%%%%%%%%%%%%

This section summarizes the main results of this work and the lessons we learned on the way to obtaining them (\autoref{sec:discussions:lessons}),
discusses threats that could have influenced the validity of the reported conclusions (\autoref{sec:discussions:threats}), and
suggests how the presented results may contribute to software engineering practices (\autoref{sec:discussions:se:practice}).

%%%%%%%%%%%%%%%%%%%%%%%%%%%%%%%%%%%%%%%%%%%%%%%%%%%%%%%%%%%%%%%%%%%%%%%%%%%%%%%
\subsection{Results and Learned Lessons}
\label{sec:discussions:lessons}
%%%%%%%%%%%%%%%%%%%%%%%%%%%%%%%%%%%%%%%%%%%%%%%%%%%%%%%%%%%%%%%%%%%%%%%%%%%%%%%

For over a decade, through the design of various measures and analytics, the process mining community shaped the intuition underpinning the comparison of a specification of a dynamic system with its executions.
Intuitively, a specification should allow for the behaviour seen in the executions and forbid other behaviour~\cite{Aalst16}.
It is only recently that this intuition started to take a concrete form in terms of formal properties that such comparison measures should satisfy~\cite{TaxLSFA17,Aalst18a}.
A repertoire of properties a given measure satisfies can then be seen as a proxy to its usefulness, \ie if a practitioner is interested in certain properties she should pick and use a measure that satisfies them.
The work reported in~\cite{TaxLSFA17} was the first attempt to propose such properties.
On several informal occasions, the properties from~\cite{TaxLSFA17} were criticized for being somewhat na\"ive.
Indeed, they can be satisfied by a measure that, for example, returns zero for any input log and the most permissive specification, i.e., the one that accepts any word, and otherwise returns some constant greater than zero and less than or equal to one.
Obviously, such a measure is not particularly useful.

One issue with the properties from~\cite{TaxLSFA17} is that two specifications, one of which exhibits strictly more behaviour than the other, are allowed to have the same precision value with a given log.
In this work, we strengthened the properties from~\cite{TaxLSFA17} to require a less permissive specification to be more precise with respect to the log (see \lemmaname~\ref{lem:fixed:num} and \lemmaname~\ref{lem:fixed:den}).
As this requirement introduces an additional restriction, every measure that satisfies our properties is guaranteed to satisfy the corresponding less restrictive properties from~\cite{TaxLSFA17}.
Finally, all the other properties from~\cite{TaxLSFA17} are trivially, by definition, satisfied by every precision measure that follows \definitionname~\ref{def:pm:precision}.

In~\cite{Aalst18a}, 21 properties for conformance measures are proposed.
Among those properties, two address both recall and precision measures, five are specifically concerned with recall measures, while six address precision measures.
The properties for precision measures aim to diversify and strengthen the properties from~\cite{TaxLSFA17}.
Recently, in~\cite{Syring2019}, it was shown that all the precision and recall properties from~\cite{Aalst18a} hold for the precision and recall measures presented in~\autoref{sec:precision_recall}.
For example, \propsname~5~and~8 in~\cite{Aalst18a} follow immediately from \lemmaname~\ref{lem:fixed:num} and the fact that a language measure is deterministic (see~\autoref{sec:framework_definition}), while \propsname~3~and~9 in~\cite{Aalst18a} follow immediately from \lemmaname~\ref{lem:fixed:den} and the definition of a language measure.

As of today (December 2019), there is no precision measure, other than the quotient (\definitionname~\ref{def:pm:precision}) instantiated with the short-circuit measure (\definitionname~\ref{def:short:circuit:measure}) induced by the eigenvalue measure (\autoref{sec:framework_instantiations}), that
satisfies the strict properties captured in \lemmaname~\ref{lem:fixed:num} and \lemmaname~\ref{lem:fixed:den}, and all the properties for precision presented in~\cite{TaxLSFA17,Syring2019}.
Note that \lemmaname~\ref{lem:fixed:num} and \lemmaname~\ref{lem:fixed:den} address both finite and infinite languages.

%%%%%%%%%%%%%%%%%%%%%%%%%%%%%%%%%%%%%%%%%%%%%%%%%%%%%%%%%%%%%%%%%%%%%%%%%%%%%%%
\subsection{Threats to Validity}
\label{sec:discussions:threats}
%%%%%%%%%%%%%%%%%%%%%%%%%%%%%%%%%%%%%%%%%%%%%%%%%%%%%%%%%%%%%%%%%%%%%%%%%%%%%%%

% Intro
An important concern about an experiment is how valid its results are~\cite{Wohlin2012}. Note that threats to validity relate to our empirical analysis--formal properties of our measures are not subject to these threats.
According to~\cite{cook1979}, there are four types of threats to the validity of experimental results: internal, construct, conclusion, and external validity.
Next, we discuss several identified aspects that threaten the construct and conclusion validity of the results of our experiments reported in~\autoref{sec:precision_recall_evaluation}.
Aspects that threaten construct validity refer to the extent to which the experiment setting reflects the phenomenon that is studied.
Aspects that threaten conclusion validity relate to the ability to make correct conclusions about the observed outcomes in response to the treatments of the experiment~\cite{Wohlin2012}.

With respect to \emph{construct validity}, we first focus on the threats of
\emph{incomplete} selections of subjects and their \emph{random heterogeneity}.
Our selection of precision and recall measures for the experiment was initiated
with the six precision measures studied in~\cite{TaxLSFA17} and, then, extended
to nine precision and six recall measures,
cf.~\autoref{tab:approaches:labelsandrefs:precision}
and~\autoref{tab:approaches:labelsandrefs:recall}.
The selection was primarily driven by the availability of open-source implementations of the measures in ProM and CoBeFra frameworks~\cite{DBLP:conf/cidm/BrouckeWVB13} in 2017.
Hence, the selection of the measures for experimentation may not be complete.
In the recent study mentioned above, namely in~\cite{Syring2019}, eight recall and eleven precision measures were evaluated.
One of these eight recall measures is the eigenvalue-based recall presented in this work.
From the remaining seven recall measures, six are also evaluated in \autoref{sec:precision_recall_evaluation}; the baseline recall measure presented in~\cite{Syring2019} is equivalent to proper completion measure.
Hence, one recall measure was studied in~\cite{Syring2019} but not evaluated in~\autoref{sec:precision_recall_evaluation}, viz. causal footprint recall~\cite{Aalst16}.
Note, however, that causal footprint recall was shown in~\cite{Syring2019} to
fulfil only four out of seven recall-related properties from~\cite{Aalst18a},
whereas our recall measure, as shown in~\cite{Syring2019}, fulfils all the
seven properties.
Out of eleven precision measures evaluated in~\cite{Syring2019}, one is the eigenvalue-based precision presented in this work, while seven are also evaluated in~\autoref{sec:precision_recall_evaluation}.
The baseline precision measure from~\cite{Syring2019}, not evaluated in this work, is undefined for specifications that encode infinite collections of executions and, thus, can be seen as a theoretical baseline measure with a rather limited practical applicability.
Furthermore, we did not study behavioural precision~\cite{Weerdt2011} and weighted negative event precision~\cite{Broucke2014}, which both aim to improve the measure from~\cite{Goedertier.etal/JoMLR2009:ProcessDiscoveryArtificialNegativeEvents} evaluated in~\autoref{sec:precision_recall_evaluation}.
However, in~\cite{Syring2019}, all the three measures
from~\cite{Goedertier.etal/JoMLR2009:ProcessDiscoveryArtificialNegativeEvents,Weerdt2011,Broucke2014}
 were demonstrated to 
violate four out of eight properties for precision measures, which suggests that the
measures are qualitatively similar.
Note that our precision measure, as shown in~\cite{Syring2019}, fulfils all the eight properties.

We further need to consider a potentially \emph{restricted generalizability
across constructs}, since several issues with measures of our comparison set may not
render it useless. Also, the choice of properties to consider may be subject to discussion.
While these threats cannot be discarded in their entirety since their
definition and selection follows conceptual arguments, we observe that our definition of properties is consistent with those defined by other research~\cite{Syring2019,TaxLSFA17}.
Also, we acknowledge that violating a property 
does not mean that it will be violated frequently or with high severity in a specific set of application
scenarios.

An important threat to \emph{conclusion validity} related to 
\emph{`fishing'} for a specific result, and indeed, a biased selection of
the measures to compare against would be problematic. However, our experiments
considered a large collection of precision and recall measures that are
commonly used and which also rely on different formal foundations. This restricts
the potential impact of this threat. Considering \emph{reliability of measures}
and
potentially \emph{low statistical power}, we acknowledge that we proved the
violation
of certain properties through counterexamples, relying on the third-party
implementations of the evaluated precision and recall measures. While we
observed that not all of the tested synthetic examples lead to the respective
violations, our main claims relate to the formal guarantees that our measure provides and which other measures miss. 
This is a formal argument that is not affected by statistical considerations. 
Moreover, our datasets are limited to
the BPIC logs and specifications synthesized with one discovery algorithm.
In the model-to-model comparisons, we have been limited to two
designed models and a small set of discovered specifications based on one
discovery technique. While we see no evidence for one of these aspects
affecting the conclusion validity, they constitute a certain threat. Lastly, to
exclude \emph{random irrelevancies} in the experimental setting, we ran our
experiments for several times to record average execution times, verifying that
the same outcome is observed.

%%%%%%%%%%%%%%%%%%%%%%%%%%%%%%%%%%%%%%%%%%%%%%%%%%%%%%%%%%%%%%%%%%%%%%%%%%%%%%%
\subsection{Software Engineering Practice}
\label{sec:discussions:se:practice}
%%%%%%%%%%%%%%%%%%%%%%%%%%%%%%%%%%%%%%%%%%%%%%%%%%%%%%%%%%%%%%%%%%%%%%%%%%%%%%%

Behavioural specifications like the ones used in this paper are extensively used in practice for problem solving (95\%) and documentation (91\%)~\citep{hutchinson2014model}, with visual use case models (39\%) and business process models (23\%) being among the most popular ones~\citep{wagner2019status}. 
The behavioural comparison of representations of dynamic systems is at the
core of many software engineering techniques. Our measures, therefore, have a
potential to influence software engineering practice. In the remainder,
we discuss the implications for several exemplary areas, such as software
configuration management, model-based software engineering, and software
testing.

Software configuration management (SCM)~\cite{leon2015software}
comprises models and methods to track,
organize, and control the evolution of the artefacts involved in the
development of a software system. It is motivated by changes in the
requirements to address, the people involved in development, the policies and
rules to obey, or a project's schedule. The behavioural
quotients defined in our work may support several of common SCM practices once
configuration items (CIs) such as source code modules, test cases, and
requirements specifications have been identified. For instance, SCM requires
the definition of baselines, formally established versions of CIs that
structure the progress of a development project. An example is the functional
baseline that describes an item's functional, interoperability, and
interface characteristics~\cite{keyes2004software}. In a functional
configuration audit, as part of SCM, behavioural quotients can enable an assessment
of the degree to which the baseline has been reached. Moreover, to assess the
impact of change requests on CIs that capture behavioural information, such as
UML activity diagrams or source code fragments, behavioural quotients can be used to quantify
the impact of the respective request.

Turning to specific software development methodologies, we consider approaches
for model-based software
engineering (MBSE)~\cite{DBLP:series/synthesis/2012Brambilla}. In essence, they aim at
structuring the development process around abstract models and automated code
generation through model transformations. However, current MBSE practice faces
challenges related to the maintenance of code generators that transform models
into executable code, concerning the design of domain-specific languages
(DSLs), and related to the integration within agile development
projects~\cite{DBLP:books/sp/18/Kautz0R18}. Behavioural quotients may support
initiatives to overcome these challenges by assessing the behavioural
difference of models to identify required changes in code generators, by
comparing instances defined in DSLs to assess their commonalities for
consolidation, and by providing notions of model consistency to enable the
identification of the impact of frequent changes of models.

As a final example, we refer to notions that may guide the definition of test
cases for a particular system~\cite{myers2011art}.
Considering approaches for white-box testing at
the level of functional units, coverage is an important quality
criterion~\cite{Berner2007,Tuya2016}. Behavioural quotients may be employed to
assess the coverage achieved by a test suite, where the abstraction employed in
the definition of the languages over which the quotients are computed enables
the realization of various coverage criteria, such as function-based or
branching-based coverage. As already discussed
in~\autoref{sec:related_work:software_engineering}, quotients may also be
considered as a basis to quantify test results, once they are lifted to a model
that distinguishes accepted and rejected test runs. By employing coarse-grained
abstractions that hide the internals of
units, in turn, this approach is also useful in an assessment of the results
obtained through black-box testing.

%%%%%%%%%%%%%%%%%%%%%%%%%%%%%%%%%%%%%%%%%%%%%%%%%%%%%%%%%%%%%%%%%%%%%%%%%%%%%%% 
%%%%%%%%%%%%%%%%%%%%%%%%%%%%%%%%%%%%%%%%%%%%%%%%%%%%%%%%%%%%%%%%%%%%%%%%%%%%%%%
\section{Conclusion}
\label{sec:conclusion}
%%%%%%%%%%%%%%%%%%%%%%%%%%%%%%%%%%%%%%%%%%%%%%%%%%%%%%%%%%%%%%%%%%%%%%%%%%%%%%%

This article proposed behavioural quotients as a means to relate the behaviours of dynamic systems.
A quotient takes a language measure as a parameter, which is responsible for mapping the system's behaviour onto the numerical domain for further comparisons with other behaviours.
Three example language measures are put forward in the article: one over finite, one over irreducible regular languages, and one over regular languages.
The language measure over regular languages is based on the notion of topological entropy. 
It is used to instantiate behavioral quotients into language coverage, precision, and recall measures for software engineering and process mining.
The reported evaluation results demonstrated that the proposed quotients can be computed in a reasonable time and qualitatively outperform (based on the property of the monotonicity) all the existing measures for precision in process mining.

Future work on behavioural quotients will aim at extending and improving them in several ways.
First of all, behavioural quotients can be extended to behavioural representations of dynamic systems other than their languages, \eg behavioural profiles~\cite{WeidlichMW11,PolyvyanyyADG16}, declarative models~\cite{Aalst.etal/CSRD09:DeclarativeWFsBalancing,DiCiccio.Mecella/ACMTMIS2015:DiscoveryDeclarativeControl}, and hybrid representations~\cite{Maggi.etal/BPM2014:AutomatedDiscoveryHybrid,DeSmedt.etal/DSS2015:FusionMinerProcess} in light of their underlying expressibility as finite-state automata \cite{PolyvyanyyADG16,DiCiccio.etal/IS2017:ResolvingInconsistenciesRedundanciesDeclare,DeSmedt.etal/BPI2016:ModelCheckingofMixedParadigmDiscovery}. 
Furthermore, one can propose new language measures for instantiating behavioural quotients and study interpretations and computational complexities of these measures.
Moreover, language quotients can be improved to account for multiplicity and similarity of words.
The quotients proposed in this article abstract from multiplicities of words and consider words as being distinct even if they differ only in a single symbol.
Also, the works on the automated inference of models of behaviour including data attributes, time, and probabilities in their transitions~\cite{Lorenzoli.etal/ICSE2008:AutomaticGenerationofSwBehavioralModels,Walkinshaw.etal/ESE2016:InferringExtendedFSMModelsfromSwExecutions,Narayan.etal/ACMTECS2018:MiningTimedRegularSpecificationsfromSystemTraces,Krismayer.etal/CAiSE2019:ConstraintMiningCyberPhysicalSystems,Emam.Miller/ACMTOSEM2018:ExtendedProbabilisticFSAfromSwExecutions} inspire an interesting future avenue for our research, \ie to measure precision and recall of such extended models.
Another context in which one can investigate the applicability and adaptation of our approach is that of the models of behaviour expressing distributed systems invariants~\cite{Grant.etal/ICSE2018:Inferringandassertingdistributedsysteminvariants}.
Finally, one can design new quality measures that relate arbitrary numbers of behaviours (not just behaviours of a specification and its execution log), \eg to establish a basis for comparing results of various process querying methods~\cite{PolyvyanyyOBA17}, models of behaviour that summarise traces at varying levels of abstraction~\cite{Hamou-Lhadj.Lethbridge/ICPC2006:SummarizingTracesSoftwareSystem}, and different behavioural representations~\cite{Prescher.etal/SIMPDA2014:FromDeclarativeToImperative}.

The recent observation that all the state-of-the-art precision measures in process mining fail to satisfy some basic desired properties~\cite{TaxLSFA17}, initiated a discussion on what properties should the standard quality measures, like precision, recall, and generalization, possess~\cite{Aalst18a}. 
The precision and recall defined as language quotients, such as the entropy-based measures, satisfy all the properties proposed in~\cite{TaxLSFA17,Aalst18a} (see \autoref{sec:precision_recall}).
This result is due to the fact that these measures are defined as ratios over language measures.
Consequently, they satisfy the properties of non-negativity, have null sets, and are strictly monotone (see~\cite{tao2013introduction} for details on the standard properties of measures).
Therefore, we propose to shift the focus of the discussion from the desired properties of the quality measures to the desired properties of measures over languages that are used to define them.
For example, it is interesting to study if an additional requirement of \emph{additivity} or \emph{sub-additivity} over a language measure used to instantiate precision and recall quotients leads to their useful properties.

The monotonicity property allows comparing measured values over behaviours, but not reasoning over their absolute values.
Indeed, the difference in the measured values over two languages in the containment relation has no particular meaning.
Also, it is not established which concrete precision values denote precise or imprecise models with respect to a given event log.
Future works will tackle these problems in dialogue with domain experts.

The devised behavioural quotients were tested using real-world and synthetic logs of IT systems that govern the execution of business processes and synthetic logs of software specifications.
In \autoref{sec:discussions:se:practice}, we discussed the use of behavioural quotients for software configuration management, model-based software engineering, and software testing.
However, future work will need to 
adapt to those use cases the behavioural quotients we present in this work, 
as yet unforeseen 
obstacles may arise in their immediate adoption in software engineering practices.

%%%%%%%%%%%%%%%%%%%%%%%%%%%%%%%%%%%%%%%%%%%%%%%%%%%%%%%%%%%%%%%%%%%%%%%%%%%%%%
%%%%%%%%%%%%%%%%%%%%%%%%%

\begin{acks}
Artem Polyvyanyy was partly supported by the Australian Research Council Discovery Project DP180102839.
Artem Polyvyanyy and Matthias Weidlich are grateful for the support by the Universities Australia (UA) and the German Academic Exchange Service (DAAD) as part of the Joint Research Co-operation Scheme.
The work of Claudio Di Ciccio and Jan Mendling received funding from the EU H2020 programme under the MSCA-RISE agreement 645751 ({RISE\_BPM}) and the Austrian Research Promotion Agency (FFG) grant 861213 (CitySPIN).
Claudio Di Ciccio was partly supported by the MIUR under grant ``Dipartimenti di eccellenza 2018-2022'' of the Department of Computer Science at Sapienza University of Rome.
We would like to thank Anna Kalenkova for her review of our manuscript and comments that helped to improve it.

\end{acks}

\bibliographystyle{ACM-Reference-Format}
\bibliography{bibliography.min}

\end{document}